\documentclass[acmsmall, nonacm, authorversion]{acmart}

\usepackage{version} 
\includeversion{full}
\excludeversion{podc}

 \usepackage{xcolor}
 \newcommand{\roncomment}[1]{{\color{red} Ron: #1}}
 \newcommand{\kayacomment}[1]{{\color{blue} Kaya: #1}}
 \newcommand{\joecomment}[1]{{\color{blue} Joe: #1}}

\newcommand{\always}{\boxdot}

\newcommand{\decided}{\mathit{decided}}
\newcommand{\deciding}{\mathit{deciding}}
\newcommand{\rd}{\mathit{jd}}  
\newcommand{\justdecided}{\mathtt{jd}}
\newcommand{\decision}{\mathtt{decision}}

\newcommand{\decides}{\mathtt{decides}}
\newcommand{\decide}{\mathtt{decide}}

\newcommand{\commentout}[1]{}
\newcommand{\ok}{\mathtt{ok}}

\newcommand{\init}{\mathit{init}}
\newcommand{\initf}{\mathit{init}}
\newcommand{\Agents}{\mathit{Agt}}

\newcommand{\exchange}{\mathcal{E}}
\newcommand{\rimp}{\Rightarrow}
\newcommand{\noop}{\mathtt{noop}}

\newcommand{\Nat}{\mathbb{N}} 
 
\newcommand{\future}{\Box} 
\newcommand{\powerset}[1]{\mathcal{P}(#1)}
\newcommand{\Prop}{Prop}
\newcommand{\Estates}{L_e} 
\newcommand{\I}{\mathcal{I}}

\newcommand{\R}{\mathcal{R}}

\newcommand{\failures}{\mathcal{F}} 
\newcommand{\adversary}{\alpha} 

\newcommand{\cb}[1]{C\hspace{-1pt}B_{#1}}
\newcommand{\eb}[1]{E\hspace{-1pt}B_{#1}}
\newcommand{\ck}[1]{C_{#1}}
\newcommand{\ek}[1]{E_{#1}}
\newcommand{\<}{\langle}
\renewcommand{\>}{\rangle}

\newcommand{\beln}{B^\N} %

\newcommand{\ekn}{\ek{\N}} %

\newcommand{\contb}[1]{E\hspace{-1pt}B^\always_{#1}}
\newcommand{\contcb}[1]{C^\always_{#1}}

\newcommand{\kbp}{\mathbf{P}}

\newcommand{\N}{\mathcal{N}}

\newcommand{\Time}{\mathit{time}}
\newcommand{\Timef}{\mathit{time}}
\newcommand{\basic}{\mathit{basic}}
\newcommand{\fip}{\mathit{fip}}
\newcommand{\nodecided}{\mathit{no\text{-}decided}}

\newcommand{\sendsto}{\rightarrow}  

\usepackage[ruled,vlined]{algorithm2e}
\newenvironment{program}[1][ht]
  { %
   \begin{algorithm}[#1]
  }{\end{algorithm}}

  \newcommand{\jdecided}{\mathit{jdecided}}
  \renewcommand{\citeyear}{\cite}
  \newcommand{\nonfaulty}{\mathcal{N}}
\newcommand{\faulty}{\Agents - \mathcal{N}} %
\newcommand{\faultyag}{\mathit{t\mbox{-}faulty}}
\newcommand{\Circ}{\mbox{{\small $\bigcirc$}}}
\newcommand{\cS}{\mathcal{S}}
\newcommand{\cO}{\mathcal{O}}
\newcommand{\Z}{\mathcal{Z}}
\renewcommand{\P}{\mathcal{P}}
\usepackage{enumerate}

\newcommand{\inv}{\mathit{inv}}
\newcommand{\dist}{\mathit{dist}}
\newcommand{\dec}{\mathit{decide}}
\newcommand{\dk}[1]{D_{#1}}
\renewcommand{\>}{\rangle}
\renewcommand{\C}{\mathcal{C}}
\newcommand{\cond}{\mathit{cond}}
\newcommand{\common}{\mathit{common}}
\newcommand{\opt}{\mathit{opt}}
\newcommand{\last}{\mathit{last}}
\newcommand{\latest}{\mathit{latest0}}
\newcommand{\len}{\mathit{len}}
\newcommand{\zchain}{$0$-chain\xspace}  
\newcommand{\zchains}{$0$-chains\xspace} 
\newcommand{\dimp}{\Leftrightarrow}
\newcommand{\nxt}{\Circ} 
\newcommand{\prev}{\ominus} 
\newcommand{\lamport}[1]{\rightarrow_{#1}}

\begin{document}

\title{Optimal Eventual Byzantine Agreement Protocols with Omission Failures} 

\author{Kaya Alpturer}
\email{ki78@cornell.edu}
\affiliation{%
  \institution{Cornell University}
  \city{Ithaca, NY}         %
  \country{USA}   %
}

\author{Joseph Y. Halpern}
\email{halpern@cs.cornell.edu}
\affiliation{%
  \institution{Cornell University}
  \city{Ithaca, NY}         %
  \country{USA}   %
}

\author{Ron van der Meyden}
\email{R.VanderMeyden@unsw.edu.au}
\affiliation{%
  \institution{UNSW Sydney}
  \city{Sydney, NSW}         %
  \country{Australia}   %
}


\begin{abstract}
  Work on \emph{optimal} protocols for \emph{Eventual Byzantine Agreement} 
  (EBA)---protocols that, in a precise sense, decide as soon as possible 
  in every run and guarantee that all nonfaulty agents decide on the
  same value---has focused on \emph{full-information protocols} (FIPs), where 
  agents repeatedly send messages that completely describe their past 
  observations to every other agent. While it can be shown that, without 
  loss of generality, we can take an optimal protocol to be an FIP,
  full information  
  exchange is impractical to implement for many applications due to the 
  required message size.
  We separate protocols into two parts, the \emph{information-exchange
  protocol} and
  the \emph{action protocol}, so as to be able to examine the effects of
  more limited information exchange.   We then define a
  notion of optimality with respect to an information-exchange protocol.
  Roughly speaking, an action protocol $P$ is 
 optimal with respect to an information-exchange
  protocol $\exchange$ if, with $P$,
  agents decide as soon as possible among action protocols that exchange
  information according to $\exchange$.
  We present a knowledge-based EBA program for omission failures 
all of whose implementations are guaranteed to be correct and are
optimal if the information exchange satisfies a certain safety condition.
  We then construct concrete programs that implement this
  knowledge-based program in two settings of interest that are shown
  to satisfy the safety condition.
  Finally, we show that a small modification of our program results in
    an FIP that 
  is both optimal and efficiently implementable, 
    settling an open problem posed by Halpern, Moses, and Waarts 
    (SIAM J. Comput., 2001).
\end{abstract}

\maketitle

\section{Introduction} \label{sec:intro}

Logics of knowledge (epistemic logics) have been shown to provide
useful abstractions for  reasoning about distributed systems 
\cite{FHMV,HM90},
enabling a focus on the information that needs to be attained in order
to perform certain actions, independent of    
how that information is encoded in the states of the system. The
approach has, in particular, been used  
fruitfully to analyze agreement protocols, where agents are
required to make consistent decisions on some value, 
based on their initial preferences 
\cite{CGM14,DM,HMW,MT}.

One particular focus of this work  has been to develop
protocols that make optimal
use of information. In the context of protocols for agreement, this 
has meant protocols  
that decide as soon as possible.
We say that a protocol $P_1$ \emph{dominates} a protocol $P_2$ if,
roughly speaking, for each possible pattern of failures and inputs,
$P_1$ decides at least as soon as $P_2$.  An \emph{optimum protocol}
is one that dominates every other protocol; an \emph{optimal protocol}
is one that is not dominated by any other protocol.
Optimum knowledge-based programs (i.e., programs with explicit tests
for knowledge) have been provided  for simultaneous
Byzantine agreement 
(SBA) with crash failures \cite{CGM14,CastanedaMRR17,DM} and omission
failures%
\footnote{Recall that with \emph{crash failures}, a faulty process
behaves according to 
the protocol, except that it might crash at some round (possibly
after sending some messages), after which it
sends no messages; with sending-omission failures, a faulty process
may omit to send an arbitrary set of messages in any given round, but
otherwise behaves according to the protocol.}
\cite{MT},
and for \emph{consistent}
SBA with omission failures \cite{NeigTuttle} (where all agents that
decide on an action must perform the same action, not just the
nonfaulty agents).
Moreover, polynomial-time implementations of these
programs were provided.
For the problem of eventual Byzantine agreement (EBA), it is
well-known that there are no optimum protocols \cite{MT}, although
there are optimal protocols.  Optimal knowledge-based programs have
been
provided for EBA in the case of crash failures \cite{CGM14} and
sending-omission failures \cite{HMW}.  While a polynomial-time
implementation of the knowledge-based program was provided for the
case of crash failures \cite{CGM14}, none was provided in the case of
omission falures.  Indeed, the problem of finding a polynomial-time
optimal algorithm for EBA in the presence of omission failures  has
been open for over 20 years.  Among other things, we solve this
problem here.

The work on optimal (and optimum) protocols has focused on
\emph{full-information protocols} (FIPs), ones where 
each agent repeatedly sends all other agents  
its complete state, containing its initial state and all messages that
it has received up to the present time. 
As far as optimal protocols go, there is no loss of generality in
considering only FIPs. As is well known 
\cite{Coan},
any protocol can be
simulated by an FIP, so for any protocol, there is an FIP that decides
at least as soon.
However, while FIPs do provide optimality, they are expensive to implement due
to their space requirements and the length of messages sent,
their analysis may 
be
complex, and in some failure
environments, they may require that intractable properties be  
computed at each step of the protocol to attain optimality. They are
therefore not 
necessarily practical. The present paper is part of 
a
program of research in which we seek to
overcome these difficulties with the  
full-information paradigm by considering 
protocols in which less than full information is exchanged
between the agents.  

Our goal in this paper is to examine the effects of more limited
information exchange.
In order to do so, 
we separate protocols into two parts, 
the \emph{information-exchange protocol}, which specifies what information
agents maintain in their local states and what message they exchange
at each step, and 
the \emph{action protocol}, which, in the case of agreement protocols,
specifies what decision 
agents make.   We then define a
  notion of optimality with respect to an information-exchange protocol.
  Roughly speaking, an action protocol $P$ is 
    optimal for a particular specification with respect to an
    information-exchange 
    protocol $\exchange$ if
    agents decide as soon as possible with $P$  as they do with any
    other protocol that satisfies the specification and exchanges
    information according to $\exchange$.
  A full-information protocol is a special case of an
information-exchange protocol, but we allow far more limited
protocols, where agents keep track of less information and send less
information in their messages.
  We focus in this paper on optimal protocols for binary EBA, where
  agents have only one of two possible initial preferences, and we assume
sending-omission failures.
For EBA, there is   
a group of agents, each with 
an initial preference
of 0 or 1. The nonfaulty agents
are required to
reach agreement on a value, but different agents may decide at
different times.   
The solution is required to be non-trivial in the sense that if all
initial preferences
are $x\in \{0,1\}$  
then a nonfaulty agent must decide $x$.

Since with EBA, agents do not have to decide simultaneously, the
literature has typically considered optimal protocols that are biased
towards 0, in that an agent decides 0 as 
soon as it learns that some agent started with
an initial preference of 0, and protocols that are baised towards 1. In the
case of crash failures, it is known that there are optimal EBA
protocols that are biased towards 0 (resp., 1) in this sense
\cite{CGM14}.  But it 
is easy to see that there cannot be an optimal EBA protocol that is
biased towards 0 (or 1) in the case of omissions failures.  Consider a
setting where there are at least three agents, and a run $r$ where 
exactly one of these agents is faulty,
say agent 1, all the remaining agents have an initial preference
of 1, and the faulty agent does 
not send any messages.  In run $r$, the nonfaulty agents must
eventually decide 1, 
because the faulty agent's initial preference may be 1, and in that
case, EBA requires a decision of 1.  Suppose that the first nonfaulty agent 
to decide in run $r$ does so at round $n$, and without loss of
generality, that 2 is a nonfaulty
agent that decides 1 at round $n$.  Now consider a run $r'$ 
where agent 1 has an initial preference of 0, agent 1 is faulty, all
the remaining agents are nonfaulty and have an initial preference of
1, agent 1 does not send
any messages up to round 
$n-1$, and in round $n-1$ sends exactly one
message, which is sent to agent 3, and says (truthfully) that
agent 1's initial preference was 0.  Since agent 2 cannot distinguish
$r$ from $r'$, agent 2 must decide 1 in round $n$ of $r'$.  Since
agent 3 does not decide in the first $n-1$ rounds of $r$, it also does
not decide in the first $n-1$ rounds of $r'$.  Since it learns that
some agent has an initial preference of 0 in round 
$n-1$,
agent 3 decides
0 in round $n$ of $r'$.  Thus, two nonfaulty agents decide on
different values in $r'$, so EBA is not achieved.  (Note a run like
$r'$ is inconsistent with crash failures; it really requires omission
failures.) 

\commentout{
Optimal EBA protocols can be biased in one of two ways: they can
prefer to 
decide 0 as soon as possible, intuitively, by having an agent
decide 0 as soon as it learns
that some agent 
has initial preference
0. Alternately, the protocol could have a
similar bias to deciding 1.
In protocols where agents decide 0 as soon as they hear about a 0,
since no protocol can decide 0 earlier, we can reduce the problem of 
constructing an optimal protocol to finding the best rule to decide 1 
without breaking correctness. 
While this approach does work in the presence of crash failures
\cite{CGM14}, due to the difficulty mentioned above, we usually get 
non-terminating runs in the presence of omission failures. 
The first contribution of this paper is to show that it is impossible 
to get an optimal protocol that terminates when agents decide 0 as soon as 
possible. 
We can get around this problem
by having agents decide on 0 only if they hear about a 0 in a chain of
messages, starting with an agent that had 0 as an initial preference,
and passed from one agent to the next.}

To deal with this issue, in a 0-biased protocol, rather than requiring
that an agent decide 0 as soon as it hears about a 0, we require only
that an agent decides 0 only if it hears about 0
via a chain of agents (where the first agent in the chain has an
initial preference of 0, and in round $k$, the $k$th agent in the
chain decides 0 and tells the $(k+1)$st agent about this).  Note that
in the case of crash 
failures, an agent can hear about a 0 only via such a chain.
We then provide a knowledge-based action protocol $\kbp^0$ based on this
(well-known) idea that we show is correct (in the sense that all of its 
implementations satisfy the EBA specification, no matter what
information-exchange protocol is used) and is optimal in
contexts that satisfy a certain safety condition.  
We then consider two information-exchange protocols where agents exchange
relatively short (and relatively few) messages, and show that
they satisfy the safety condition.  Finally, we provide concrete
polynomial-time
action protocols that implement $\kbp^0$ with respect to these two
information-exchange 
protocols.

The knowledge-based program $\kbp^0$ is not optimal in the case of
full-information contexts, but, as we show, a small modification
$\kbp^1$ of $\kbp^0$ is optimal.  Moreover, $\kbp^0$ and $\kbp^1$ are
equivalent in the two limited-information contexts that we considered,
so $\kbp^1$ is also optimal in these contexts.
Roughly speaking, $\kbp^1$ allows agents to decide if they get
common knowledge of who the nonfaulty agents are. Such common
knowledge can't be obtained in the limited-exchange contexts that we
consider, but it can be obtained if agents use a full-information
protocol and keep track of everything they have learned.
The knowledge-based program $\kbp^1$ does not involve the (rather
complicated) \emph{continual common knowledge operator} used by Halpern,
Moses, and Waarts; rather, it uses more standard knowledge and common
knowledge operators. This allows us to provide a polynomial-time
implementation of it.  

\commentout{
Interestingly, the FIP does \emph{not} satisfy the safety condition.  
We nevertheless conjecture that $\kbp^0$ is in fact optimal with respect to the
FIP (the safety condition is only a sufficient condition to ensure
optimality).  We provide a weakening of the safety 
condition closely related to conditions considered by Halpern, Moses,
and Waarts \citeyear{HMW} and show that $\kbp^0$ is optimal with respect
to the FIP if this weaker condition holds.  
} %

\commentout{
  This protocol can be simulated by an FIP,
simply by having the FIP first determine 
from its state what would be the state of the limited
information-exchange protocol.   
The obvious question is whether an FIP can decide earlier.  
As we show, it can.   Indeed, we
provide a polynomial-time optimal FIP for EBA that dominates our
initial protocol, thereby solving a problem left open by Halpern,
Moses, and Waarts \citeyear{HMW}.
In runs where the limited-exchange protocol decides 0, the FIP decides
0 as well, and decides no earlier.  But the FIP can decide earlier in
runs where the limited-exchange protocol decides 1 (and may decide 0
in some of these runs).
We also provide a knowledge-based characterization of the conditions
under which the optimal FIP decides.  The characterization does not
involve the (rather complicated) continual common knowledge operator
used by Halpern, Moses, and Waarts \citeyear{HMW}; rather, it uses more
standard knowledge and common knowledge operators.}

\commentout{
As we suggested above, there are a number of advantages to considering
implementations of $\kbp^0$ using information-exchange protocols with
limited information flow.  
Of particular interest to us is that it
opens the door to applying epistemic synthesis techniques that allow
the automated derivation of protocols from a knowledge-based program
in the context of a given information-flow model
\cite{HuangM13,HuangM14}.
We hope to explore this direction in future work.
}

The rest of this paper is organized as follows. Section
\ref{sec:knowledge} 
reviews the epistemic framework of \cite{FHMV}.
In Section \ref{sec:faults}, we introduce the 
separation of information-exchange and action protocols 
along with a representation of the failure model. 
Section \ref{sec:kbps} defines knowledge-based programs.
The specification of EBA and a formal definition of optimality with respect to 
an information-exchange protocol are given in Section \ref{sec:spec}.
In Section~\ref{sec:opt-lim-ie}, we define the knowledge-based
program $\kbp^0$, show that it satisfies EBA, define the safety
condition that suffices for $\kbp^0$ to be optimal, describe two
natural limited information-exchange protocols that satisfy the safety
condition, and provide concrete action protocols that implement
$\kbp^0$ with respect to these two information-exchange protocols.
In Section~\ref{sec:fip-optimal}, we define $\kbp^1$, a modification
of $\kbp^0$ that is optimal for the full-information-exchange
protocol, and provide a polynomial-time 
implementation
of it.
We conclude with a discussion of the cost of limited information 
exchange in Section \ref{sec:cost}.
\begin{full}
Proofs for all the results in the paper can be found in the appendix.
\end{full}
\begin{podc} 
We omit most proofs here due to lack of space; they can be found in
\cite{FIXME-fullversion}.
\end{podc} 

\commentout{

  \section{Related Work} 

  The state of the art with respect to knowledge-based analysis of consensus protocols 
  consists of
  \begin{itemize} 
  \item a simple characterization of SBA using a knowledge-based program based on common knowledge, 
  that is correct across a range of environments (information-exchange protocols) 
  \item implementations for a number of FIP environments, developed by hand. 
  \item a complex, continuous common knowledge characterization of optimality for EBA in FIP using EBA
  \item a knowledge-based process for transforming an EBA protocol into an optimal EBA protocol wrt FIP. 
  \item further notions of common knowledge extending continuous common knowledge for 
  uniform and terminating versions of EBA with respect to  FIP (Neiger \& Tuttle, Neiger \& Bazzi) 
  \item A simpler knowledge-based program for EBA for the crash failure case, 
  and an implementation in the FIP crash failure  case, developed by hand. 
  \end{itemize}
}

\section{Semantic Model}\label{sec:knowledge}

We assume that a set $\Agents$ of agents  communicate using a
message-passing network, which may be subject to various types of failures.  
To model such systems semantically, we use the standard
runs-and-systems model \cite{FHMV,HF}, which we briefly review.

\emph{Interpreted systems} \cite{FHMV} model multi-agent scenarios
in which some number $n$ of agents change their states over time.   
An interpreted system is a  pair $\I = (\R, \pi)$, where $\R$ is a set of runs, describing how
the system evolves over time,  and  
$\pi: \R\times \Nat \rightarrow \powerset{\Prop}$ 
is an interpretation function that indicates which atomic propositions are 
true at each {\em point} of the system, where a point is a pair $(r,m)$
consisting of a 
run $r\in \R$ 
and time $m\in \Nat$.
The set $\R$ is called a \emph{system}. 
Formally, a run $r\in \R$ is a function
$r:\Nat \rightarrow \Estates\times \Pi_{i\in \Agents} L_i$, 
where $\Estates$ is the set of possible \emph{local states} of the
environment in which the 
agents operate, and each $L_i$ is  the set of possible {\em
local states} of agent $i$.
The elements of $\Estates\times \Pi_{i\in \Agents} L_i$ are called
\emph{global states}.
Given a run $r$, agent $i$, and time $m$, we write $r_i(m)$ for 
the $(i+1)$st component of $r(m)$, which is the local state of agent
$i$ in the global state $r(m)$, 
and $r_e(m)$ for the first component of $r(m)$, which is 
the local state of the environment.

To reason about the knowledge of agents in interpreted systems, we 
use a standard language for reasoning about
knowledge and time.  We start with a set $\Phi$ of primitive
propositions, and close off under $\land$, $\neg$, the epistemic
operators $K_i$ for $i = 1,\ldots, n$ (one for each agent) and
$C_\cS$ (common knowledge among the agents in an \emph{indexical} set $S$; see
below) 
and the temporal operators $\future$, $\always$, $\Circ$, and $\ominus$.
The formula $K_i \phi$ says that agent $i$ knows that formula $\phi$ holds, 
$\future \phi$ says that $\phi$ holds at all times in the future,
$\always \phi$ says that $\phi$ holds at all times,
$\Circ \phi$ says that $\phi$ holds at the next time,
and $\ominus \phi$ says that $\phi$ holds at the previous time. 

The semantics of the logic is given by a relation $\I,(r,m)\models \phi$,
where $\I$ is an interpreted system, $(r,m)$ is a point of $\I$, and
$\phi$ is a formula.  
For formulas not of the form $C_\cS\phi$, the relation $\models$ is 
defined inductively as
follows (we omit the obvious cases for the propositional operators):
\begin{itemize}
\item $\I,(r,m) \models p$ if   $p\in \pi(r,m)$,

\item
$\I,(r,m)\models K_{i} \phi$ if $\I,(r',m') \models \phi$ for all points $(r',m')$ of $\I$ such that $r_i(m)  = r'_i(m')$,

\item $\I,(r,m)\models \future\phi$ if   $\I,(r,m') \models \phi$ for all $m'\geq m$,

\item $\I,(r,m)\models \always\phi$ if   $\I,(r,m') \models \phi$ for all $m'\in \Nat$,

    \item $\I,(r,m) \models \Circ\phi$ if $\I,(r,m+1) \models \phi$
    \item $\I,(r,m) \models \ominus\phi$ if $m > 0$ and $\I,(r,m-1)
      \models \phi$. 
        \end{itemize}
\commentout{

  A formula $\phi$ is said to be \emph{valid} in an 
interpreted system $\I$, 
written $\I \models \phi$, 
if $\I,(r,m) \models \phi$ for all points $(r,m)$ of $\I$.  
A formula $\phi$ is \emph{local} to an agent $i$ in an interpreted system $\I$ if for all 
points $(r,m)$ and $(r',m')$ of $\I$ with $r_i(m) = r'_i(m')$ we have 
$\I,(r,m) \models \phi$ iff $\I,(r',m') \models \phi$. In particular, this
means that $\I \models \phi \rimp K_i \phi$ and  $\I \models \neg \phi \rimp K_i\neg  \phi$.
} 

The intuition for the definition of the knowledge operator $K_i\phi$
is that $r'_i(m) = r_i(m)$ says that agent $i$ considers it possible,
when  
in the actual situation $(r,m)$, that it is in situation $(r',m')$,
since it has the same local state there.  
An agent then \emph{knows} $\phi$ if $\phi$ is true in all the
situations that the agent considers to be possible.

We can now define the modal operator $C_\cS$.
Intuitively, $C_\cS \phi$ is true at a point $(r,m)$ if $\phi$ is common
knowledge among the agents in $S$; each of the agents in $\cS$ knows
that each of the agents in $S$ knows \ldots that $\phi$ is true.  
The fact that $\cS$ is an indexical set means that its membership can
depend on the point; that is, 
semantically, $\cS(r,m)$ is a set of agents for each point $(r,m)$.  
We define $\ek{\cS}\phi$ (everyone in $\cS$ knows $\phi$) as an
abbreviation for $\bigwedge_{i\in \cS} K_i \phi$.  That is,
$$\mbox{$\I,(r,m) \models \ek{\cS}\phi$ if, for all $j \in \cS(r,m)$, we have
$\I,(r,m) \models K_j \phi$.}$$  
Taking $\ek{\cS}^1\phi$ to be an abbreviation of $\ek{\cS}\phi$,
 and
$\ek{\cS}^{m+1}\phi$ to be an abbreviation of $\ek{\cS}(\ek{\cS}^m\phi)$, we
define 
$$\mbox{$\I,(r,m) \models C_\cS\phi$ if, for all $m \ge 1$, 
$\I,(r,m) \models \ek{\cS}^m \phi$.}$$  

As usual, we say that $\phi$ is 
\emph{valid in $\I$}, 
and write $\I
\models \phi$, if $\I,(r,m) \models \phi$ for all points $(r,m)$ in
$\I$.

\section{Communication and Failure Models} \label{sec:faults}

We now specialize the general model from the previous section to
represent an omissions-failure model.  
In our representation, we separate the \emph{information-exchange protocol},
which characterizes 
the information maintained by agents in their local states, and which messages are sent and when, from the
\emph{action protocol}, which characterizes the rules for performing actions
other than sending messages.  (In our case, these actions will be decisions.)
In the literature, the information-exchange protocol has often been
the \emph{full-information protocol}, in which at each step each agent
sends all other agents 
a complete description of everything it has learned up to that
time. However,  we will be interested in protocols  
in which less information is exchanged, so it helps to separate out
the information-exchange protocol as a parameter of the interpreted
systems we construct.  
A further parameter is the failure model $\failures$, defined below.  

We assume that information is exchanged by sending messages.
Our focus will be on \emph{synchronous} message passing, in which agents operate in a sequence of synchronized rounds. 
In each round, each agent performs some actions, sends a set of messages to the other agents, 
receives some of the messages from the other agents that were sent  in the same round, and updates its state depending on these events. 
The information-exchange protocol describes the possible initial states of the agent (which may include information such as the agent's preference for the 
outcome of the consensus decision to be made), how it chooses the messages to send at each time, and how it updates its state in response to receiving messages.  

We assume that each agent has a set $A_i$ of actions that it
can perform.  
In our applications, $A_i = \{\decide_i(x)~|~ x\in
\{0,1\}\}\cup \{\noop\}$, but in general, $A_i$ can be arbitrary.
Formally, an information-exchange protocol $\exchange$ for agents
$\Agents = \{1, \ldots,n\}$ is given by a tuple  
$\langle \exchange_1, \ldots , \exchange_n\rangle$ consisting of a 
local information-exchange protocol $\exchange_i$ for each each agent $i$. 
Each  local information-exchange protocol 
$\exchange_i$ is a tuple $\langle L_i, I_i , A_i, M_i, \mu_i,
\delta_i\rangle$,   
where
\begin{itemize} 
\item $L_i$ is a set of local states.
\commentout{
\item $\init_i : L_i \rightarrow \{0,1\}$ maps each local state to a
  vote for the decision to be made,
\item $d_i: L_i \rightarrow \{0,1,\bot\}$ maps each local state to
  either a decision or $\bot$ (in case no decision has yet been made).  
\item $\Time_i : L_i \rightarrow \Nat$ maps each local state to a time. 
}
\item $I_i \subseteq L_i$ is a set of  initial states.

\newcommand{\Messages}{\mathcal{M}} 
\item $M_i$ is a set of messages that can be sent by agent $i$. 
  \item $\mu_i : L_i \times A_i \rightarrow \Pi_{j\in \Agents} ( M_i
  \cup \{\bot\})$ is a function mapping a local state $s$  
and an action $a$ to  the messages to be sent in the current round,
one to each agent $j$.   
Intuitively, $\mu_i(s,a) = \sigma$ 
means that when action $a$ is performed in state $s$, the 
information-exchange protocol  transmits message $\sigma_j$ to each agent $j$.  
Here $\sigma_j= \bot$ represents that no message is sent by $i$ to $j$.
Let $\mu_{ij}(s,a)$ denote the message that $i$ sends to $j$ in this tuple.

\item 
  $\delta_i: L_i \times A_i \times \Pi_{j\in \Agents} (M_j \cup \{\bot\})
  \rightarrow L_i$  
is a function that 
updates the local state, given 
an action and 
a tuple $(m_1, \ldots, m_n)$ 
of messages $m_j\in M_j \cup \{\bot\}$ (where $m_j = \bot$ if $i$
receives no message from $j$).
\end{itemize}

The failure model describes what failures  can occur. 
Typically a failure model comes with a parameter $t$ that indicates
the maximum number agents that may be faulty.  
A \emph{failure pattern},  or \emph{adversary}, defines
the failures that actually occur in   
a particular run consistent with the failure model.
\commentout{
Each failure model is defined as a set of adversaries.  
An omissions-failure pattern indicates which messages that are
intended to be sent are not actually sent (in case the sender is
faulty), or which messages that were sent are not delivered (in case
it is the  
recipient that is faulty.) 

}
Formally, a failure pattern 
$\adversary$ is a pair $(\nonfaulty,F)$, where $\nonfaulty\subseteq
\Agents$ and 
$F$ is a mapping $F: \Nat \times \Agents \times \Agents \rightarrow \{0,1\}$. 
Here $\nonfaulty$ is the set of nonfaulty agents,
and  
$F(m,i,j)$ describes whether the message sent
from agent $i$ to agent $j$ in 
round $m+1$ is 
delivered.
(If it is not delivered, we assume that the message $\bot$ is
delivered instead.)
A failure model is a set of failure patterns.
The sending-omissions model $SO(t)$ for agents $\Agents$ is the set of
all failure patterns 
$(\nonfaulty,F)$ such that $|\faulty| \leq t$, so that there are at
most $t$ faulty agents,  and for all $m\in \Nat $
and $j \in \Agents$,  
if $F(m,i,j) = 0$ then $i \in \faulty$.
The crash-failures model is the special case where if $F(m,i,j) = 0$
then $F(m',i,j') = 0$ for all $m' > m$ and agents $j'$.

\commentout{
Formally, a failure pattern $F$ is a mapping $F: \Nat \times \Agents \times \Agents \rightarrow \{0,1\}$. 
Here the Boolean value $F(m,i,j)$ represents whether the message sent
from agent $i$ to agent $j$ at time $m$ is delivered.
(If it is not delivered, we assume that the message $\bot$ is
delivered instead.)
A failure model $\failures$ is a set of failure patterns.
The sending omissions model $SO(t)$ for agents $\Agents$ is the set of all failure patterns
$F$ such that there exist at most $t$ agents $i$ such that $F(m,i,j) = 0$ for some $m\in \Nat $ and $j \in \Agents$.

For the sending omissions model, the indexical set of nonfaulty agents
at a point $(r,m)$ with failure  pattern $F$ is 
defined as the set $\N(r,m)$ containing all agents $i$ that do not have a sending fault at any time, that is, for 
which $F(m',i,j) = 1$ for all $m'\in \Nat$ and $j\in \Agents$. Since this definition does not depend on the time, 
we may also write $\N(r)$ for this set. 
}

An \emph{action protocol} $P$ 
for an information-exchange protocol $\exchange$, 
is a tuple $(P_1, \ldots,P_n)$ containing, for each agent $i=1\ldots n$, 
a \emph{local action protocol} $P_i : L_i \rightarrow A_i$ mapping
the local states $L_i$ for agent $i$ in $\exchange$ to actions in $A_i$.

To connect these definitions to the semantic model of
Section~\ref{sec:knowledge}, we describe how an information-exchange
protocol $\exchange$, a failure model $\failures$, and an
action protocol $P$ determine 
a system $\R_{\exchange,\failures,P}$.  
In this system, the set $\Estates$ of possible local states of the
environment consists of the possible failure patterns.
The  local states $L_i$ of the agent are the states of the
information-exchange protocol for agent $i$.  
An \emph{initial state} $(s_e, s_1, \ldots , s_n)$ 
is a global 
state,
where $s_i\in I_i$ is an initial state of 
each 
agent $i$'s
information-exchange protocol.  
For each initial state, a run $r$ with that initial state is uniquely determined by the information-exchange protocol $\exchange$, 
the failure model $\failures$, and the action protocol $P$. In this run, 
the protocol $\exchange$, the failure pattern $\adversary$, and $P$
determine at each step, in order, 
what 
actions are taken, what messages are sent, and what messages are received.
Each agent updates its local state depending on the actions taken and the messages received in the round.  
Formally, the global state $r(k+1) = (s_e',s_1', \ldots ,s_n')$ at time $k+1$ is determined
from the global state $r(k) = (s_e,s_1, \ldots ,s_n)$ at time $k$ as follows: 
\begin{itemize} 
  \item $s_e' = s_e$ (so the failure pattern remains unchanged
         throughout the run).
\item For each 
pair of agents 
$i$ and $j$, let $m_{i,j}$ be the message that agent $i$ sends to $j$, 
given that it performs action $P_i(s_i)$ in state $s_i$, that is,  $m_{i,j} = \mu_i(s_i,P_i(s_i))(j)$. 
\item For each 
pair of agents
$i$ and $j$, let $m'_{i,j}$ be the result of applying
the failure pattern to the messages sent.
  Specifically,
suppose that $s_e = (\nonfaulty,F)$. If 
$F(k,i,j) =0$ then $m'_{i,j} = \bot$ and if $F(k,i,j)
=1$ then $m'_{i,j} = m_{i,j}$.  
\item Finally, for each agent $i$, 
$s'_i = \delta_i(s_i,P_i(s_i),(m'_{1,i}, \ldots,m'_{n,i}))$. 
\end{itemize}

The system $\R_{\exchange,\failures,P}$ consists of all 
runs
generated from some
initial state.
\commentout{
For the interpretation function $\pi$, we work with  atomic propositions $p$, for $x \in \{0,1\}$, 
which hold at a point $(r,m)$ with global state $r(m) = \langle s_e,s_1, \ldots, s_n\rangle$, 
(that is, $p\in \pi(r,m)$) 
as follows:
\begin{itemize} 
\item $\decided_i(x)$ holds  if $d_i(s_i) = x$, 
\item $\decides_i(x)$ holds if $P_i(s_i) = x$, 
\item $\init_i= x$ holds  if $\init_i(s_i) = x$,
\item $i \in \N$ holds if $i$ is nonfaulty in the environment state $s_e = F$. 
For $SO(t)$, this means that there do  not exist $j\in \Agents$ and $m\in \Nat$ such that $F(m,i,j)=0$. 
\end{itemize} 
We also use abbreviations $\decided_i$ for $\decided_i(0)  \lor \decided_i(1)$, 
$\decided(x)$ for $\bigvee_{i\in \Agents} \decided_i(x)$, 
and  
$\exists x$ (with $x \in \{0,1\}$) for $\bigvee_{i\in \Agents} \init_i = x$.
}

\section{Knowledge-Based Programs} \label{sec:kbps}

Knowledge-based
programs specify how an agent's actions are determined, given what the agent knows. 
As defined by Fagin et al. \cite{FHMV}, these programs are interpreted relative
to an \emph{interpreted context} that  
defines the global states, how they are updated as a result of actions, and an interpretation of 
atomic propositions. 
In our setting, we can take the interpreted context to be a tuple
$(\exchange,\failures,\pi)$ consisting of an  
information-exchange protocol $\exchange$,
a failure model $\failures$, and an
interpretation $\pi$ of atomic propositions in the set of all runs over global states constructed from $\exchange$
and $\failures$.  

For our purposes, it is convenient to take
knowledge-based programs to 
have the form $\kbp= (\kbp_1, \ldots,\kbp_n)$, 
where for each agent $i$, the local knowledge-based program $\kbp_i$ is in the 
language with grammar 
$$ \kbp_i ::= a_i ~| ~\text{if}~\phi_i~\text{then}~
\kbp_i~\text{else}~\kbp_i,$$ 
where $a_i$ denotes actions in the set $A_i$ of actions of agent $i$,
and  $\phi_i$ is a Boolean 
combination
of formulas of the form
$K_i\psi$. That is, the tests in agent $i$'s  
local knowledge-based program concern agent $i$'s knowledge.
Note that the truth of such a formula $\phi_i$ at a point $(r,m)$ in
an  interpreted system $\I$ depends only on agent $i$'s local state 
at that point. That is, for points $(r,m), (r'm')$ with $r_i(m) = r'_i(m')$, we have $\I,(r,m) \models \phi_i$ iff 
$\I,(r',m') \models \phi_i$. Given a local state $s$ of agent $i$,
we may therefore write $\I,s\models \phi_i$ 
to express that $\I,(r,m) \models \phi_i$ for all points $(r,m)$ of $\I$ with $r_i(m) = s$. 

To interpret a knowledge-based program semantically, we first define
how a knowledge-based program $\kbp= (\kbp_1, \ldots,\kbp_n)$ 
determines a concrete action protocol $\kbp^\I$ given an
interpreted system $\I$.  
For each agent $i$ and local state $s$ of $i$ in $\I$, we define $\kbp^\I_i(s)$ to be the action resulting
from executing the program $\kbp_i$ with its tests interpreted at
local state $s$ in $\I$.  
Formally,
we define $\kbp^\I$ by induction on the structure of $\kbp$, taking
$(a_i)^\I(s) = a_i$, and for $\kbp_i = $``$\text{if}~\phi_i~\text{then}~ \mathbf{Q}_i~\text{else}~\mathbf{R}_i$'',
we define $\kbp^I_i(s) = \mathbf{Q}_i^\I(s)$ if $\I,s\models \phi_i$, and 
$\kbp^I_i(s) = \mathbf{R}_i^\I(s)$ otherwise.

Given a knowledge-based program $\kbp= (\kbp_1, \ldots,\kbp_n)$ and a
concrete action protocol $P = (P_1, \ldots,P_n)$ 
for an information-exchange protocol $\exchange$, 
we say  that $P$ \emph{implements} $\kbp$ in the context
$\gamma = (\exchange,\failures,\pi)$  
if, for 
$\I 
= (\R_{\exchange,\failures,P},\pi)$, we have $P_i (s)=
\kbp^\I_i(s)$ for all agents $i = 1, \ldots, n$ and  
local states $s$ of agent $i$ that arise in $\I$.

\section{Eventual Byzantine agreement} \label{sec:spec}

We briefly review the specification of the eventual Byzantine agreement
problem that we consider in this paper.
The specification assumes that
each agent starts with an independently selected value $\init_i \in
\{0,1\}$. The actions in $A_i$ have the form $\decide_i(v)$, where $v
\in \{0,1\}$, 
as well as a ``do-nothing'' action $\noop$.
\commentout{
The formula $\decided_i(v)$ is true if $i$ has
performed the action $\decided_i(v)$ in the past;
We take $\decided_i$ to be an abbreviation for $\decided_i(0)  \lor
\decided_i(1)$. The formula
$\exists v$ is true (for $v \in \{0,1\}$) if $\init_j=v$ for some
agent $v$.)  (We asssume that the truth of $\decided_i(v)$ is 
encoded in the environment's state, and that each agent $i$'s initial 
value
is encoded $i$'a local state, so the truth of formulas can be
determined from the global state.) 
}%
We seek protocols (i.e.,
an information-exchange protocol and an action protocol) such
that every run satisfies the following  
four properties:
\begin{itemize} \item {\bf Unique Decision:} If agent $i$ performs
  an action $\decide_i(v)$ (for some $v$), then it does not later
  perform $\decide_i(1-v)$.
\item {\bf Agreement:}  
  If agents $i$ and $j$ are both nonfaulty,
   $i$ performs $\decide_i(v)$, and $j$ performs $\decide_j(v')$,
   then $v=v'$.  
\item {\bf Validity:} If  
  a nonfaulty agent
  $i$ performs $\decide_i(v)$ 
then 
$\init_j = v$ for some agent $j$.
\commentout{ 
In the case of benign failure models, such as crash and omissions failures, it may make sense to strengthen the definition of these safety properties to the following:
\begin{itemize} 
\item[] {\bf Agreement:}  If  $\decided_i(v)$ and $\decided_j(v')$ then $v=v'$. 
\item[] {\bf Validity:} If  $\decided_i(v)$ then $\ok(v)$.
\end{itemize} 
Intuitively, these properties require that even faulty nodes should make correct decisions, if they do decide. For example, 
{\bf Agreement} and {\bf Validity} state that  a node that is
faulty because it crashes, but makes a decision before it does so,  
are excused from making a correct decision that agrees with the decisions of the correct nodes. 
By contrast, {\bf Agreement} and {\bf Validity} hold faulty nodes to the same standard as nonfaulty nodes, when they do make a decision. 
This allows the system to remain in a safe state should the faulty node be repaired and rejoin the system.  

A stronger form of the agreement property requires that the nonfaulty nodes make their decisions simultaneously: 
\begin{itemize} 
\item {\bf Simultaneous-Agreement:} If  $\decides_i(v)$ then, at the same time, $\decides_j(v)$ for all $j\in \N$.
\item {\bf Simultaneous-Agreement(\N):} If $i\in \N$ and $\decides_i(v)$ then, at the same time, $\decides_j(v)$ for all $j\in \N$.
\end{itemize} 

This can be important in the context of real-time systems, where decisions are associated with actions that must be synchronized in the real world (e.g., controlling actuators in avionic or robotic systems).  
}%

\item {\bf Termination:} For all nonfaulty agents $i$,
    eventually  
$i$ performs $\decide_i(v)$ for some value $v\in \{0,1\}$. 
\end{itemize}  

To relate this specification to our formal
model, we define  
an \emph{EBA context} to be a tuple $(\exchange,\failures,\pi)$ 
consisting of  an information-exchange protocol $\exchange$, a failure model 
$\failures$, and an interpretation $\pi$ of a set $\Prop$ of
propositions,
such that the following conditions hold:
\begin{itemize}
  \item The local states in $\exchange_i$ have the form
    $\<\Time_i,\init_i, 
    \decided_i,\rd_i,\ldots\>$, 
    where $\Time_i \in \Nat, \init_i \in \{0,1\}, \decided_i \in
    \{0,1, \bot\}$ and 
    $\rd_i \in \{0,1,\bot\}$.
   Intuitively, $\rd_i = v$ if $i$ learned that some agent just
   decided $v$, for $v \in \{0,1\}$.
      \item The initial local states in $\exchange_i$ have the form
        $\<0,\init_i, 
      \bot,\rd_i^0, \ldots\>$, where 
      $\rd_i = \bot$.
 \item The message-selection 
 value $\mu_i(s,a)$ 
 satisfies the following
   constraint: 
   $i$ sends different messages in the
   following three cases: (a) $a = \decide_i(0)$,
      (b) $a = \decide_i(1)$,
and (c) the remaining cases.
    That means that 
    $j$  can tell from the message it receives
        from $i$ whether $i$ is 
        deciding 0 or 1 in the current round. 
Formally, this means that there are three disjoint sets 
$M^0$, $M^1$, and $M^2$ with 
$\bot \notin M^0\cup M^1$ 
such that if $a =
\decide_i(0)$, then $i$ sends each agent $j$ a message in 
$M^0$, 
if $a
= \decide_i(1)$,
then $i$ sends each agent $j$ a message in 
$M^1$, and otherwise, $i$
sends each agent $j$ a message in 
$M^2$. 
\item The transition function $\delta_i$, when given as input state 
    $s$, action $a$, and a message tuple $(m_1, \ldots,m_n)$,
increases the time component 
$\Time_i$
of $s$ by 1;
if 
$a = \decide_i(v)$, 
it sets 
  $\decided_i$ to $v$, and otherwise leaves $\decided_i$
  unchanged;
it also sets 
$\rd_i = 0$ if $i$ received a message in round $m$ from an agent that
performs  
action $\decide_i(0)$ in that round, sets
$\rd_i = 1$ if $i$ received a message in round $m$ from an agent that performs 
action $\decide_i(1)$ in that round, and otherwise sets $\rd_i = \bot$
  (the assumptions on $\mu_i$ ensure that such messages are
    distinguishable from other messages).
Note that the fact that the $\Time$ component increases
by 1 at every step ensures that the system is \emph{synchronous}; all
agents $i$ have $\Time_i = m$ at time $m$.
\item $\Prop$ contains at least the following propositions (all of
  which are 
  necessary to define the specification below) for each agent $i\in
  \Agents$:  
  \begin{itemize} 
\item $\init_i=v$ for $v\in\{0,1\}$,
  \item   $\decided_i = v$  for $v\in\{\bot,0,1\}$, 
\item   $\Time_i = k$ for $k \in \Nat$, and
  \item $i \in \N$;
\end{itemize}
\item $\pi$ interprets $\init_i=v$, $\decided_i=v$, and $\Time_i = k$ 
in the obvious way from $i$'s local state (e.g., $\pi(r,m)$
makes $\init_i=v$ true iff $r_i(m)$ has 
its second component $\init_i$ equal to $v$),
and interprets $i \in \N$ in the obvious way from $\failures$
(i.e., $\pi(r,m)$ makes $i \in \N$ true iff $i \in \N(r)$, where
$\N(r)$ is the set of nonfaulty agents in run $r$).
\end{itemize}
An EBA context satisfies some minimal properties that we expect all
contexts that arise in the analysis of EBA to satisfy.

We define $\decided_i$ to be an
abbreviation of $\decided_i=0 \lor \decided_i =1$,
take
$\jdecided_i= v$ to be an abbreviation 
for $\decided_i = v \land \ominus \decided_i = \bot$ (intuitively, $i$
just decided $v$),
take
$\deciding_i=v$ to be an abbreviation for
$\decided_i = \bot \land \Circ\decided_i = v$ (intuitively, $i$
is deciding $v$ in the current round)
and
take
$\exists v$, for $v\in \{0,1\}$ to be an abbreviation of
$\bigvee_{i\in \Agents}\init_i =v$.  

Given an  EBA context $\gamma = (\exchange,\failures,\pi)$ 
and an action protocol $P$ for $\exchange$, we define the system $\I_{\gamma,P} = (\R_{\exchange,\failures,P}, \pi)$.
To satisfy the informal specification above, we now seek an EBA-context $\gamma = (\exchange,\failures,\pi)$
and an action protocol $P$ for $\exchange$ such that the following are valid in 
the system $\I_{\gamma,P}$ for 
all agents $i$ and $j$: 
\begin{itemize} 
  \item {\bf Unique Decision:} $ \decided_i = v \rimp \future
\neg  (\decided_i = 1-v) $, for $v\in \{0,1\}$.  
\item {\bf Agreement:}  $ \neg (i \in \N \land j\in \N \land
  \decided_i = v \land \decided_j = v')$, for $v\neq v'$.  
  \item {\bf Validity:} $(\decided_i = v  \land i \in \N) \rimp \exists v$
  \item {\bf Termination:} $i \in \N \rimp \Diamond(\decided_i \ne \bot)$.
\end{itemize} 
If these conditions are satisfied, we call $P$ an \emph{EBA
decision protocol} for the context $\gamma$.  
The tuple $(\exchange,P,\pi)$ is an \emph{EBA-protocol} for failure model $\failures$.  
That is, a protocol solving EBA in the failure model consists of an information-exchange protocol, 
an action protocol that makes decisions, and an interpretation of the basic propositions. 

We are  interested in protocols that are optimal
given the information that is maintained 
by the information-exchange protocol. The following definitions formalize this notion. 
Runs $r,r'$ of two action protocols 
$P,P'$,  respectively,
\emph{correspond} 
if $r(0) = r'(0)$. That is, the two runs have the same failure pattern
and the same initial states for all agents.
Recall that the initial global state of a run, the information-exchange, and
the action protocol together
determine the complete run. 
An action protocol $P$ \emph{dominates} action protocol $P'$  
with respect to a context  $\gamma = (\exchange,\failures,\pi)$, written
 $P' \leq_\gamma P$
if, for all corresponding runs 
$r\in \R_{\exchange,\failures,P}$ and $r'\in \R_{\exchange,\failures,P'}$
and all agents $i$ that are nonfaulty in $r$ and times $m$, if
$P_i(r_i(m)) = \decide_i(v)$ for 
$v\in \{0,1\}$ then  
$P'_i(r'_i(m')) \neq \decide_i(w)$ for any $m' < m$ and $w\in \{0,1\}$. 
That is, $P$ makes it decisions no later than $P'$.   
$P$ \emph{strictly dominates} $P'$ with respect to $\gamma$ if $P'
\leq_\gamma  P$ and it is not the case that that $P \leq_\gamma  P'$.
An EBA decision protocol $P$ is \emph{optimal} with respect to an  
EBA context $\gamma$ 
if no EBA decision protocol 
for $\gamma$
strictly dominates $P$.

\section{Optimal EBA with respect to limited information exchange} \label{sec:opt-lim-ie}

In this section, we describe 
a knowledge-based program for EBA that is somewhat biased towards 0,
show that it is correct, and show that it is 
optimal with respect to all information-exchange protocols that
satisfy a certain safety condition.
As discussed in the introduction, there is no protocol for EBA in the
presence of omission failures where an
agent decides 0 as soon as it hears that some agent had an initial
preference of 0.  So we consider instead a program where an agent
decides 0 if it hears that some agent had an initial preference of 0
via a chain of agents; this is essentially the condition used to
decide 0 in the crash-failure case.

\commentout{
It is useful to reason about the pattern in which the information about the
initial preferences of agents spread along the messages. In the presence of
crash failures, if an agent learns about an initial preference $v$
for the first 
time, that value necessarily has to reach that agent through a chain of faulty
agents where the failures occur at distinct times. This notion of a chain can be
extended to the omission failure case in order to emulate the information
pattern of crash failures, forming a solution to the problem outlined in
Section~\ref{sec:impossibility}.

In light of this observation, we now describe a knowledge-based program for EBA
where the decision rules are characterized by chains carrying information about
0's. We consider two EBA decision protocols that implement this program and are
optimal for the particular information exchange that they use. This is
interesting not  only because it illustrates the concept of optimality with
respect to limited information exchange, involves short messages, and requires
little space, but in the process we develop sufficient conditions on the
information exchange for the protocol to be optimal. 
}

A sequence $i_0, \ldots, i_{m}$ of distinct agents is a \emph{0-chain
of length $m$ in run $r$ of interpreted system $\I$} 
if (a)
$\I,(r, 0) \models \init_{i_0} = 0$, (b) for all $m'$ with 
$0 \le m' \le m$, agent $i_{m'}$ first decides 0 in round $m'+1$ of $r$, and (c)
for all $m'$ with $1 \le m' \le m$,
$i_{m'}$ knows at the point $(r,m')$ that $i_{m'-1}$ has just decided 
0.
  We say that \emph{an agent $i$ receives a 0-chain in round $m$} in
  run $r$   if there is a 
  $0$-chain of length $m$ that ends with
  agent $i$ in run $r$.

Variants of what we call a 0-chain exist in the literature \cite{CGM14, HMW}.
The 0-chain definition in \cite{CGM14}, defined for crash failures,
requires only that $i_{m'}$ receives a message 
from 
$i_{m'-1}$ in round
$m'$ for $m' \ge 1$.
The 0-chain definition in \cite{HMW} 
requires that $i_m$ receives a message from $i_{m-1}$ in round $m'$
for $m' \ge 1$ and $i_m$ considers it possible that $i_{m-1}$ is
nonfaulty when it receives the message (this will automatically be the
case with crash failures, which is all that are considered in
\cite{CGM14}, but is not necessarily the case with omission failures,
which are considered in \cite{HMW}); moreover, $i_m$ is required to be
nonfaulty.  

Let $\kbp^0_i$ be the following knowledge-based program for agent $i$:

\begin{program}
  \DontPrintSemicolon
  \lIf{$\decided_i\neq \bot$}{$\noop$}
  \lElseIf{$\init_i = 0
    \; \lor \;
            K_i (\bigvee_{j \in \Agents} \jdecided_j = 0)$
  }{$\decide_i(0)$}
  \lElseIf{$K_i(\bigwedge_{j \in \Agents} \neg (\deciding_j = 0))$}
  {$\decide_i(1)$}
  \lElse{$\noop$}
  \caption{$\kbp^0_i$}
\end{program}
\noindent In words: as long as $i$ hasn't already decided, then $i$ decides 0
if it has an initial preference of 0 or knows that someone just decided 0;
$i$ decides 1 if it knows that no agent can be currently deciding 0; otherwise,
it does nothing.
$\kbp^0$ is essentially the same as the knowledge-based program used
by Casta\~{n}eda et al.~\cite{CGM14} in the case of crash failures.
We will show 
that the second condition for deciding 0, that $i$ knows that someone
has just decided 0, holds iff $i$ receives a 0-chain.  It follows that
if $i$ hasn't already decided and is not deciding 0, then $i$ decides 1
iff $i$ knows that no agent is receiving a 0-chain.  This latter
condition is very close in spirit to Castan\~{n}eda et al.'s notion of
there being no \emph{hidden paths}.

$\kbp^0$ satisfies all the EBA properties
in all EBA contexts.

\begin{proposition} \label{prop:kbp0correct}
  If $\gamma=(\exchange, \failures, \pi)$ is an EBA context,
  then 
all implementations of the knowledge-based program $\kbp^0$ with
respect to $\gamma$ are EBA decision protocols for $\gamma$.
Indeed, all implementations of $\kbp^0$ terminate after at most $t+1$
    rounds of message exchange and Validity holds even for faulty agents.
\end{proposition} 

We next prove that some implementations of $\kbp^0$ are
actually optimal EBA 
decision
protocols
in certain EBA contexts.
Instead of individually proving optimality of $\kbp^0$ with respect to specific 
information exchanges, we first give a sufficient condition for optimality
in EBA contexts. We then look at two specific contexts  
that satisfy this safety property and show that implementations of $\kbp^0$
in them are optimal.

\begin{definition}[safety]
A knowledge-based program $\kbp$ is \emph{safe with respect to 
an EBA context $\gamma = (\exchange,\failures,\pi)$} if, for all
implementations $P$ of 
$\kbp$ with respect to $\gamma$ and
all points $(r,m)$ of $\I= (\R_{\exchange,\failures,P},\pi)$, the
following two conditions hold:
\begin{enumerate}
\item If $i$ has not received a 0-chain by $(r,m)$,
then there exists a point $(r',m)$ such that 
$r_i(m)= r'_i(m)$ and all agents have initial preference 1 in $r'$.
    \item If $\I, (r,m) \models \neg
        K_i(\bigwedge_{j \in \Agents}  
        \neg (\deciding_j =0 ))$ and $i$ does not decide before
        round $m+1$ in $r$, then
   there exists a
  point $(r',m)$ such that:
  \begin{enumerate}
    \item $r_i(m)= r'_i(m)$,
    \item $i$ is nonfaulty in $r'$,
    \item some agent $j$ that is nonfaulty in $r'$ decides 0 in round
      $m+1$ of $r'$;
      moreover, if $m \ge 1$, there exists a run $r''$ and an agent
      $j'$ such that $j$ and $j'$ are
      nonfaulty in $r''$, $r'_j(m) = r''_j(m)$, and 
      $j'$ decides 0 in round $m$ in round $r''$.
        \end{enumerate}
\end{enumerate}
\end{definition}

Intuitively, (1) says that the only way that an agent learns that some 
agent had an initial preference of 0 is via a
0-chain; thus, (1) implies that if $i$ has not received a 0-chain by
$(r,m)$, then $I, (r,m) \models \neg K_i \exists 0$.  Clause
(2) says that the only way an agent is unable to 
decide 1 (i.e., the test for deciding 1 in $\kbp^0$ does not hold) is if it 
considers it possible that some 
nonfaulty
agent is 
deciding 0; thus, (2) implies that if $\I, (r,m) \models
\neg K_i (\bigwedge_{j \in \Agents} \neg (\deciding_j = 0))$ and $i$ has
not decided by round $m$, then
$\I, (r,m) \models
\neg K_i (\bigwedge_{j \in \N} \neg (\deciding_j = 0))$.

Note that a knowledge-based program will in general not be safe with
respect to an FIP, since 
an agent $i$ may learn that some agent $j$ had an initial preference of 0 
without
receiving a 0-chain. As a result, 
the first condition will not hold, 
since 
at all points that $i$ considers possible, 
$j$ has an initial preference of 0. 
(In Section~\ref{sec:fip-optimal}, we show that a small modification
of $\kbp^0$ is optimal even with full information exchange.)
But, as we shall see, $\kbp^0$
is safe with 
respect
to two EBA contexts 
of interest, where
agents do not keep track of who is faulty.

\begin{theorem}\label{thm:safe}
    If $\gamma$ is 
  an  EBA context and $\kbp^0$ is safe with respect to $\gamma$, then
all implementations of $\kbp^0$ are optimal with respect to $\gamma$.
\end{theorem}

We now describe two families of concrete 
EBA contexts with
respect to which $\kbp^0$ is safe, parameterized by the number $n$ of
agents involved.  For the first, let $\exchange_{\min}(n)$ be the
\emph{minimal} information-exchange protocol for $n$ agents, where for
each agent $i$, the following hold:
\begin{itemize}
  \item The local states have the form $\<\Time_i,\init_i,
    \decided_i,\rd_i\>$.
    Thus, the local states include just what is required
        in an EBA context
  \item The initial local states have the form $\<0,\init_i,\bot,\bot\>$.
  \item 
  $M_i = \{0, 1\}$, $M^0 = \{0\}$, $M^1 = \{1\}$ and $M^2= \{\bot\}$. 
  \item For each agent $j$, if $a = \decide_i(v)$ then $\mu_{ij}(s,a)
    = v$; otherwise,
        $\mu_{ij}(s,a) = \bot$. 
    Note $\mu_{ij}$ satisfies the constraint we imposed for
EBA     contexts. 
    Intuitively, if 
    $\mu_{ij}(s,a) = v \ne \bot$, then $i$ is about to decide $v$.
  \item The state-update component $\delta_i$ is defined 
    on $\Time_i$, $\init_i$, $\decided_i$, and $\rd_i$ just
    as in
    EBA contexts.
\end{itemize}

For the second, let $\exchange_{\basic}(n)$ be the
\emph{basic} information-exchange protocol for $n$ agents, where the
local states of 
agents are like those in a minimal information-exchange protocol,
except that, in addition to the other messages allowed in a basic
information-exchange protocol,  each agent $i$ can send 
a message of the form $(\init,1)$, 
and their local states 
have one additional component, $\#1_i$, that
intuitively counts how many messages of the form 
$(\init, 1)$ $i$
received in the last round.  In more detail, 
for each agent $i$, the following hold:
\begin{itemize}
  \item The local states have the form $\<\Time_i,\init_i,
    \decided_i,\rd_i,\#1_i\>$, where $\#1_i \in
  \{0, \ldots, n\}$.
  \item The initial local states have the form $\<0,\init_i,\bot,\bot,0\>$.
  \item 
  $M_i = \{0, 1,(\init,1)\}$, $M^0 = \{0\}$, $M^1 = \{1\}$, and $M^2 = \{(\init,1), \bot\}$.  
  \item For all agents $j$, if $a = \decide_i(v)$ then $\mu_{ij}(s,a)
    = v$ for $v \in 
    \{0,1\}$; if $a = 
    \noop$
    and $s$ has the form $\<m,1,\bot,\bot,k\>$, 
    then $\mu_{ij}(s,a)  = (\init,1)$; otherwise,
    $\mu_{ij}(s,a) = \bot$.
  \item The state-update component $\delta_i$ is defined as in
EBA contexts, with the added constraint
that $\#1_i$ is updated to the number of messages of the form
$(\init,1)$ that $i$ receives in the current round if $\decided_i =
\bot$ and $i$ does not receive a message $v \in \{0,1\}$ from some
agent $j$; otherwise, $\#1_i$ is set to $0$ (essentially, for technical
reasons, once a decision is made, $\#1_i$ is ignored).
\end{itemize}

Let $\gamma_{\min,n,t} = (\exchange_{\min}(n), SO(t), \pi_{\min,n})$ denote
the family of minimal contexts where there are $n$ agents, at most $t
< n$ faulty agents, the language includes $\rd_i = v$ in addition to
$\Time_i=k$, $\init_i = v$, $\decided_i(v)$, and $i \in \N$, for $i \in
\{1,\ldots,n\}$, and $\pi_{\min,n}$ interprets these primitive
propositions in the obvious way.
Similarly, let $\gamma_{\basic,n,t} = (\exchange_\basic(n), SO(t),
\pi_{\basic,n})$ denote the family of basic contexts with $n$ agents
and $t < n$ faulty agents, where the language includes $\#1_i = k$ in
addition to all the primitive proposition used in minimal contexts,
where again, $\pi_{\basic,n}$ interprets all the primitive
propositions in the obvious way.

\begin{proposition}\label{pro:optimal}
$\kbp^0$ is safe with respect to all contexts $\gamma_{\min,n,t}$ and
  $\gamma_{\basic,n,t}$ such that $n - t \ge 2$.
  \end{proposition}
\commentout{
The proof is constructive;
we explicitly construct the required runs.
Proving that $\kbp^0$ satisfies the first part of the safety condition
is straightforward.
For the second part, there exists a run $r'$ that $i$ considers
possible where an agent $j$ decides 0 in round $m+1$.  The problem is
that $i$ and $j$ may be faulty in $r'$.  
In this case, 
we choose agents $i'$ and $j'$ that are nonfaulty in $r'$ (there must
be such agents, since $n \ge t-1$),
and construct
a run $r''$ where all agents have
the same initial preferences as in $r'$, $i$ and $j$ are nonfaulty as
required, and messages from $i$ (resp., $j$) to an agent $j''$ are
blocked in $r'$ iff messages from $i'$ (resp., $j'$) to $j''$ are
blocked in $r''$.
This turns out to work.  For reasons of space, we defer the details to
Appendix~\ref{proofpro:optimal}.
} %

Finally, we provide EBA decision protocols that implement $\kbp^0$ in the two
contexts of interest.  Let $P^{min}$ be 
the protocol implemented by the following (standard) program:

\begin{program}
  \DontPrintSemicolon
   \lIf{$\decided_i \neq \bot$}{$\noop$}
  \lElseIf{$\init_i = 0 \lor \rd_i = 0$}{$\decide_i(0)$}
  \lElseIf{$\Timef_i = t+1$}{$\decide_i(1)$}
  \lElse{$\noop$}
  \caption{$P^{min}_i$}
\end{program}

Intuitively, this EBA decision protocol decides $0$ 
if the agent has initial value $0$ or hears of a 0-decision by another
agent.
If the agent does not hear about a $0$-decision by time 
$t+1$, 
then it decides 1.

\begin{theorem} \label{thm:ImplementsSimple}
    If $t\leq n-2$, then
        $P^{min}$ implements
    $\kbp^0$ in the EBA context~$\gamma_{\min,n,t}$. 
\end{theorem}

Finally, let $P^{\basic}$ be the EBA decision protocol that implements the
following program:

\begin{program}
  \DontPrintSemicolon
  \lIf{$\decided_i \neq \bot$}{$\noop$}
  \lElseIf{$\init_i = 0 \lor \rd_i = 0$}{$\decide_i(0)$}
  \lElseIf{$\#1_i > n- \Timef_i \lor \rd_i = 1$}{$\decide_i(1)$}
  \lElse{$\noop$}
  \caption{$P^{basic}_i$}
\end{program}

\begin{theorem} \label{thm:implement2}
If $t\leq n-2$, then 
$P^{\basic}$ implements 
$\kbp^0$ in the EBA context~$\gamma_{\basic,n,t}$. 
\end{theorem} 

We get the following immediate corollary to Theorems~\ref{thm:safe},
\ref{thm:ImplementsSimple}, and
\ref{thm:implement2}.
\begin{corollary} $P^{\basic}$ is optimal  
  with respect to $\gamma_{\basic,n,t}$ and $P^{min}$ is optimal with 
  respect to $\gamma_{\min,n,t}$.
\end{corollary}

\section{An optimal full-information protocol for EBA} \label{sec:fip-optimal}
While $\kbp^0$ is optimal with respect to the basic and minimal 
information-exchange settings,
it is not optimal in the full-information setting, as the following
example shows.

\begin{example}
  Suppose that $t=10$ and $n=20$.  Consider a run where all agents
  have initial value 1, agents 1--10 are faulty, and no faulty agent
  sends a message in any round.  This means that, at the end of the
  first round, each nonfaulty agent knows who the faulty agents are.
  At the end of the second round, it is common knowledge among the
  nonfaulty agents who the faulty agents are: each nonfaulty agent $i$  
  will know at that point that each other nonfaulty agent $j$ knows who
  the nonfaulty agents are and (by the same reasoning) $i$ knows that
  $j$ knows that all the nonfaulty agents know who the faulty agents
  are, and so on.   Moreover, it is common knowedge among the nonfaulty
  agents  
  that no nonfaulty agent has already decided, and it is
  not common knowledge 
  among 
  the nonfaulty agents that some agent had an
  initial preference of 0, while it is 
  common
  knowledge that some
  agent had an initial preference of 1.    As we show below, this
  means that the nonfaulty agents can decide on 1 in round 3.  On the
     other hand, with $P^{basic}$ and $P^{min}$, the nonfaulty agents
  would not decide in this run until round 12.
\end{example}

Intuitively, we take $\ck{\N}(\faultyag)$ to hold when it is 
common 
knowledge among the nonfaulty agents who the nonfaulty agents are.  It
turns out this can happen only if the nonfaulty agents have common
knowledge of $t$ agents that are faulty.  Thus, we 
take 
$\ck{\N}(\faultyag\land \phi)$, for each formula $\phi$, to be an abbreviation for $
\exists A \subseteq \Agents(|A| = t \land 
\ck{\N}(\bigwedge_{ i \in   A} (i \notin \N) \land \phi ))$.
We also define 
$\nodecided_\N(x)$ as an abbreviation for $\bigwedge_{j \in \N} \neg
(\decided_j = x)$ for $x \in \{0,1\}$. 
We can now formalize the situation in the example. 

\begin{proposition} \label{prop:cfaulty-necessary}
If $P$ is an optimal protocol in 
the context $\gamma_{\fip,n,t}$ and $\I_{P,\gamma_{\fip,n,t}} ,(r,m) \models
  \decided_i = \bot \land K_i(\ck{\N}(\faultyag \land \nodecided_\N(1) \land
 \exists 0))$, then all undecided agents in $\N(r)$ make a
  decision in round $m+1$, and 
  similarly if $\I_{P,\gamma} ,(r,m) \models \decided_i = \bot \land
  K_i(\ck{\N}(\faultyag \land \nodecided_\N(0) \land \exists 1))$.
\end{proposition} 

It turns out that if we add a condition to $\kbp^0$ that tests for
this common knowledge and decides 
appropriately 
if it holds, we get a
program that is optimal even with full information exchange.
Specifically, let
$\kbp^1_i$ be the following knowledge-based program for agent $i$:

\begin{program}
  \DontPrintSemicolon
  \lIf{$\decided_i\neq \bot$}{$\noop$}
  \lElseIf{$K_{i}(C_{\N}(\faultyag \land \nodecided_{\N}(1) \land \exists 0))$}
    {$\decide_i(0)$}
  \lElseIf{$K_{i}(C_{\N}(\faultyag \land \nodecided_{\N}(0) \land \exists 1))$}
    {$\decide_i(1)$}
  \lElseIf{$\init_i = 0
            \; \lor \;
            K_i (\bigvee_{j \in \Agents} \jdecided_j = 0)$
  }{$\decide_i(0)$}
  \lElseIf{$K_i(\bigwedge_{j \in \Agents} \neg (\deciding_j = 0))$}
  {$\decide_i(1)$}
  \lElse{$\noop$}
  \caption{$\kbp^1_i$}
\end{program}

\smallskip
Note that in the basic and mininimal contexts, agents never learn who
is faulty, so there is never common knowledge among the nonfaulty
agents who the faulty agents are.  Thus, in the contexts
$\gamma_{\min,n,t}$ and $\gamma_{\basic,n,t}$, $\kbp^1$ is equivalent
to $\kbp^0$, so $\kbp^1$ is correct and optimal in these contexts.
As we are about to show, $\kbp^1$ is also correct and optimal with full
information exchange.

To prove correctness and optimality, we follow the approach of
Halpern, Moses, and Waarts \citeyear{HMW} 
and consider a slightly nonstandard full-information context.  
We assume that each agent $i$'s local state does \emph{not}
contain the variables $\decided_i$ and $\rd_i$, but does contain a
variable or variables that keep track of all messages received from
all agents.
If agents keep track of all messages received in their local state,
then, given a decision protocol $P$, 
the variables $\decided_i$ and $\rd_i$ are redundant; their values can
be inferred from the messages received.
Let $\gamma_{\fip,n,t}$ denote the family of full-information contexts
as described above.  Not including $\decided_i$ and $\rd_i$ in the
local state has the 
advantage that, for all decision protocols $P$ and
$P'$, corresponding runs of $P$ and $P'$ in $\gamma_{\fip,n,t}$ are
actually identical; although agents may make different 
decisions, 
their
local states are the same at all times.  (This would not be the case
if the local states had included information  about decisions, and in
particular, if they had included the variables $\decided_i$ and
$\rd_i$.)  It is critical that we are dealing with
FIPs here; the claim is not true for arbitrary information-exchange protocols.

\commentout{
The argument that $\kbp^1$ is correct mainly follows the lines of the
proof of Proposition~\ref{prop:kbp0correct}.  The only 
non-trivial change is in showing Agreement, which intuitively follows 
from the observation that if the common knowledge condition for
deciding $v$ holds,
then it must hold for all agents simulatenously; moreover it must be
safe for the agents to decide $v$, (since none of them has decided $1-v$).
}

\begin{proposition} \label{p:kbp1correct}
All implementations of $\kbp^1$ with
respect to $\gamma_{\fip,n,t}$ are EBA decision protocols for $\gamma_{\fip,n,t}$.
\end{proposition}

We need to recall some material from \cite{HMW} in order to use the
characterization of optimality with respect to $\gamma_{\fip,n,t}$.  
Given an \emph{indexical set} $\cS$ of agents,
a point $(r', m')$ is \emph{$\cS$-$\boxdot$-reachable
from $(r, m)$} if there exist runs $r^0, \ldots, r^k$, times $m_0, m_0' ,
\ldots, m_{k}, m_{k}'$, and agents $i_0, \ldots, i_{k-1}$ such that $(r^0, m_0)
= (r, m)$, $(r^k, m_k') = (r', m')$, and for $0 \le j \le k-1$, 
we have 
$i_j \in
\cS(r^{j},m_{j}') \cap \cS(r^{j+1}, m_{j+1})$ and $r^{j}_j(m_j') =
r^{j+1}_j(m_{j+1})$.%
\footnote{Halpern, Moses, and Waarts \cite{HMW} introduced a family
of \emph{continual common knowledge} 
operators $\contcb{\cS}$ such that $\contcb{\cS}\phi$ holds at a point $(r,m)$
iff $\phi$ is true at all points $(r',m')$ that are $\cS$-$\boxdot$-reachable
from $(r, m)$. We get standard (indexical) common knowledge by taking $m_k =
m_k'$ in the definition of continual common knowledge; since we are working with
synchronous systems, we could restrict to taking $m_j' = m_{j+1}$.} Using the
notation of \cite{HMW}, let $\N \land \cO$
denote the indexical set where $(\N \land \cO)(r,m)$ consists of all agents that
are nonfaulty and about to decide 1
or have already decided 1
at the point $(r,m)$. 
Let $(\N \land \Z)$ be the analogous set for 0.

\begin{definition}[weak safety]
A knowledge-based protocol 
$\kbp$
is \emph{weakly safe with respect
to an EBA context $\gamma$} if, for all implementations $P$ of $\kbp$ and 
all points $(r,m)$ of $\I= (\R_{\exchange,\failures,P},\pi)$
and all agents $i$, 
if $\I_{\gamma_{\fip,n,t},P},(r,m) \models i \in \N \land 
\Circ (\decided_i = \bot)$, then there exist points
      $({r^0}',m)$, $({r^0}'',m)$, $({r^1}',m)$, and 
      $({r^1}'',m)$ 
      such that:
          \begin{enumerate}[1.]
                      \item $r_i(m) = {r^0_i}'(m) = {r^1_i}'(m)$,
            \item $i$ is nonfaulty in ${r^0}'$ and ${r^1}'$,
            \item $({r^0}'',m'')$ is $(\N \land \Z)$-$\boxdot$-reachable
              from $({r^0}',m)$, 
            \item $({r^1}'',m'')$ is $(\N \land \cO)$-$\boxdot$-reachable
              from $({r^1}',m)$, 
    \item all agents have initial preference 0 in ${r^0}''$, 
            \item all agents have initial preference 1 in ${r^1}''$.
          \end{enumerate}
          \commentout{
        \item If $\I_{\gamma_{\fip,n,t},P},(r,m) \models i \in \N \land 
          \Circ (\decided_i = \bot)$, then there exists points $(r',m)$ and 
          $(r'',m'')$ such that:
          \begin{enumerate}[1.]
            \item $r_i(m) = r'_i(m)$,
            \item $i$ is nonfaulty in $r'$,
            \item $(r'',m'')$ is $(\N \land \Z)$-$\boxdot$-reachable from $(r',m)$,
            \item all agents have initial preference 0 in $r''$.
          \end{enumerate}
          }
\end{definition}

Our interest in weak safety is motivated by the following result
proved by Halpern, Moses, and Waarts \cite{HMW}.
The statement of the result uses two operators, $B_i^\N$ 
  $C_{\cS}^\boxdot$.
$B_i^\N$ is an abbreviation of
  $K_i(i \in \N \rimp \phi)$. Thus, 
  $\I,(r,m) \models B_i^\N\phi$ if and only if  
  $\I,(r',m') \models \phi$ for all points $(r',m')$
  such that $r_i(m) = r'_i(m')$ and $i \in \N(r')$.
 Intuitively, $B_i^\N$ holds if $i$ knows that if it is nonfaulty, then
 $\phi$ holds.
   The    $C_{\cS}^\boxdot$. operator has a characterization in terms of 
  $\cS$-$\boxdot$-reachability. In \cite{HMW}, it is shown that 
  $\I,(r,m) \models C_\cS^\boxdot \phi$ if and only if 
  $\I,(r',m') \models \phi$ for all points $(r',m')$ 
  that are $\cS$-$\boxdot$-reachable from $(r,m)$.  We are interested
  in the cases that $\cS$ is either
  $\N \land \cO$ or $\N \land \Z$.

\begin{theorem}\label{thm:HMWchar} \cite[Theorem 5.4]{HMW}
An EBA protocol $P$ is optimal with respect to $\gamma_{\fip,n,t}$
  iff the following two conditions hold:
  $$\begin{array}{l}
\I_{\gamma_{\fip,n,t},P}
\models i \in \N \rimp
\ (\Circ(\decided_i = 0)
        \Leftrightarrow B_i^\N (\exists 0 \wedge
      C^\boxdot_{\N \land \cO}\exists 0 \wedge \neg (\Circ(\decided_i = 1)))) \\
\I_{\gamma_{\fip,n,t},P}
\models i \in \N \rimp 
\ (\Circ(\decided_i = 1)
        \Leftrightarrow B_i^\N (\exists 1 \wedge
      C^\boxdot_{\N \land \Z}\exists 1 \wedge \neg (\Circ(\decided_i =
      0)))).
\end{array}$$
\end{theorem}

Using Theorem~\ref{thm:HMWchar}, we can prove that weak safety implies
optimality for $\kbp^1$.

\begin{theorem} \label{thm:kbp1opt} 
  If $\kbp^1$ is weakly safe with respect to $\gamma_{\fip,n,t}$
    then all implementations of $\kbp^1$
  are optimal with respect to $\gamma_{\fip,n,t}$.
\end{theorem}

To show that the knowledge-based program $\kbp^1$ satisfies weak safety with
respect to the full-information context, we give a 
constructive proof that  explicitly constructs the sequences of points witnessing 
the conditions of the definition of weak safety.

The main idea of the proof comes from the following observation. If a
nonfaulty agent $i$ is unable to decide, then the common knowledge
conditions in $\kbp^1$ do not hold. Then, roughly speaking, we can
show that there exists points $({r^0}',m)$ and $({r^1}',m)$, as in the
definition of weak safety, such that a faulty agent $k$ acts nonfaulty
throughout the run.%
\footnote{Note that since the set of faulty agents in a run is
determined by the failure pattern, it is consistent that an agent $i$ is
faulty in a run although it acts nonfaulty throughout the run.   Since
all that really matters for our result is that no agent can detect
that agent $i$ is faulty, we could obtain our result by assuming that
$i$'s faulty behavior involved only not sending messages to itself.}
Moreover, since this faulty agent $k$ did not
display faulty behavior, all runs where another agent is faulty
instead of $k$ are indistinguishable from this point. The existence of
a faulty agent that acts nonfaulty turns out to be a strong condition
that allows the 
construction
of a sequence of $\boxdot$-reachable points where the end point
has a modified message pattern in addition to satisfying the same condition on
$k$. This is possible by temporarily making $k$ exhibit faulty behavior in
intermediate points of the sequence. Therefore we use the existence of such $k$
as an invariant that allows taking steps through the $\boxdot$-path.

\begin{theorem} \label{thm:kbp1-safe} $\kbp^1$ is weakly safe with respect to 
  $\gamma_{\fip,n,t}$.
\end{theorem}

Proposition~\ref{p:kbp1correct}, Theorem~\ref{thm:kbp1opt}, and
Theorem
\ref{thm:kbp1-safe}  
together imply the following corollary. 
\begin{corollary} \label{theorem:optimal}
  All implementations of $\kbp^1$ with respect to $\gamma_{\fip,n,t}$ are
  optimal with respect to $\gamma_{\fip,n,t}$.
\end{corollary}

Implementing $\kbp^1$ in polynomial time is possible using the 
compact communication graph representation of the full-information 
exchange due to \cite{MT}. Intuitively, the common knowledge conditions are implemented
using the observation that if an agent's faultiness is common knowledge among 
the nonfaulty agents, it must be \emph{distributed knowledge} at the previous
time (where a fact $\phi$ is distributed knowledge among an indexical
set $\cS$ of agents if the agents would know $\phi$ if they pooled
their knowledge together; for example, the set of faulty agents is
distributed knowledge among the nonfaulty agents if, between them, the
nonfaulty agents know who the faulty agents are).  
Since nonfaulty agents 
send 
messages describing their complete state
in every round, we can check whether
$C_\N(\faultyag)$ holds at a point $(r,m)$ by considering the local
states at $(r,m-1)$ of the agents that  
nonfaulty agents heard from in round $m$.

\begin{proposition} \label{prop:kbp1-polynomial}
  There exists a polynomial-time implementation 
    $P^{\opt}$ 
   of $\kbp^1$ with respect to a 
  full-information exchange. 
\end{proposition}

\section{Discussion} \label{sec:cost}
We introduced the notion of limited information exchange,
examined optimality for EBA with respect to various information-exchange
protocols
and described an efficiently implementable optimal FIP. 
There is clearly far more to be done.  
There are two
short-term directions we are currently pursuing.  
First, we hope to explore the impact of limited information
exchange on other 
protocols of interest. 
Second,
we are exploring the application of epistemic synthesis techniques that allow
the automated derivation of protocols from a knowledge-based program
in the context of limited information-exchange models. This seems to give
the techniques far more scope 
(cf. \cite{HuangM13,HuangM14}).

We conclude this discussion by taking a closer look at the costs and
benefits of limited information exchange for EBA.
We focus on the two settings considered in Section
\ref{sec:opt-lim-ie}, as well as the full-information context,
and consider the cost both in terms of the number of bits sent and the
number of rounds required to reach a decision in the most likely case,
where there are no failures.  While the results are straightforward,
they help highlight the tradeoffs involved.
Let $\gamma_{\fip,n,t}$ be a full-information EBA context
with omission 
failures and let $P^{\fip}$ be an implementation 
of $\kbp^1$ in $\gamma_{\fip,n,t}$.

We start  by considering message complexity in terms of bits.
In the minimal information-exchange protocol $P^{min}$, each message 
can be represented using a single bit and agents send a message only 
when they first decide, otherwise staying silent.
Since each agent sends exactly one message in each run, and
sends it to all the other agents, $n^2$ bits are sent altogether.
In the basic 
information-exchange protocol $P^{\basic}$, we still require only a
constant number of bits  
to represent messages and agents send messages to all other agents
as long as they are undecided, which means for at most $t+1$
rounds. We then get the following result: 
\begin{proposition}
  In each run of $P^{min}$, $n^2$ bits are sent in total; in each run
  of    $P^{\basic}$, at most $O(n^2t)$ bits are sent in total.
\end{proposition}

By way of contrast, a standard communication graph implementation of a
full-information 
exchange uses $O(n^{4}t^{2})$ bits \cite{MT}.

We next consider decision times.  
We focus on the failure-free case as, in most applications, 
the most common runs are runs with no failures. 
\begin{proposition} \label{prop:decision-times}
  If $r$ is a 
  failure-free run, then
    \begin{enumerate}[(a)]
      \item If there is at least one agent with an initial preference
        of 0 in $r$, then 
all agents decide by round 2
with $P^{min}$, $P^{\basic}$, and $P^{\fip}$.
\item If all agents have an initial preference of 1, then 
    all agents decide by round $t+2$ with $P^{min}$ 
       and by round $2$ with $P^{\basic}$ and $P^{\fip}$.      
  \end{enumerate}
\end{proposition}

Thus, for failure-free runs, the agents in the basic context 
and the full-information context decide at the same time. The only 
failure-free 
run
where the  
basic context results in an earlier decision than the minimal context
is the run where every agent starts with an initial preference of 1. 
As a result, implementing full-information and incurring a quadratic overhead in the
number of bits never leads to  an improvement for failure-free runs. If we
assume  
that each configuration of initial preferences is equally likely, 
using the basic context over the minimal context for failure-free runs
is only an improvement $1/2^n$ of the time.

If failure-free runs are sufficiently common, this suggests that the
gain of using an FIP may not be worth the cost; even the tradeoff
between $P^{\basic}$ and $P^{min}$ is not so clear.  We conjecture that
even in runs with failures, $P^{\basic}$ may not be much worse than
$P^{\fip}$.  This emphasizes the advantages of considering limited
information exchange, and further motivates considering optimal protocols
with limited information exchange more broadly.

\begin{podc}
\begin{acks}
  Alpturer and Halpern were supported in part by AFOSR grant
FA23862114029.  Halpern was additionally supported in part by ARO grants
W911NF-19-1-0217 and W911NF-22-1-0061. 
The Commonwealth of Australia (represented by the Defence Science and Technology
Group) supported this research through a Defence Science Partnerships agreement.
We thank Yoram Moses for useful comments on the paper.
\end{acks}
\end{podc} 

\begin{podc}
\bibliographystyle{ACM-Reference-Format}
\bibliography{z,joe}
\end{podc}

\begin{full}
\appendix

\section{Proofs}

\newenvironment{RETHM}[2]{\trivlist \item[\hskip 10pt\hskip\labelsep{\sc #1\hskip 5pt\relax\ref{#2}.}]\it}{\endtrivlist}
\newcommand{\rethm}[1]{\begin{RETHM}{Theorem}{#1}}
\newcommand{\erethm}{\end{RETHM}}
\newcommand{\relem}[1]{\begin{RETHM}{Lemma}{#1}}
\newcommand{\recor}[1]{\begin{RETHM}{Corollary}{#1}}
\newcommand{\repro}[1]{\begin{RETHM}{Proposition}{#1}}
\newcommand{\erepro}{\end{RETHM}}
\newcommand{\erelem}{\end{RETHM}}
\newcommand{\erecor}{\end{RETHM}}

\subsection{Proofs for Section~\ref{sec:opt-lim-ie}}

\repro{prop:kbp0correct}
If $\gamma=(\exchange, \failures, \pi)$ is an EBA context,
then 
all implementations of the knowledge-based program $\kbp^0$ with
respect 
to $\gamma$ are EBA decision protocols for $\gamma$.
  Indeed, all implementations of $\kbp^0$ terminate after at most $t+1$
    rounds of message exchange and Validity holds even for faulty agents.
\erepro
\begin{proof}

  Fix an implementation $P$ of $\kbp^0$ in $\gamma_{\fip,n,t}$.

  Unique-Decision follows from the fact that $P_i$ makes at most one decision
  per round and the fact that whether a decision was made is recorded in the
  local state variable $\decided_i$.

  To see that Agreement holds, we first show by induction on $m$ that if $i$ has
  not decided before round $m+1$ and $\I,(r,m) \models \init_{i} = 0 \lor
  K_{i}(\bigvee_{j \in \Agents}\jdecided_{j} = 0)$, then $i$ receives a \zchain
  in round $m$. If $m = 0$, then it must be the case that $\I,(r,m) \models
  \init_{i} = 0$ which is a \zchain with $i_1 = i$. If $m > 0$, we cannot have
  $\I,(r,m) \models \init_{i} = 0$ (as $i$ would have decided earlier) so we
  must have $\I,(r,m) \models K_{i}(\jdecided_{j} = 0)$ for some $j \in
  \Agents$. Then, $\I,(r,m) \models \jdecided_{j} = 0$ and the result follows
  from the induction hypothesis.

  Suppose by way of contradiction that $r$ is a 
  run where there exist nonfaulty agents $i$ and $j$ and a time $m$ such that
  $\I,(r,m) \models \decided_i = 0 \land \decided_j = 1$.
  Suppose that $j$ decides 1 in round $m_j+1$ and $i$ decides 0 in round $m_i+1$, 
  so that the decision conditions first hold at times $m_i$ and $m_j$, respectively. 
  If $m_{j} \leq m_{i}$, must have either $\I,(r,m_{i}) \models \init_{i} = 0$
  or $\I,(r,m_{i}) \models K_{i}(\jdecided_{k} = 0)$ for some $k \in \Agents$.
  Using our observation, we can 
  conclude that $i$ receives a \zchain at time $m_i$, which implies that 
  there exists an agent $i'$ such that 
  $\I,(r,m_{j}) \models \deciding_{i'} = 0$. Hence, 
  $\I,(r,m_{j}) \models \neg K_{j}(\neg (\deciding_{i'}=0))$,
  so  $j$ cannot decide 1 at $m_{j}$.
  If $m_{j} > m_{i}$, since $i$ decides 0 in round $m_i + 1$, 
  we must have $\I,(r,m_i) \models \init_i = 0 \bigvee_{j \in \Agents}
  K_i (\jdecided_j = 0)$. We can again use our observation to 
  conclude that $i$ receives a \zchain at time $m_i$ in $r$.
  As $i$ is nonfaulty, $j$ must hear from $i$ in round $m_{i}+1$,
  so $\I,(r,m_j) \models K_j(\jdecided_i = 0)$. 
  It follows that agent $j$ should decide 0 in this run,
  contradicting the assumption that $j$ decides 1.

  For Validity, observe that if $i$ decides 0, using the previous observation,
  there must be a \zchain, and hence an agent 
  that had an initial preference of 0. 
  If agent $i$ decides 1, then $i$ did not 
  decide 0 in the first round and therefore we must have $\init_i = 1$.
  Note that this argument holds even if $i$ is faulty.
  
  Finally, we prove Termination by showing that in all runs, all nonfaulty agents
  must decide by round $t+2$ after at most $t+1$ rounds of message
  exchange.  To see this, we first show that if some agent
  decides 0 in run $r$, then all agents that decide 0 must do so by round
  $t+2$ after at most $t+1$ rounds of message exchange.  For suppose by
  way of contradiction that some agent $i$ decides 0 in round $m > t+2$. 
  Then there must be a \zchain $i_1, \ldots, i_m$ with $i_m = i$.  
  All the agents on the \zchain are distinct. Since there are at most
  $t$ faulty agents in a run, one of $i_1, \ldots, i_{t+1}$ must be
  nonfaulty, say $i_j$. But that means that all agents (including $i$)
  would have received a message 
  from $i_j$ in round $j \le t+1$ from which they could 
  infer that $i_j$ is about to decide 0 (by our assumption regarding
  $\mu$ in an EBA context), and they  would all decide
  0 in round $j+1 < m$ if they have not done so yet. This gives the
  desired contradiction.  
  It follows that if $i$ has not decided 0 by round $t+2$, then 
  $\I,(r,t+2) \models K_{i}(\bigwedge_{j \in \Agents}
  \neg (\deciding_j = 0 ))$, 
  so $i$ will decide $1$ in round $t+2$ if it has not already decided.  
  \end{proof}

\rethm{thm:safe}
If $\gamma$ is an EBA context and $\kbp^0$ is safe with
respect to $\gamma$, then all implementations of $\kbp^0$ are
    optimal with respect to $\gamma$. 
\erethm
\begin{proof}
Suppose 
by way of contradiction
that $P$ is an implementation of $\kbp^0$ with respect to 
an  EBA context $\gamma$ and
EBA-protocol
$P'$ strictly dominates $P$ with respect to
$\gamma$.
If $P'$ strictly dominates $P$, by definition, there exist
corresponding runs  
$r$ of $\I_{\gamma,P}$ and $r'$ of $\I_{\gamma,P'}$, 
a nonfaulty agent $i$, and a round $k$ such that $P'$ decides in round
$k$ of $r'$ 
and $P$ does not decide before round $k+1$ in $r$.
Clearly, there also exist corresponding runs $r$ and $r'$,
an agent $i$  
(possibly faulty) and a round $k$ such that $P'$ decides in round
$k$ of $r'$, and  
$P$ either does not decide before round $k+1$ of $r$, or decides
differently in round  $k$ of $r$.
Let $k$ be the earliest such round, and let $r$ and $r'$ be the
corresponding runs of $P$ and $P'$, respectively, where in $r'$, $P'$
decides in round $k$ and in $r$, $P$ either does not decide by round
$k$ or reaches a different decision in round $k$ than $P'$.
(Note that the $k$ we use here may be smaller than the smallest $k$
such that $P'$ decides in round $k$ and $P$ does not decide before
round $k+1$ in corresponding runs.)
\commentout{
at which
some
agent $i$ decides in round $k$ of
a run $r'$ of 
$\I_{\gamma,P'}$ and $i$ either does not decide at or before round
$k$ of the corresponding run $r$ of $\I_{\gamma,P}$ or $i$
makes a different decision in round $k$ of $r$ than in round $k$ of
$r'$.   
} %
Since the same information-exchange protocol is used in both
systems, all agents must have exactly the same state up to time $k-1$ in
corresponding runs of the two systems:
either they have not decided yet, or they have decided and made the
same decision, so they
will send the same
messages and undergo the same state transitions in corresponding runs
of $P$ and $P'$ up to time 
$k-1$.
We now consider two cases.

Suppose that $i$ decides 0 in round $k$ of $r'$.  Since $i$ either
does not decide 
at or before
round $k$ of $r$ or decides 1 in this round,
$i$ did not receive a 0-chain by $(r,k-1)$.
Since $\kbp^0$ is safe with respect to $\gamma$,
by the first part of the safety condition, it follows that there
is a run $r_1$ of $P$ such that $r_i(k-1) = (r_1)_i(k-1)$
and all agents have initial preference 1 in
$r_1$.  Let 
$r_1'$ be the run of $P'$ corresponding to $r_1$.  As observed above,
$i$ must have the same state 
at $(r_1,k-1)$ and $(r_1',k-1)$.  
Since $i$ also has the same state in $(r,k-1)$ and $(r',k-1)$, $i$ has the same
state in $(r',k-1)$ and $(r_1',k-1)$, so must decide 0 in round $k$ of
$r_1'$.  
We now get the desired contradiction by observing that the decision rule 
for deciding 0 in $\kbp^0$ requires that there exists an agent with an
initial preference 0. Note that this is because, as shown in
Proposition~\ref{prop:kbp0correct}, Validity holds even for faulty agents.

Suppose that $i$ decides 1 in round $k$ of $r'$, does not decide 1 in
round $k$ of  
$r$, and does not decide earlier than round $k$ in $r$.  Then we must
have that
$\I_{\gamma,P}, (r,k-1) \models
\neg K_i(\bigwedge_{j \in \Agents} 
\neg (\deciding_j = 0))$:
If $i$ decides 0 in round 
$k$,
then clearly $i$ considers it possible
that some agent is deciding 0; and if $i$ does not decide 0 and does
decide 1 (the only other possibility), then this
formula must also hold.
\commentout{
\roncomment{We don't get that conclusion until we have ruled out that
  $i$ is not deciding 0 in round $k$ or $r$, which it  
might well be doing since it differs from the decision being made in $r'$. 

For this, suppose that $i$ is deciding 0. Then either $\I,(r,k-1) \models \init_i=0 \lor K_i(\bigvee_j \jdecided_j=0)$. 
We can't have that $\I,(r,k-1) \models \init_i=0$ since then we must have $k=1$, and $P'$ would also have to decide 0, 
since no 1-decision can be made in round 1. Thus, we must have $\I,(r,k-1) \models K_i(\bigvee_j \jdecided_j=0)$. 
This implies also $\I',(r',k-1) \models K_i(\bigvee_j \jdecided_j=0)$ by equivalence to time $k-1$. 
However, I don't see how to get a contradiction 
here. I suspect an additional safety condition is required: 
$\I \models K_i(\bigvee_j \jdecided_j=0) \implies \neg K_i \neg (i\in \N \land \bigvee_{j\in \N} \jdecided_j=0)$. 
That allows us to reason to 
$\I',(r',k-1) \models  \neg K_i \neg (i\in \N \land \bigvee_{j\in \N} \jdecided_j=0)$, 
which contradicts the fact that $i$ is deciding 1 in $r'$. 

(In my proof way back when I was using this: 
$ (\decided_i = \bot \land 
\Time >0 \land  K_i \exists 0) \rimp \neg K_i \neg (i \in \N \land \bigvee_{j\in \N} \ominus
(\decided_j = \bot \land K_j \exists 0))$.) 

The purpose of the argumentation about bounds on $k$ in the following is not clear to me, since it
does not occur in the safety condition as stated - it may be a vestige from my arguments.  
} 
}
Since $\kbp^0$ is safe with respect to $\gamma$, 
it follows that there
is a run $r_1$ of $P$ such that $r_i(k-1) = (r_1)_i(k-1)$, $i$ is nonfaulty
in $r_1$, and there exists a nonfaulty agent $j$ that decides 0 in
round $k$ of $r_1$.  Again, let 
$r_1'$ be the run of $P'$ corresponding to $r_1$.  As above,
$i$ has the same
state in $(r',k-1)$ and $(r_1',k-1)$, so must decide 1 in round $k$ of $r_1'$.
Since $P'$ dominates $P$, $j$ must also decide in round $k$ of $r_1'$,
and must decide 1 (since $i$ is deciding 1).  

\commentout{
\roncomment{It seems to me that there is still an issue with this
  proof, in the last paragraph.  
I still don't get the reason for worrying about $k>2$: safety does not mention any bound (it once did). 
More significantly, why does it follow that  $j$ must decide 0 in
round $k$ of $r_1'$?  
Agent $j$ is running a different protocol in $r_1'$, so it is not enough that it has the same local state. 
We can't use the first case above since that goes from $P'$ deciding 0 to $P$ deciding 0, and we need the 
other direction here. I suspect we may need an additional condition like validity of 
$$ K_i (\bigvee_j \jdecided_j = 0) \rimp K_i ( \bigvee_j (\jdecided_j = 0\land \neg K_i \neg (j\in \N))) $$ 
as part of the definition of safety. Then 
we can conclude that $ \neg K_i ( \neg \bigvee_{j'} (\jdecided_{j'} = 0 \land {j'}\in \N))$ at $(r_1,k-1)$. 
By minimality of $k$, the runs $r_1,r_1'$  are also identical to to time $k-1$, so also 
 $ \neg K_i ( \neg \bigvee_j (\jdecided_{j'} = 0 \land j'\in \N))$ also holds at $(r_1',k-1)$. 
 This means that $j$ cannot decide 1 at $(r_1',k-1)$. 
 Since $P' \leq P$ and $j$ decides at $(r_1,k-1)$, also $j$ decides at $(r_1',k-1)$, so it must decide 0. 
That now  gives the desired contradiction. 

This argument requires amending the definition of safety and checking that the protocols satisfy the additional property. 
}
}
If $k=1$, since $j$ decides 0 in $r_1$, $j$ must have an initial
value of 0 in both $r_1$ and $r_1'$.  Let $r_1''$ be the run of $P'$ where all
agents have initial value 0 and are nonfaulty.  Since $(r_1'')_j(0)
= (r_1')_j(0)$, $j$ must decide 1 with $P'$ in run $r_1''$, giving us
the desired contradiction.  If  $k > 1$, by the second half of
condition 2(c) of the safety 
condition, there exists a run $r_2$ of $P$ and an agent $j'$ such that
$(r_2)_j(k-1) = (r_1)_j(k-1)$, $j$ and $j'$ are nonfaulty in $r_2$, and
$j'$ decides 0 in round $k-1$ of $r_2$.  Let $r_2'$ be the run
of $P'$ corresponding to $r_2$.  As observed above, $(r_2)_j(k-1) = 
(r_2)'_j(k-1)$.  Thus, $(r_1)'_j(k-1) = (r_2)'_j(k-1)$, so $j$ must
decide 1 in round $k$ of $r_2'$.  
Since $j'$ is nonfaulty and decides in round $k-1$ of $r_2$, 
and $P'$ dominates $P$, agent $j'$ must also decide at or before round $k-1$ of
$r_2'$.   Since, by construction, round $k$ is the earliest
round that $P$ and $P'$ reach different decisions in corresonding
runs, $j'$ must decide 0 
in round $k-1$ of $r_2'$.  But this means that $j$ and $j'$ make
different decisions in $r_2'$, despite both being nonfaulty.  This
gives  us the desired contradiction.
\end{proof}

\repro{pro:optimal}
$\kbp^0$ is safe with respect to all contexts $\gamma_{\min,n,t}$ and
$\gamma_{\basic,n,t}$ such that $n - t \ge 2$.
\erepro
\begin{proof}

  We do the argument simultaneously  for $\gamma_{\min,n,t}$ and
  and $\gamma_{\basic,n,t}$.
  Let $P$ be an implementation of
  $\kbp^0$ in $\gamma_{\min,n,t}$ (resp., $\gamma_{\basic,n,t}$) and
  let $r$ be a run in  $\I_{\gamma_{\min,n,t},P}$ (resp., $\I_{\gamma_{\basic,n,t},P}$).  
  We first show that the first part of the safety condition holds. 
  Suppose that $i$ has not received a 0-chain by $(r,m)$.
  We want to
  show that there exists a 
      point $(r',m)$ such that $r_i(m)= r'_i(m)$ and
      all agents have initial preference 1 in $r'$.
Let $r'$ be the run where all agents start with initial preference 1 
and the adversary is the same as in $r$.  
An easy argument by induction on $k$, using the fact that the
failure pattern 
is the same in $r$ and $r'$, shows that that for all agents $j$ and
times $k$, if $j$ 
has not received a 0-chain by $(r,k)$
then $r_j(k) = r'_j(k)$.
It immediately follows that $r_i(m) =
r'_i(m)$, which completes the proof of the first part of the argument.

To prove that the second part of the safety condition holds,
  suppose that $\I, (r,m) \models 
  \neg 
  (K_i(\bigwedge_{j \in \Agents}$ $
  \neg (\deciding_j = 0))$ and $i$ does not 
  decide before 
  round $m+1$ of $r$.
  The second part of the safety condition is
  easily seen to hold if $m=0$, so we assume that $m \ge 1$.  
  This implies that $\init_i = 1$, otherwise agent $i$ would have already decided
  in round 1.  
  There
  must exist a
  point $(r^+,m)$ such that $r_i(m)= r^+_i(m)$
  and some agent $j$ decides 0 in
  round $m+1$ of $r^+$.  
  Since $n - t \ge 2$, there must be
  agents, say $i'$ and $j'$, that are nonfaulty in $r^+$,
  where we can take $i= i'$ if $i$ is nonfaulty in $r^+$, and $j=j'$ if $j$ is
  nonfaulty in $r^+$.  
  If $m>1$ then the initial preferences of 
  $i'$ and $j'$ in $r^+$ are 1, for otherwise $i$ would decide 0
  at or before round 2 in $r^+$, and hence also in $r$, contradicting
  the assumption $i$ does not decide before round $m+1$ in $r$.
   The initial
    preference of $j$ must also be 1, for otherwise $j$ decides in
    round $1 \ne m+1$, given that $m \ge 1$.

  Observe that in $r^+$, (a) no agent decides 1 at or before round $m+1$    
(for if agent $j''$ decides 1 in round $k$ of $r^+$, then $\I,(r^+,k-1)
\models K_{j''} (\bigvee_{j''' \in \Agents} \neg (\deciding_{j'''} =
0))$, and this contradicts the fact that $j$ decides 0 in round $m+1$ of
$r$, so  there must be a 0-chain of length $m$ ending with $j$ in $r$), (b)
no nonfaulty agent decides 0 before round $m$ in $r^+$ (otherwise $i$ and $j$
would decide 0 at or before round $m$ in $r^+$), (c) $i$ and $j$ do
not decide 0 at or before round $m$ in $r^+$, and (d) $i$, $j$, and all
the nonfaulty agents send no message (i.e., send $\bot$) in
$\gamma_{\min,n,t}$ and send $(\init,1)$ in $\gamma_{\basic,n,t}$
up to and including round $m-1$; $i$ and $j$ also send $\bot$ (resp.,
$(\init,1)$) in round $m$ of $r$ in $\gamma_{\min}(n)$
(resp., ($\gamma_{\basic,n,t}$).  Note that $i'$ or $j'$ may decide 0 in
round $m$ of $r^+$, in which case they will send 0; otherwise, like $i$
and $j$, they send $\bot$ (resp., $(\init,1)$) in round $m$ of $r$ in
$\gamma_{\min,n,t}$ (resp., ($\gamma_{\basic,n,t}$). 

We want to modify $r^+$ to get a run $r'$ such
  that (a) $r_i(m)= r'_i(m)$, (b) $i$ and $j$ are nonfaulty in $r'$, 
  and (c) $j$ decides 0 in round $m+1$ of $r'$.  Let
  $\alpha=(\N',F')$ be the adversary in $r'$.  We define $r'$ by 
  assuming that all agents have the same initial preferences in
  $r'$ as in $r^+$, and
  the adversary $(\N'',F'')$ in $r'$ is defined as follows:
  $\N'' = \N' - \{i',j'\} \cup \{i,j\}$ (so that $i$ and $j$ are
  nonfaulty in $r'$)
  and, roughly speaking, $F''$ 
  interchange the failures of $i$ and $i'$ and of $j$ and $j'$ in
  $r^+$ and $r'$.   
  More precisely, for all agents $j''$, (a) if
  $j''$ does not 
  receive a message from $i$ (resp., $j$) in round 
  $k$
  of $r^+$
  according to $F'$, and, in the special case that $k=m$, neither
  $i'$ nor $j'$ sends the message 0,  then $j''$ does not
  receive a message from $i'$ (resp., $j'$) in round $k$ of $r'$
  according to $F''$; (b) if $j''$ receives a message from $i$ (resp.,
  $j$) in round 
  of $r^+$ according to $F'$ or if $k=m$ and either
  $i'$ or $j'$ send the message 0 in round $m$ of $r$,
  then $j''$
  receives a message from $i'$ (resp., $j'$) in round $k$ of $r'$
  according to $F''$.

  We claim that for all $k \le m$, all agents $j''$ have the same state
  at time $k$ of runs $r^+$ and $r'$.  We prove this by induction on
  $k$.  In the case that $k=0$, this is immediate since all agents
  have the same initial preferences in $r^+$ and $r'$.  If $0 < k < m$,
  this follows from the fact that all agents have the same state at
  time $k-1$, and the only way in which the runs differ is that
  $j''$ does not receive a message from $i$ (resp., $j$) in round
  $k$ of $r^+$, then $j''$ does not receive a message from $i'$
  (resp., $j'$) in round $k$ of $r'$.  These message are either
  $\bot$ (if the context is $\gamma_{\min,n,t}$) or $(\init,1)$ (if the
  context is $\gamma_{\basic,n,t}$).  But clearly this difference does
  not 
  affect
  the state of $j''$; in particular, if the context is
  $\gamma_{\basic,n,t}$, then $j''$ gets the same number of messages of
  the form $(\init,1)$ in both cases, and this is all it keeps track
  of in its state.  The same argument applies if $k=m$ and neither
  $i'$ or $j'$ send the message 0 in round $m$ of $r^+$.  If $i'$ or
  $j'$ do send the message 0 in round $m$ of $r^+$, since they are
  nonfaulty in $r^+$, all agents will get the message.  By
  construction, they will also send this message in round $m$ of
  $r'$ and all agents will get it in $r'$.  The transition
  function then guarantees that all agents will have the same state
  in round $m$ of $r^+$ and $r'$.  In particular, it follows that
  $r^+_i(m) = r'_i(m)$ and $j$ decides 0 in round $m+1$ of $r'$, as desired.

  Finally, if $m \ge 1$, we must construct a run $r''$ as required in
  the second part of condition 2(c).  Since $j$ decides 0 in $r'$ in
  round $m+1$, as 
  shown earlier, $j$ must receive a 0-chain at time $m$ in $r$.  Let
  $j'$ be the last agent on that chain.  Thus, $j'$ decides 0 in round
  $m$ in $r'$.  If $j'$ is nonfaulty in $r'$, we are done.  If $j'$ is
  faulty, then we consider two cases.  If $m = 1$, then we must have
  $\init_{j'} = 0$.  We consider a run $r''$ where all agents have the
  same initial values as in $r'$, and if $(N',F')$ is the adversary in
  $r'$, then the adversary in $r''$ is $(N'-\{j'\},F'')$, where
  $F''$ agrees with $F'$ on all the agents in $N'-\{j'\}$.  It is easy to see
  that $r'_j(1) = r''_j(1)$, completing the argument.  If $m > 1$,
  then there must be some nonfaulty agent $j''$ other than $i$, since
  $n-t \ge 2$.  As in the argument above, the initial values of $j'$
  and $j''$ must be 1 (otherwise $j'$ would have decided in round 1
  and $i$ would have decided in round 2 in $r'$).  
  We now proceed much as in the previous argument to
  construct $r''$: all agents have the same initial values in
  $r'$ and $r''$,  and we take the  the adversary in $r''$ to be $(N' \cup
  \{j''\} - \{j'\},F'')$, where  $F''$ interchanges the roles
  of $j'$ and $j''$.  We leave details to the reader.
\end{proof}

\rethm{thm:ImplementsSimple}
If $t\leq n-2$ then $P^{min}$ implements $\kbp^0$ in the 
EBA context~$\gamma_{\min,n,t}$. 
\erethm
\begin{proof}
Let $\I$ be the system $\I_{ \gamma_{\min,n,t},\kbp^0}$.  
We show that for all runs $r$ and times $m$, we have $P^{min}_i(r_i(m)) =
(\kbp^0_i)^\I(r_i(m))$.  
\begin{itemize}
  \item If $P^{min}_i(r_i(m)) = \noop$ because 
$\decided_i \neq \bot$ in $r_i(m)$, we clearly also have
$(\kbp^0_i)^\I(r_i(m)) = 
\noop$, because $\I,(r,m) \models K_i(\decided_i \neq \bot)$.
\item If $P^{min}_i(r_i(m)) = \decide_i(0)$, we must have
$\decided_i = \bot$  in $r_i(m)$, 
and either $\init_i =0$ or 
$\rd_i =0$.
If $\init_i =0$ in $r_i(m)$, then $\I,(r,m)\models K_i(\init_i = 0)$,
so  $(\kbp^0_i)^\I(r_i(m)) = \decide_i(0)$.  
If $
\rd_i =0$ and $\init_i \ne 0$ in $r_i(m)$, 
then we must have $\Time_i > 0$ in $r_i(m)$ and 
$\rd_i =\bot$ in $r_i(m')$ for $m' < m$,
for otherwise,
agent $i$ would have decided $0$ earlier and we would have
$\decided_i \neq\bot$ in $r_i(m)$.  Moreover, $i$ must have
received the message 0 from some agent $j$.
Thus, $\I,(r,m) \models K_i(\bigvee_{j \in \Agents} 
\jdecided_j = 0)$.  
\commentout{
Finally, for the agent $j$ from which
$i$ received the message $0$, the same construction as in the proof of
Theorem~\ref{thm:safe} shows that there exists a run $r'$ such that
$r_i(m) = r'_i(m)$, 
$j$ decides 0 in $r'$, and $j$ is nonfaulty in $r'$.  Thus,
$\I,(r,m) \models K_i(\bigvee_{j \in \Agents} (
\jdecided_j= 0 \land \neg K_i \neg (j \in \N \land
\jdecided_j = 0)))$.  
} %
It follows that $(\kbp^0_i)^\I(r_i(m)) = \decide_i(0)$.
\item  If $P^{min}_i(r_i(m)) = \decide_i(1)$, we must have 
 $\decided_i = \bot$, $\init_i = 1$, $ \rd_i \neq 0$, and
$\Time_i = t+1$ in $r_i(m)$.  
As shown in the argument for Termination in the proof of
Proposition~\ref{prop:kbp0correct}, we have
$\I,(r,m) \models K_{i}(\bigwedge_{j \in \Agents}
     \neg (\jdecided_j = 0))$.  Thus,
    $(\kbp^0_i)^\I(r_i(m)) = \decide_i(1)$.  
   \item Finally, if $P^{min}_i(r_i(m)) = \noop$ by the final line 
  of
    $P^{min}_i$,  
  we must have $\Time_i <t+1$, $\init_i =1$, $ \decided_i = \bot$, and
 $\rd_i = \bot$ in $r_i(m)$.  (Note that $t' = m$.) Consider a run
$r'$ where all 
agents have initial preference $1$ and are nonfaulty.  It is easy to
see that $i$ receives the same messages up to time $m$ in $r$ and $r'$, so
$r_i(m) = r'_i(m)$.  
Hence, we must have $\I,(r,m) \models \neg K_i(\init_i=0 \lor
\bigvee_{j \in \Agents} \decided_j=0)$.  Thus, $(\kbp^0_i)^\I(r_i(m)) \ne
\decide_i(0)$.  It is also not hard to construct a run $r''$
that $i$ considers possible where  
there are exactly $t'$ faulty agents such that some nonfaulty agent
$j$ in $r''$ gets a 0-chain of length
$t'$ in round $t'$ of $r''$.   Thus,
$\I,(r,m) \models \neg K_{i}(\bigwedge_{j \in \Agents}
\neg (\deciding_j = 0))$, so
  $(\kbp^0_i)^\I(r_i(m)) \ne \decide_i(1)$. Therefore,
  $(\kbp^0_i)^\I(r_i(m)) = \bot$.
\end{itemize}
It follows that $P^{min}$ implements $\kbp^0$ in
$\gamma_{\min,n,t}$. 
\end{proof}

\rethm{thm:implement2}
  If $t\leq n-2$, then $P^{\basic}$ implements $\kbp^0$ in the
    EBA context~$\gamma_{\basic,n,t}$. 
\erethm
\begin{proof}
We proceed just as in the previous argument.
If $P^{\basic}_i(r_i(m)) = \bot$ or $P^{\basic}_i(r_i(m)) =
\decide_i(0)$, then the 
argument is identical to that for $P^{min}$.  
If $P^{\basic}_i(r_i(m)) = \decide_i(1)$, then 
we proceed by induction to show that when 
we have $\decided_i = \bot$,
$\init_i=1$, and either 
$\rd_i = 1$ or $\#1_i > n - m$ in
$r_i(m)$, then
we also have $\I, (r,m) \models K_i(\bigwedge_{j \in \Agents} \neg ( \deciding_j = 0))$. 

So suppose that $\decided_i = \bot$ and $\init_i=1$ are in $r_i(m)$.  
If $\#1_i > n-m$ in $r_i(m)$, then it is easy to see
that there cannot be a 0-chain of length $m$ in $r$ (since the only
agents that can be involved in this 0-chain are ones that did not
send an $(\init,1)$ message). Thus, $\I, (r,m) \models 
K_i(\bigwedge_{j \in \Agents} \neg ( \deciding_j = 0))$.
On the other hand, if $\rd_i = 1$ is in $r_i(m)$, then 
we must have that 
$\rd_i =\bot$ in $r_i(m')$ for $m' < m$, for otherwise,
agent $i$ would have decided $1$ earlier and we would have
$\decided_i \neq\bot$ in $r_i(m)$.  Moreover, $i$ must have
received the message 1 from some agent $j$ in round $m$ of $r$.
Thus, $j$ decides 1 in round $m$ of $r$, so we must have 
$\decided_j = \bot$, $\init_j=1$, and either $\rd_j = 1$ or $\#1_j > n - m$
in $r_j(m-1)$. By the inductive hypothesis, $\I, (r,m-1) \models 
K_j(\bigwedge_{j' \in \Agents} \neg ( \deciding_{j'} = 0))$.  If some
agent $j'$ decides 0 in round $m+1$ of $r$, then there must be a 0-chain 
that ends  0 with $j'$, so $j'$ must get a message from an agent that 
decides 0 in round $m$, contradicting the fact that $\I, (r,m-1) \models 
K_j(\bigwedge_{j' \in \Agents} \neg ( \deciding_{j'} = 0))$. 
It thus follows from the information in $r_(m)$ that no agent decides
0 in round $m+1$ of $r$, so
$\I, (r,m) \models K_i(\bigwedge_{j \in \Agents} \neg ( \deciding_j = 0))$.
\commentout{
If $\rd_i = 1$, then arguments similar to those used in the
case that $P^{min}(r_i(m)) = \decide_i(0)$ show that
$\I, (r,m) \models  K_i   (\bigvee_{j \in \Agents} (
\jdecided_j = 1 \land \neg K_i \neg (j \in \N 
     \land \jdecided_j = 1))) $.  
     
If $\rd_i \ne 1$ and $\#1_i > n - m$, then it is easy to see
  that there cannot be a 0-chain of length $m$  (since the only
  agents that can be involved in this 0-chain are ones that did not
  send an $(\init,1)$ message), let alone one that ends with an
  honest agent.  Thus,
  $\I, (r,m) \models  K_i(\bigwedge_{j \in \Agents} 
          \neg ( \jdecided_j = 0))$.
} %
In either case, we have $(\kbp^0_i)^\I(r_i(m)) = \decide_i(1)$.

Finally, if $P^{\basic}_i(r_i(m)) = \noop$ by the final line of
$P^{\basic}_i$, then arguments similar in spirit to those used above
show that agent $i$ considers it possible that all agents started with
an initial preference of 1, and hence does not know that there is a
0-chain, but also cannot rule out the possibility of a 0-chain.
We must have $\decided_i = \rd_i = \bot$, $\init_i = 1$, and 
$\#1_i \leq n - m$ in $r_i(m)$. Since $\init_i = 1$ and $\rd_i = \bot$
are in $r_i(m)$, at the point $(r,m)$, 
  $i$ considers possible the run $r'$ where every agent started with an initial
  preference of 1 and the message pattern is identical to $r$. We then have
  $\I,(r,m) \models \neg K_{i}(\bigwedge_{j \in \Agents}
  (\jdecided_j = 0 
    ))$, so $(\kbp^0_i)^\I(r_i(m)) \ne \decide_i(0)$.
  Similarly, since $\#1_i \leq n - m$, $i$ considers it possible that a
there is a 0-chain consisting of the agents that $i$ didn't hear a 1
from.  
    Hence, 
        $\I,(r,m) \models \neg K_i
    (\bigwedge_{j \in \Agents} \neg (\deciding_j=0))$, so 
$(\kbp^0_i)^\I(r_i(m)) \ne \decide_i(1)$.
  It follows that     $(\kbp^0_i)^\I(r_i(m)) =     \bot$.
Thus, $P^{\basic}$ implements $\kbp^0$ in
$\gamma_{\basic,n,t}$.  
\end{proof}

\subsection{Proofs for Section~\ref{sec:fip-optimal}}

\subsubsection{Motivating the KBP}

\repro{prop:cfaulty-necessary}
If $P$ is an optimal protocol in 
the context $\gamma_{\fip,n,t}$ and $\I_{P,\gamma_{\fip,n,t}} ,(r,m) \models
  \decided_i = \bot \land K_i(\ck{\N}(\faultyag \land \nodecided_\N(1) \land
 \exists 0))$, then all undecided agents in $\N(r)$ make a
  decision in round $m+1$, and 
  similarly if $\I_{P,\gamma} ,(r,m) \models \decided_i = \bot \land
  K_i(\ck{\N}(\faultyag \land \nodecided_\N(0) \land \exists 1))$.
\erepro
\begin{proof} 
  For ease of exposition, let $\gamma = \gamma_{\fip,n,t}$.
  Suppose that
$\I_{P,\gamma} ,(r,m) \models \decided_i = \bot \land
    K_i(\ck{\N}(\faultyag \land \nodecided_\N(1) \land \exists 0))$.
Let $P'$ be the protocol for context $\gamma$ obtained by modifying $P$ as follows. 
For each agent $i$ and local state $s\in L_i$, we define $P'_i(s)$ to be 
$\decide_i(0)$ if 
$\I_{P,\gamma} ,(r,m) \models K_i(\ck{\N}(\faultyag \land \nodecided_\N(1) \land \exists 0))$ 
for points $(r,m)$ with $r_i(m) = s$, and $P'_i(s) = P_i(s)$ otherwise. We claim that 
$P' \leq P$ and $P'$ is an EBA-protocol in context $\gamma$. 

Let $r$ and $r'$ be corresponding runs of $\I_{P, \gamma}$ and $\I_{P',\gamma}$, respectively. 
We show by induction on $k$ that $r(k) = r'(k)$ if $k$ is less than or
equal to the earliest time $m$ such that
$\I_{P,\gamma}, (r,m) \models \ck{\N}(\faultyag \land \nodecided_\N(1) \land \exists 0)$. 
The base case is trivial, and the inductive case follows
the fact that if  
$r(k) = r'(k)$ and  $\I_{P,\gamma}, (r,k) \not\models \ck{\N}(\faultyag
\land \nodecided_\N(1) \land \exists 0))$,  
then for all agents $i$, $\I_{P,\gamma}, (r,k) \not\models K_i
\ck{\N}(\faultyag \land \nodecided_\N(1) \land \exists 0)$
(the fact that $\I_{P,\gamma}, (r,k) \not\models K_i
\ck{\N}(\faultyag \land \nodecided_\N(1) \land \exists 0)$ for all
agents $i$ follows from Lemma~\ref{lemma:cfaultyimpliesall}, proved below),
so $P_i(r_i(k)) = P'_i(r'_i(k))$ for all agents $i$. Since the
failure patterns are the same in these runs,  
it follows that $r(k+1) = r'(k+1)$. 
Moreover, since $\I \models (\ck{\N}(\phi) \land i \in \N ) \rimp
K_i(\ck{\N}(\phi))$ (see \cite{FHMV}),
it follows that once $\I_{P,\gamma}, (r,m) \models \ck{\N}(\faultyag
\land \nodecided_\N(1) \land \exists 0)$,  
all undecided agents in $\N(r')$ decide 0 simultaneously using
$P'$.  

It is immediate from these facts that $P'\leq_\gamma P$. If a
nonfaulty agent $i$
decides in $r'$ at a time $m$ before 
the common knowledge condition has become true in $r$, then $r(m) = r'(m')$ and 
agent $i$ makes the same decision at time $m$ in $r$. Once the common knowledge condition 
becomes true, an undecided nonfaulty agent $i$ decides using $P'$, 
so does so at least as soon as it does using $P$. 
Recall that if $i$ is faulty, then  the definition of $P'\leq_\gamma P$ allows that agent $i$ decides using $P$ before
it does so using $P'$, so we do not need to consider this case. 

Next, we show that $P'$ is an EBA protocol in context $\gamma$. 
Unique Decision follows from the fact that the context $\gamma$ records 
decisions in the local state, and $P_i(s) = \bot$ for states $s$ that record that a decision has already been made.  
For Validity, consider runs $r'$ of $P'$ and the corresponding
run $r$ of $P$ in context $\gamma$. 
If the common knowledge condition has not yet become true at a point where nonfaulty agent $i$ makes its
decision on value $v$ using $P'$, then $r(m) = r'(m)$, and $i$ makes the same decision using $P$.
It follows from Validity for $P$ that some agent has initial value $v$. Alternately, if nonfaulty agent
$i$ decides $0$ at $(r',m)$ because it knows at $(r,m)$ in 
$\I_{P,\gamma}$ that the common knowledge condition has become true, 
then in fact $\I_{P,\gamma},(r,m) \models \exists 0$, and the same fact holds at $(r',m)$. 

For Agreement,  consider a run $r'$ of $P'$ where agent $i \in
\N(r')$ decides $0$ in round $m_0+1$ and agent $j \in \N(r')$ decides
$1$ in round $m_1+1$.  
If the common knowledge condition has not become true in the corresponding run $r$ of $P$
by time  $\max(m_0,m_1)$, then since $r$ and $r'$ are
identical up to at least this time,
we have a contradiction to  Agreement for $P$. 
If the common knowledge condition becomes true at time $k\leq m_1$, then 
according to $P'$, all nonfaulty agents undecided by time $k$ decide 0
in round $k+1$.
That means that agent $j$, which is undecided at time $k$, decides 0
in round $k$ of $r'$ and decides on a different value in round
$m_1+1$, contradicting 
the fact that, as we have shown, $P'$ satisfies Unique Decision.
On the other hand, the common knowledge condition cannot become true after time $m_1$, since it 
implies $\nodecided_\N(1)$ and  $j\in \N(r')$ has decided
1 in round $m_1+1$.  
Thus, $P'$ satisfies Agreeement$(\N)$. 

Finally, Termination for $P'$ follows from $P'\leq_\gamma P$ and the fact
that $P$ satisfies Termination.  

Since we have $P'\leq P$ and $P'$ is an EBA protocol for context $\gamma$, it follows that $P\leq P'$, 
so in $\I_{P, \gamma}$, all undecided nonfaulty agents make a decision as soon as they know that the common knowledge condition 
holds. 
\end{proof} 

Note that the result does not tell us \emph{what} decision the nonfaulty agents must make when the common knowledge condition holds. There may be situations where no nonfaulty agent has made a decision, and both common knowledge conditions hold. 
Either a 0 or a 1 decision would be acceptable in this case.

We now
prove all the required properties of $\kbp^1$.  
We start by examining what each of the conditions is $\kbp^1$ tells
us.

\subsubsection{The common knowledge conditions}

As a first step to understanding $\kbp^1$, we characterize the common
knowledge conditions (i.e., the second and third lines).  The next
proposition gives necessary and sufficient conditions for these
conditions to hold.  These conditions show that the problem of computing
when the common knowledge conditions hold is tractable.

Define 
 $\dist_{\N}(\faultyag)$ to be an abbreviation for 
          $$\exists A \subseteq \Agents(|A| = t \land \forall i \in A\,
\exists j \in \N\, (K_{j}(i \not\in\N))).$$
Roughly speaking,  $\dist_{\N}(\faultyag)$ holds if, between them, the
nonfaulty agents know about $t$ faulty agents.

  \begin{definition}[\emph{hears-from}]
      We define a \emph{one-step hears-from} relation in a run $r$ 
      on pairs $(j, m)$ consisting of agents $j$ and times $m$ by saying that 
      $(j', m')$ \emph{one-step hears from} $(j, m)$ in $r$ if 
      agent $j$ sends $j'$ a non-$\bot$ message in round $m + 1$ of $r$
        that $j''$ receives and $m+1 \le m'$.
       The \emph{hears-from} relation is the  
       transitive closure of the \emph{one-step hears-from} relation.
       We write $(j,m) \lamport{r} (j',m')$ if $(j',m')$ \emph{hears-from} $(j,m)$
  \end{definition}
Those familiar with the \emph{Lamport causality} relation \cite{Lamclocks}
will recognize that the \emph{hears-from} relation is similar in spirit.

  \begin{proposition} \label{lemma:dfaulty}
    For all implementations $P$ of the knowledge-based program $\kbp^1$
    with respect to $\gamma_{\fip,n,t}$, 
    \begin{itemize}
      \item[(a)]
    $\I_{\gamma_{\fip,n,t},P}
        \models \Time > 0 \rimp (\ominus \dist_\N(\faultyag)
        \Leftrightarrow C_\N(\faultyag))$. 
      \item[(b)]
        $\I_{\gamma_{\fip,n,t},P} \models \Time > 0 \rimp (\nodecided_\N(v) \rimp
        \land_{j \in \N} \ominus (K_j \Circ \nodecided_j(v)))$ and  
        $\I_{\gamma_{\fip,n,t},P} \models ( C_{\N}(\faultyag) \land
    (\nodecided_\N(v))) \dimp   C_\N (\faultyag \land \nodecided_\N(v))$,
    for $v 
  \in \{0,1\}$. 
  \item[(c)] $\I_{\gamma_{\fip,n,t},P} \models  \Time > 0 \rimp ((C_{\N}(\faultyag) \land
  \ominus(\lor_{j \in \N} K_j (\exists v))) \dimp C_{\N}(\faultyag \land
    \exists v))$, for $v \in \{0,1\}$.
  \commentout{
\item[(d)] For all agents $i$, $\I_{\gamma_{\fip,n,t},P} \models
  C_{\N}(\faultyag) \rimp K_i (\faultyag \land C_{\N}(\faultyag))$.}
  \end{itemize}
    \end{proposition}
  \begin{proof}
    Let $P$ be an implementation of $\kbp^1$ and let
    $\I = \I_{\gamma_{\fip,n,t},P}$.

For part (a), first suppose by way of contradiction that $\I,(r,m)
\models \neg \ominus  
    \dist_\N(\faultyag) \land C_\N(\faultyag)$. 
    Let $A' = \{i: \exists j \in \N(r)(\I, (r,m-1) \models K_j (i \notin \N))\}$.  By
  assumption,  $|A'| < t$.  There must exist some set $A$ with $|A| =
  t$ such that $\I, (r,m) \models \ck{\N}(\bigwedge_{ i \in   A} (i \notin \N))$.
  Moreover, $A'$ is a strict subset of $A$, since each nonfaulty agent
  in $r$ will learn 
  in round $m$ (if they did not already know it) that each agent in
  $A'$ is faulty (by getting a message from a nonfaulty agent from
  which it can infer this).
  Let $r'$
    be a run where $\N(r') = A'$, all
  agents have the same 
  initial state in $r$ and $r'$, and for all agents in $A'$, 
the same messages are delivered in $r$ and $r'$.
  It is
  easy to check that for all agents in $\N(r)$, we must have $(r,m-1)
  \sim_j (r',m-1)$.  (Formally, we show by induction on
 $k$ that if $j \in \N(r)$ and $k < m' \le m-1$, 
  then $(j',m'-k) \lamport{r} (j,m')$ iff
      $(j',m'-k) \lamport{r'} (j,m')$.
  That is, all agents in $\N(r)$
  consider possible a  run, namely $r'$, where the only nonfaulty
  agents are those in $A'$.
Now even though $i$ may learn about other other nonfaulty agents in
round $m$ of $r$, for all $j \in \N(r)$ such that $j \ne i$, $i$ must
consider it possible that at $(r,m)$, $j$ considers $(r',m)$ possible,
because $j$ hears from all agents other than those in $A'$ in round
$m$.  Thus, $\faultyag$ is not commnon knowledge among the nonfaulty
agents at $(r,m)$.  (Note that here we are using the fact that $n - t
\ge 2$, so that there is an agent $j \in \N(r)$ such that $j \ne i$.)

Conversely, suppose that $\I,(r,m-1) \models \dist_\N(\faultyag)$.
Thus, there exists some set $A$ with $|A| = t$ such that
$\I,(r,m-1) \models \dist_\N(\faultyag_A)$, where $\dist(\faultyag_A)$ is the
formula    $\forall i \in A
\exists j \in \N(K_{j}(i \not\in\N))).$  Moreover, for all runs $r'$,
if $\I,(r',m-1) \models \dist_\N(\faultyag_A)$, then $\N(r) = \Agents - A$.

It is well known (see \cite{FHMV}) that, for all formulas $\phi$ and
$\psi$, if $\I \models \phi \rimp E_\N(\psi \land \phi)$, then
$\I\models \phi \rimp C_\N\psi$.  Thus, 
it suffices to show 
    $\I \models \ominus \dist_\N(\faultyag_A) \Rightarrow E_\N (
\faultyag \wedge  
\ominus \dist_\N(\faultyag_A))$.

For all points $(r',m')$, if $\I,(r',m') 
\models \ominus \dist_\N(\faultyag_A)$, then at $(r',m)$, all the
nonfaulty agents know that the agents in $A$ are faulty and that these
are the only faulty agents.  Since
$|A|=t$,
we have that $\I,(r',m') \models E_\N (\faultyag)$.  Moreover, since all
the nonfaulty agents hear from all the other nonfaulty agents in round
$m$ of $r'$, they all know $\ominus \dist_\N(\faultyag_A)$.  This
completes the proof of part (a).

\commentout{
We next prove some properties 
that follow 
when the identities of faulty
agents become common knowledge among the nonfaulty.
  
\begin{lemma} \label{lemma:cfaulty-nd}
Suppose that $\I$ is an FIP sending omissions interpreted system for a deterministic protocol and $t \leq n-2$.
Then $\I \models ( C_{\N}(\faultyag) \land (\nodecided_\N(v)))
\Rightarrow   C_\N (\nodecided_\N(v))$, for $v \in \{0,1\}$.
  \end{lemma}
  \begin{proof}
}
For part (b), the first part is immediate: if a nonfaulty agent $i$
does not decide on a value $v$ by round $m$, then at time $m-1$, $i$
knows this will be the case. The if direction of the second part is
immediate from the fact that, since 
$\N \ne \emptyset$, $\I \models C_\N \phi \rimp \phi$ for all formulas
$\phi$ \cite{FHMV}.  For the only-if direction, suppose that 
$\I,(r,m) \models C_{\N}(\faultyag) \land \nodecided_\N(v)$.  For
each agent $i \in \N(r)$, as we obseved, $i$ knows at time $m-1$ in
$r$ that it will 
not decide $v$ in round $m$.  Since we are using a full-information
protocol, all the nonfaulty agents will know at time $m$ in $r$ that
$\nodecided_i(v)$ holds.  Thus, $\I,(r,m) \models E_\N(
\nodecided_\N(v))$.  It is a standard property of $C_\N$ that
$\I \models C_\N \phi \rimp E_{\N}C_\N \phi$ \cite{FHMV}.  Thus,
$\I,(r,m) \models 
E_\N(\nodecided_\N(v) \land C_\N \faultyag )$.  We have just shown that
$\I \models C_{\N}(\faultyag) \land \nodecided_\N(v)) \rimp 
E_\N C_{\N}(\faultyag) \land \nodecided_\N(v))$.  It follows that
$\I \models ( C_{\N}(\faultyag) \land (\nodecided_\N(v))) \Rightarrow
C_\N (\nodecided_\N(v))$, as desired.

\commentout{
Consequently, $\I \models \nodecided_\N(v) \rimp \prev \dk{\N}(\nxt \nodecided_\N(v))$. 
Suppose $\I,(r,m) \models \ck{\N}(\N = G)$. Then for all $j \not \in G$ we have 
$\I,(r,m) \models \ck{\N}(j \not \in \N)$, and by Lemma~\ref{lem:CprevD}
we obtain $\I,(r,m-1) \models \dk{\N}(j \not \in \N)$. 
We also have $\I,(r,m) \models K_i(\N = G)$ for $i \in \N(r)$, 
so by Lemma~\ref{lem:count-faulty}, we must have that $kf(r,i,m)= t$, so that the number of  agents $j$ for which 
$\I,(r,m-1) \models \dk{\N}(j \not \in \N)$ is also $t$. 
It follows that in fact $\I,(r,m-1) \models \dk{\N}(\N = G)$. 
With the conclusion above, this means that we have $\I, (r, m-1) \models \dk{\N}(\N = G \land \nxt \nodecided_\N(v) )$. 
It follows that $\I, (r, m) \models \ck{\N}(\N = G \land \nodecided_\N(v) )$ by Lemma~\ref{lem:DnextC}.
 \end{proof}

  \begin{lemma} \label{lemma:cfaulty-ev}
Suppose that $\I$ is an FIP sending omissions interpreted system for a deterministic protocol with $t \leq n-2$.
Then $\I 
\models (C_{\N}(\faultyag) 
    \land 
    \ominus \dk{\N}(\exists v))
    \Leftrightarrow 
    C_{\N}(\faultyag \land \exists v)$,
    where $v \in \{0,1\}$.
  \end{lemma}
  \begin{proof}
}
For part (c), the proof of the only-if direction is similar in spirit
to that of part (b); we leave details to the reader.  For the if
direction, suppose that 
$\I,(r,m) \models C_{\N}(\faultyag)  \land \exists v$ and, by way of
contradiction, $\I, (r,m) \not\models 
\lor_{j \in \N} (\ominus K_j(\exists v))$.
Let $r'$ be a run where all the agents have initial state $1-v$,
$\N(r') = \N(r)$, and the failure pattern is the same in $r$ and $r'$,
except that in round $m$ of $r'$, all agents hear only from the
nonfaulty agents.
We claim that, for all $i \in \N(r)$, we have $(r,m-1) \sim_i
(r',m-1)$.  For clearly, 
 if $j \in \N(r)$ and $k < m' \le m-1$, 
 then $(j',m'-k) \lamport{r} (j,m')$ iff
    $(j',m'-k) \lamport{r'} (j,m')$.  Since $(j,m')$ does not hear
  from any agent with initial value $v$ in $r$, it follows that this
  must also be the case in $r'$.  Now a straightforward induction on
 $m''$ shows that if $j \in \N(r)$, $m'' \le m' < m-1$, and
  $(j',m'') \lamport{r} (j,m')$, then $(r,m'') \sim_{j'} (r',m'')$.  As in
  part (a), even though $i$ may learn $\exists v$ in
round $m$ of $r$, for all $j \in \N(r)$ such that $j \ne i$, $i$ must
consider it possible that at $(r,m)$, $j$ considers $(r',m)$ possible,
so does not learn $\exists v$.
Thus, $\exists v$ is not commnon knowledge among the nonfaulty agents
at $(r,m)$.

\end{proof}

  \commentout{
  the nonfaulty agents in $r$
    By Lemma~\ref{lem:count-faulty} it follows that $|\N(r)| = n-t$.   
    Let $r'$ be a run where $(r,m) \sim_i (r',m)$. 
    Let $G = \Agents \setminus \N(r)$. 
    By Lemma~\ref{lem:CprevD}, we get $\I,(r,m-1) \models \dk{\N}(j \not \in \N)$
    for all 
    $j \in G$. It follows that $\I,(r,m-1) \models \dk{\N}(\N = G)$. 
    Since $i$ receives a message from all nonfaulty agents in round $m$, 
    all nonfaulty agents must have the same local state in both $(r,m-1)$ and 
    $(r',m-1)$. Then 
    we also have 
    $\I,(r',m-1) \models \N = G \land \dk{\N}(\faultyag)$. 
    Therefore 
    using Lemma~\ref{lem:DnextC},  
    we get
    $\I,(r',m) \models C_\N(\faultyag)$
    and hence 
    $\I,(r,m) \models K_i(\faultyag \land C_\N(\faultyag))$.
  \end{proof} 
  }
  
  \begin{lemma} \label{lemma:cfaultyimpliesall}
    For all implementations $P$ of the knowledge-based program $\kbp^1$
    with respect to $\gamma_{\fip,n,t}$, 
    if $\I_{\gamma_{\fip,n,t},P},(r,m)
    \models C_{\N}(\faultyag \land \nodecided_{\N}(1-v) \land \exists v)$,
    then for all $i \in \Agents$,
    $\I_{\gamma_{\fip,n,t},P},(r,m)
    \models K_i(C_{\N}(\faultyag \land \nodecided_{\N}(1-v) \land
    \exists v))$.
    Moreover,  for all agents $i$,   $\I_{\gamma_{\fip,n,t},P},(r,m)
    \models     C_{\N}(\faultyag)$ $\rimp K_i(C_{\N}(\faultyag)$.
  \end{lemma}
  \begin{proof}
        Let $\I = \I_{\gamma_{\fip,n,t},P}$ and suppose that $\I,(r,m)
    \models C_{\N}(\faultyag \land \nodecided_{\N}(1-v) \land \exists v)$
    for some $v \in \{0,1\}$ and $(r,m) \sim_i (r',m)$. 
\commentout{
    By Lemma~\ref{lem:CNtoKCN}, we have 
$\I,(r,m) \models K_i(\faultyag \land C_\N(\faultyag))$. It follows
that $\N(r')= \N(r)$. 
We then observe that in particular we have $\I,(r,m) \models 
\nodecided_\N(1-v)$. As $i$ receives a message from all nonfaulty agents in 
    every round 
    of $r$ and $\N(r') = \N(r)$, 
    we must also have $\I,(r',m) \models 
    \nodecided_\N(1-v)$. By Lemma~\ref{lemma:cfaulty-nd}, 
    $\I,(r',m) \models C_\N(\nodecided_\N(1-v))$.

    Lastly, we observe that $\I,(r,m) \models C_\N(\faultyag \land \exists v)$.
    Using Lemma~\ref{lemma:cfaulty-ev}, we get $\I,(r,m) \models
    \ominus \dk{\N}(\exists v)$.
    Since $i$ hears from all nonfaulty agents, we also have 
    $\I,(r',m) \models 
    \ominus \dk{\N}(\exists v)$.
    We previously showed that $\I,(r',m) \models C_\N(\faultyag)$.
    Therefore, by Lemma~\ref{lemma:cfaulty-ev}, 
    we can conclude that $\I,(r',m) \models C_\N(\exists v)$ 
    and $\I,(r,m) \models K_i(C_\N(\exists v))$.
}
We clearly must have $m > 0$, since we cannot have common knowledge of
the faulty agents at time 0.
By Proposition~\ref{lemma:dfaulty}, it follows that
$\I, (r,m-1) \models \dist_\N(\faultyag) \land
(\land_{j \in \N} (K_j \Circ \nodecided_j(v))
\land (\lor_{j \in \N} K_j (\exists v))$.  Since the nonfaulty agents
send messages to all agents in round $m$, all agents (not just the
nonfaulty agents) will know at time $m$ that this was true at time
$m-1$.  It now follows from Proposition~\ref{lemma:dfaulty} that, for
all agents $i$, we have
    $\I_{\gamma_{\fip,n,t},P},(r,m)
    \models K_i(C_{\N}(\faultyag \land \nodecided_{\N}(1-v) \land \exists v))$.
  \end{proof}

  We next show that once the nonfaulty agents know about $t$ faulty
  agents, every agent decides by the end of the 
  following round.
  \begin{lemma} \label{lemma:cfaulty}
    For all implementations $P$ of the knowledge-based program $\kbp^1$
    with respect to $\gamma_{\fip,n,t}$, 
    $\I_{\gamma_{\fip,n,t},P} \models C_{\N}(\faultyag) \Rightarrow  
    \bigwedge_{i \in \Agents}\Circ \neg (\decided_{i} = \bot)$.
  \end{lemma}
  \begin{proof}
    Let $\I = \I_{\gamma_{\fip,n,t},P}$ and
    suppose that $i \in \Agents$, $j \in \N$,  and 
    $\I,(r,m) \models
    C_{\N}(\faultyag)$.
    If $j$ decides 0 in round $m' \le m$, then $i$ must decide by
    round $m'+1$, 
    because it will hear from $j$ in round $m'$ that it is deciding,
so that $\I, (r,m') \models K_i(\jdecided_j=0)$.
(Note that $i$ may decide earlier or due to one of the common
knowledge conditions, but it will surely decide by round $m'+1$.)
If $j$ decides 1 at some round $m' \le m$, then either
\begin{itemize} 
\item $\I, (r,m'-1) \models K_j(\ck{\N}( \faultyag \land
  \nodecided_\N(0) \land \exists 1)$, so 
  $\I, (r,m'-1) \models K_i(\ck{\N}( \faultyag \land \nodecided_\N(0)
  \land \exists 1)$ by Lemma~\ref{lemma:cfaultyimpliesall}, and $i$
  must also decide by round $m'$;
\item $\I, (r,m'-1) \models K_j(\bigwedge_{k \in \Agents}
  \neg (\deciding_k= 0))$, so, since $P$ is a full-information protocol,
  $\I, (r,m') \models \prev K_j(\bigwedge_{k \in \Agents} \neg
  (\deciding_k= 0))$.
  Since $P$ is a FIP, $i$ hears from $j$ in round $m'$ that no agent
decides 0 in round $m'$.  Thus, $i$ knows that no agent will decide 0
in round $m'+1$ due to line 4 of $P$.  
\commentout{
It follows that  
  $\I, (r,m') \models K_i (\bigwedge_k (\neg(init_k = 0) \alnd \neg(
  \jdecided_k=0))$, from 
  which it is immediate that
  $$\I, (r,m') \models K_i(\bigwedge_{i' \in \Agents}
  (\bigwedge_{k\in \Agents}  
    \neg K_{i'} ( \jdecided_k=0))).$$
  knows that no other agent $i'$ knows $\justdecided_k=0$ 
     Moreover, no agent $k$ decides 0 in round $m'+1$ because
      $\init_k = 0$.
Thus, $i$ knows at time $m'$ that no agent is
decides 0.}
  If some agent decides 0 due to the common
      knowledge condition in round $m'+1$, then, by the argument above,
      $i$ also decides 0 in round $m'+1$.  If not,
then line 5 of $P$ applies, and
      $i$ decides 1 in round $m'+1$.
    \end{itemize} 

We have just shown that that $i$ decides by round $m+1$ if some nonfaulty
agent $j$ decides by round $m$.  If no nonfaulty agent decides by
round $m$, then $\I,(r,m)
\models \nodecided_\N(0) \land \nodecided_\N(1)$.  We must have $m >
0$ since $C_\N (\faultyag)$ cannot hold at time 0, and clearly
$\I,(r,m-1) \models (\lor_{k\in \N} \exists 0) \lor (\lor_{k\in \N}
\exists 1)$.  It now follows easily from
Proposition~\ref{lemma:dfaulty} that
    $\I,(r,m)    \models C_{\N}(\faultyag \land \nodecided_{\N}(1)
\land \exists 0)$ or 
    $\I,(r,m)    \models C_{\N}(\faultyag \land \nodecided_{\N}(0)
\land \exists 1)$; that is, one of the common knowledge conditions
must hold.  Hence, by Lemma~\ref{lemma:cfaultyimpliesall}, agent $i$
must decide by round $m+1$.
  \end{proof}

  \commentout{
    In the remaining possibilities, 
    $\I,(r,m) \models \nodecided_\N(0,1)$.
    By Lemma~\ref{lemma:cfaulty-nd}, we have $\I,(r,m) \models C_\N 
    (\nodecided_\N(0,1))$. ByLemma~\ref{lemma:cfaulty-ev}, if some nonfaulty 
    agent knows about a 0 at time $m-1$, we get $\I,(r,m) \models C_\N
    (\exists 0)$ and therefore $\I,(r,m) 
    \models C_{\N}(\faultyag \land \nodecided_{\N}(1) \land \exists
  0)$. If no nonfaulty agent knows about a 0 at time $m-1$, then all
  nonfaulty agents
    must know that there exists a 1 
    at time $m-1$. Therefore, by the second claim, $\I,(r,m) \models 
    C_\N(\exists 1)$ and we get $\I,(r,m) 
    \models C_{\N}(\faultyag \land \nodecided_{\N}(0) \land \exists 1)$.
    Hence, in either case we can apply Lemma~\ref{lemma:cfaultyimpliesall}
    to conclude that $i$ would decide with one of the common knowledge 
    conditions.
  \end{proof}
}
  
\commentout{
The following well-known result \cite{FHMV} will also prove useful in
the sequel:

\begin{lemma}\label{lem:CNprop} The following are equivalent, for all formulas $\phi$:
  \begin{itemize}
    \item $\I,(r,m) \models C_\N \phi$
    \item for all agents $i \in \N(r)$, $\I,(r,m) \models K_i C_\N \phi$ 
    \item for some agent $i \in \N(r)$, $\I,(r,m) \models K_i C_\N
      \phi$.
\end{itemize}
  \end{lemma}      
}

\subsubsection{Characterizing the condition for deciding 0}

  \begin{lemma} \label{lemma:chain}
    For all implementations $P$ of the knowledge-based program $\kbp^1$
       with respect to $\gamma_{\fip,n,t}$,
    if $\I_{\gamma_{\fip,n,t},P},(r,m)
    \models \neg K_{i}(C_{\N}(\faultyag \land \nodecided_{\N}(1) \land \exists 0))
    \land \neg K_{i}(C_{\N}(\faultyag \land \nodecided_{\N}(0) \land \exists 1))$,
        then agent $i$ receives a \zchain 
        in round $m$
    if and only if 
    $\I_{\gamma_{\fip,n,t},P},(r,m) \models \init_{i} = 0 
    \lor \bigvee_{j \in \Agents} K_{i}(\jdecided_{j}=0)$ and $i$ has not 
    decided 
    before round $m+1$.
  \end{lemma}
  \begin{proof}
    We proceed by induction on $m$. Let $\I = \I_{\gamma_{\fip,n,t},P}$ and 
    suppose that $\I,(r,m)
    \models \neg K_{i}($ $C_{\N}(\faultyag \land \nodecided_{\N}(1) \land \exists 0))
    \land \neg K_{i}(C_{\N}(\faultyag \land \nodecided_{\N}(0) \land \exists 1))$.
    By Lemma~\ref{lemma:cfaultyimpliesall}, we then have
    $\I,(r,m)
    \models \neg C_{\N}(\faultyag \land \nodecided_{\N}(1) \land \exists 0)
    \land \neg C_{\N}(\faultyag \land \nodecided_{\N}(0) \land \exists 1)$.    

    If $m = 0$, then the only-if direction follows from the fact that 
        agent $i$ receives
    a 
    \zchain in round 0 iff $\I,(r,0) \models
    \init_i = 0$.  
    Obviously, $i$ has not decided at time 0. 
    For the 
    converse, we first observe that $\I,(r,0) \models \neg K_i(\jdecided_j = 0)$
    for all $j \in \Agents$ as at time 0, $\I,(r,0) \models \neg(\jdecided_j=0)$
    for all $j \in \Agents$.
Thus,    $\I_{\gamma_{\fip,n,t},P},(r,0) \models \init_{i} = 0 
    \lor \bigvee_{j \in \Agents} K_{i}(\jdecided_{j}=0)$ iff
$\I_{\gamma_{\fip,n,t},P},(r,0) \models \init_{i} = 0$. 
    So agent $i$ decides 0 in round 1 
    and receives a \zchain in round 0.
    
    If $m > 0$, then the only-if direction follows immediately from the 
definition of a \zchain, as the last agent $i$ decides for the first time
    after hearing from an agent that just decided 0. For the converse,
    suppose
    that $\I,(r,m) \models \init_i = 0 \lor \bigvee_{j \in \Agents} 
    K_i(\jdecided_j=0)$ holds for the first time at time $m$. 
    Since $m>0$ is the first time this formula holds and $i$ did not
    decide before round $m+1$, we must have $\I,(r,m) \models \init_i
    \neq 0$.  
    Thus, $\I,(r,m) \models K_i(\jdecided_j=0)$ for some 
        $j \in \Agents$, so agent $j$ decides 0 in round 
    $m$. 
We must have $\I,(r,m-1) \models \neg C_\N(\faultyag)$ since
otherwise, by Lemma~\ref{lemma:cfaulty},
$i$ would have decided by round $m$.
Thus,    $j$ could not have decided 0 
due to the common knowledge condition.  It follows from $\kbp^1$
that $\I,(r,m-1) \models \init_j = 0 \lor \bigvee_{j' \in \Agents} 
    K_j(\jdecided_{j'}=0)$. By the inductive hypothesis, $j$ receieves a \zchain 
    in round $m-1$.
    Since $\I,(r,m) \models K_i(\jdecided_j=0)$, agent $i$ receives the message sent by $j$ in round $m$.
    Because the common knowledge conditions are not satisfied and $i$ has not already decided,
    agent $i$ also decides 0 in round $m+1$.
    Thus. $i$ receives a \zchain in round $m+1$.
  \end{proof}

 \subsubsection{Characterizing the condition for deciding 1}
  
  To characterize when an agent is unable to decide 1 (excluding the decisions 
  made using one of the common knowledge conditions),
  we need some additional definitions.

\begin{definition}  Let $\mathit{len}_i(r,m)$ be the length of the longest \zchain 
that $i$ knows about at time $m$ in run $r$ (where $\mathit{len}_i(r,m) = 0$ if 
$i$ does not know about any \zchains), let
  $\mathit{last}_{ij}(r,m)$ be the last time $m'$ for which 
$(j,m') \lamport{r} (i,m)$ (where $\mathit{last}_{ij}(r,m) = -1$ if
$(i,m)$ does not hear from $(j,m')$ at all)
and let $\latest_i(r,m)$ to be the last time $m'$ such that, for some
agent $j$, we have $(j,m') \lamport{r} (i,m)$ and  
$\I,(r,m') \models \deciding_j=0$ (where $\latest_i(r,m) = -1$ if
there is no such time).
\end{definition}
  
Intuitively, if the common knowledge conditions do not hold
by time $m$, then an agent $i$ is unable to decide 1 if and only if 
  there are enough agents that could potentially extend the longest 
  \zchain that $i$ knows about,
  so that it has length at least $m$.
  The following lemma formalizes this intuition.
    \begin{proposition} \label{lemma:notdeciding1}
    For all implementations $P$ of the knowledge-based program $\kbp^1$
        with respect to $\gamma_{\fip,n,t}$, if
    $\I_{\gamma_{\fip,n,t},P},(r,m) 
    \models \neg K_i(C_{\N}(\faultyag \land \nodecided_{\N}(0) \land \exists 1))$,
    $\I_{\gamma_{\fip,n,t},P},(r,m)
    \models \neg K_i(C_{\N}($ $\faultyag \land \nodecided_{\N}(1) \land
        \exists 0))$, 
    and 
        $\I,{\gamma_{\fip,n,t},P},(r,m) \models \decided_i = \bot$, then
    the  
    following holds:
    \begin{itemize}
    \item $\I_{\gamma_{\fip,n,t},P},(r,m) \models
      \neg K_i \neg (\exists j \in \Agents (\deciding_j = 0))$
      if and only if 
      for all $m''$ with $\mathit{len}_i(r,m) < m'' \leq m$, 
      there exist at least $m'' - \mathit{len}_i(r,m)$ agents $j$ such that
      $\mathit{last}_{ij}(r,m) < m''$ and
      $\I,(r,\mathit{last}_{ij}(r,m)+1) \models \decided_j = \bot$.
    \end{itemize}

      \end{proposition}

  \begin{proof}
        Let $\I = \I_{\gamma_{\fip,n,t},P}$ and suppose that $\I,(r,m)
    \models \neg K_{i}(C_{\N}(\faultyag \land \nodecided_{\N}(1) \land \exists 0))
    \land \neg K_{i}(C_{\N}(\faultyag \land \nodecided_{\N}(0) \land \exists 1))$.
    By Lemma~\ref{lemma:cfaultyimpliesall}, we have that
    $\I,(r,m)
    \models \neg C_{\N}(\faultyag \land \nodecided_{\N}(1) \land \exists 0)
    \land \neg C_{\N}(\faultyag \land \nodecided_{\N}(0) \land \exists 1)$. 
    Note that the common knowledge conditions also do not hold for any 
    earlier time, as otherwise $i$ would have decided earlier by 
    Lemma~\ref{lemma:cfaulty}.

    For the only-if direction, suppose that $\I,(r,m) \models \neg K_i \neg (
    \deciding_j = 0)$ for some $j \in \Agents$.
There must exist a run $r'$ such that 
    $(r,m) \sim_i (r',m)$ and 
    $\I,(r',m) \models \deciding_j = 0$.

Since $(r,m) \sim_i (r',m)$ and $i$ the common knowledge conditions do not
hold for $i$ in $(r,m)$, they do not hold in $(r',m)$ either. By 
Lemma~\ref{lemma:chain}, $j$ must receive a \zchain in round $m$ of $r'$.
Let $i_0,\dots,i_m$ be this \zchain.
    For all $m^*$ such that $\mathit{len}_i(r,m) < m^* \leq m$, consider
        the agent $i_{m^*}$ in the \zchain.
   We claim that 
    $(i_{m^*},m') \not\lamport{r} (i,m)$
    for all $m' \geq m^*$. For suppose that 
    $(i_{m^*},m') \lamport{r}
    (i,m)$.  
    Then, since $(r,m) \sim_i (r',m)$, 
    we also have $(i_{m^*},m') \lamport{r'} (i,m)$.
    Since $i_{m^*}$
    receives a \zchain at time $m^*$ in $r'$, this must also be
    the case in $r$ and $i$ learns about it, contradicting the
    assumption that $\mathit{len}_i(r,m) < m^*$.  
    We must have
        $\mathit{last}_{ii_{m^*}}(r,m) < m^*$ (otherwise $\len_i(r,m)$
    would be at least $m^*$).
    It follows that
          $\I,(r,\mathit{last}_{ii_{m^*}}(r,m)+1) \models
    \decided_{i_m^*} = \bot$, for otherwise $i_{m^*}$ would not be on
    the $0$-chain.
    Thus, for all $m''$ with
$\mathit{len}_i(r,m) < m'' \leq m$, 
      there exist at least $m'' - \mathit{len}_i(r,m)$ agents $j$ such that
            $\mathit{last}_{ij}(r,m) < m''$, namely, the agents $i_{m^*}$ with
      $\mathit{len}_i(r,m) < m^* \leq m^*$.
    
      Conversely, suppose
      that
      for all $m''$ where $\mathit{len}_i(r,m) < m'' \leq m$, 
      there exist at least $m'' - \mathit{len}_i(r,m)$ agents $j$ such that
      $\mathit{last}_{ij} < m''$ and
      $\I,(r,\mathit{last}_{ij}(r,m)+1) \models \decided_j = \bot$. 
    Let $i_0,i_1,...,i_{\mathit{len}_i(r,m)}$ be the longest \zchain 
    that $i$ knows about at time $m$ in run $r$.

By assumption, there must exist agents 
$i_{\mathit{len}_i(r,m)+1}, \ldots, i_m$ such that for all $m''$ with
$\mathit{len}_i(r,m) < m'' \le m$, we have $\last_{ii_{m''}} < m''$
and       $\I,(r,\mathit{last}_{ii_{m''}}(r,m)+1) \models
\decided_{i_{m''}} = \bot$. Consider a run $r''$ such that the initial
  state of all agents is the same in $r$ and $r''$, $\N(r) = \N(r'')$,
  and the failure pattern in $r$ and $r'$ is the same except for the
  messages received by and from the agents  $i_{\mathit{len}_i(r,m)+1}, \ldots,
  i_m$.  Each agent $i_{m''}$ with $\mathit{len}_i(r,m) < m'' \le m$
  receives messages from the same agents in $r$ and $r''$ up to and
  including round $\last_{i,i_{m''}}$ and the same agents receive
  messages from $i_{m''}$ up to round $\last_{i,i_{m''}}$; after round
$\last_{i,i_{m''}}$, agent $i_{m''}$ receives only one message, from
  agent $i+_{m''-1}$ at round $m''$, and if $m'' < m$, only one
  message is received from $i_{m''}$  
  after round   $\last_{i,i_{m''}}$: a message is received by 
  $i_{m''+1}$ from $i_{m''}$ at round $m''+ 1$.  (If $m'' = m$, no
  messages are received from $i_{m''}$ after round
    $\last_{i,i_{m''}}$.)  It is easy to check that $(r,m) \sim_i
  (r'',m)$ and $i_0, \ldots, i_m$ is a \zchain in $r''$.
\end{proof}

  \begin{corollary} \label{lemma:deciding1}
    For all implementations $P$ of the knowledge-based program $\kbp^1$
    with respect to $\gamma_{\fip,n,t}$, if $\I_{\gamma_{\fip,n,t},P},(r,m)
    \models \neg C_\N(\faultyag)$, 
    $\I_{\gamma_{\fip,n,t},P},(r,m) \models \decided_i = \bot$,
        $\mathit{len}_i(r,m) \leq m-2$, and
    agent $i$ hears from all but one agent in round $m$ of $r$, then 
    $\I_{\gamma_{\fip,n,t},P},(r,m) \models \deciding_i = 1$.
  \end{corollary}
  \begin{proof}
    Suppose that $\I_{\gamma_{\fip,n,t},P},(r,m)
    \models \neg C_\N(\faultyag)$, $\mathit{len}_i(r,m) \leq m-2$, and 
    agent $i$ hears from all but one agent in round $m$ of $r$.
    For all agents $j$ that $i$ heard from in round $m$, 
    $\mathit{last}_{ij}(r,m) \ge m$. So there is only one agent that 
    could satisfy $\mathit{last}_{ij}(r,m) < m$. But 
 since $m - \mathit{len}_i(r,m) \geq 2$, Proposition~\ref{lemma:notdeciding1}
    implies that $\I_{\gamma_{\fip,n,t},P},(r,m) \models 
    K_i \neg (\exists j \in \Agents (\deciding_j = 0))$.
    Therefore agent $i$ decides 1 in round $m+1$.
  \end{proof}

    \subsubsection{$\kbp^1$ satisfies the EBA conditions}
\repro{p:kbp1correct}
All implementations of $\kbp^1$ with
respect to $\gamma_{\fip,n,t}$ are EBA decision protocols for $\gamma_{\fip,n,t}$.
\erepro

\begin{proof}
    Fix an implementation $P$ of $\kbp^1$ in $\gamma_{\fip,n,t}$.

    Unique Decision follows from essentially the same argument as in
    the $\kbp^0$ case. 
    The only difference is since $P$ is a FIP, the decisions 
    can be inferred from the local state without explicitly storing them.

    To see that Agreement holds, suppose by way of contradiction that $r$ is a 
    run where there exist nonfaulty agents $i$ and $j$ and a time $m$ such that
    $\I,(r,m) \models \decided_i = 0 \land \decided_j = 1$.
    Suppose that $j$ decides 1 in round $m_j+1$ and $i$ decides 0 in round $m_i+1$, 
    so that the decision conditions first hold at times $m_i$ and $m_j$, respectively. 
  We first observe that if either $i$ or $j$ decides using one of the common
    knowledge conditions, then we get a contradiction. If 
        $\I,(r,m_{i}) \models K_{i}(C_\N(\faultyag \land \nodecided_{\N}(1)
    \land \exists 0))$ then $j$ couldn't have decided 1 at or before round
    $m_{i}$, since  
    $j\in\N$ and $\I,(r,m_{i}) \models \nodecided_{\N}(1)$. Agent $j$ would 
    then decide 0 at round $m_{i}+1$, since $\I,(r,m_{i}) \models
    K_{j}(C_\N(\faultyag  
        \land \nodecided(1) \land \exists 0))$, contradicting
the assumption that $j$ decides 1 in this run. 
    If 
    $\I,(r,m_{i}) \models \neg K_{i}(C_{\N}(\faultyag \land \nodecided_{\N}(1)
    \land \exists 0))$ and $\I,(r,m_{j}) \models K_{j}(C_\N(\faultyag \land 
    \nodecided_{\N}(0) \land \exists 1))$, then $i$ could not have decided 0 
    before $m_{j}$ as $i \in \N$ and $\I,(r,m_{j}) \models \nodecided_{\N}(0)$. 
Since $\I,(r,m_{j}) \models K_{i}(C_\N(\faultyag \land 
    \nodecided_{\N}(0) \land \exists 1))$, and the other common knowledge 
    condition for deciding 0 does not hold for $i$
    at time $m_j$, by the argument above, 
    $i$ decides 1 
    in round $m_j + 1$,
        contradicting the assumption that $i$ decides 0 in round $m_i +1$.  
    
It remains to consider the cases where 
    neither $i$ nor $j$ decides 
using one of the common knowledge conditions.

If 
$m_{j} \leq m_{i}$, 
must have either $\I,(r,m_{i}) \models \init_{i} = 0$
    or $\I,(r,m_{i}) \models K_{i}(\jdecided_{k} = 0)$ for some $k \in \Agents$.
    Since, agent $i$ does not decide using a common knowledge condition, 
    we can apply Lemma~\ref{lemma:chain} to 
    conclude that $i$ receives a \zchain at time $m_i$, which implies that 
    there exists an agent $i'$ such that 
    $\I,(r,m_{j}) \models \deciding_{i'} = 0$. Hence, 
        $\I,(r,m_{j}) \models \neg K_{j}(\neg (\deciding_{i'}=0))$,
so  $j$ cannot decide 1 at $m_{j}$.
    
If $m_{j} > m_{i}$, since $i$ decides 0 in round $m_i + 1$ 
without using the common knowledge condition, we must have 
    $\I,(r,m_i) \models \init_i = 0 \bigvee_{j \in \Agents} 
K_i (\jdecided_j = 0)$. We can again apply Lemma~\ref{lemma:chain} to 
conclude that $i$ receives a \zchain at time $m_i$ in $r$.
 As $i$ is nonfaulty, $j$ must hear from $i$ in round $m_{i}+1$,
    so $\I,(r,m_j) \models K_j(\jdecided_i = 0)$. 
    It follows that agent $j$ should decide 0 in this run,
    contradicting    the assumption that $j$ decides 1.

    For Validity, observe that if an agent $i$ decides $v$ using the common 
    knowledge condition, it follows that some agent had an initial
        preference of $v$.
If $i$ decides 0 without using the common knowledge condition,
    by Lemma~\ref{lemma:chain}, there must be a \zchain, and hence an agent 
    that had an initial preference of 0. 
Finally, if agent $i$ decides 1 without using a common knowledge
condition, then $i$ did not 
    decide 0 in the first round and therefore we must have 
    $\init_i = 1$.

        For Termination,
                we must show that all nonfaulty agents decide by round $t+2$. 
        Suppose that a nonfaulty agent $i$ does not decide by round $t+1$
        in run $r$, and that $r'$ is a run such that $(r',t+1) \sim_i
        (r,t+1)$.  
Since we are using a full-information protocol, it easily follows that
$(r,m) \sim_i (r',m)$ for
all $m \leq t+1$.
Since $i$ does not decide by round $t+1$ of $r$, we do not have 
$\I,(r,m) \models K_i \ck{\N}(\faultyag \land \nodecided_\N(1-v)
    \land \exists v)$ 
        for $v\in \{0,1\}$ and $m \leq t$.     
     By Lemma~\ref{lemma:cfaultyimpliesall}, it easily follows that
    $\I(r',m)\models \neg \ck{\N}(\faultyag \land \nodecided_\N(1-v)
     \land \exists v)$  for $v\in \{0,1\}$ and $m \leq t$.  
     This implies that no agent $j$ decides in $r'$ using the
          common knowledge conditions at or before round $t+1$.
    
    By Lemma~\ref{lemma:chain}, if any agent $j$ decides 0 in round
    $t+1$ in $r'$, then that agent receives a \zchain $i_0, \ldots,
    i_{t}$  
    at time $t$ in $r'$. Since agents on a chain are distinct, and there are
    at most $t$ faulty agents,  
    this chain contains at least one nonfaulty agent $i_k$ that
    decides 0.  But then $i$ must receive a message (and a \zchain)
    from $i_k$ in round 
        $k+1 $, which means that $i$ decides 0 by round $t+2$, as claimed.

    We have shown if an agent decides by round $t+1$ using the commnon knowledge
conditions or decides 0 by  line 4 of $\kbp^1$
in a run $r'$ that $i$ considers possible, then $i$ decides by round
$t+2$ in $r$.  If this is not the case, then
  $\I,(r,t+1) \models K_i(\bigwedge_{j \in \Agents} \neg (\deciding_j
    = 0))$, so $i$ decides 1 in round $t+2$.
\end{proof}

\commentout{
Then using Proposition 4.3 from \cite{HMW}, we get:
\begin{corollary} \label{corollary:hmwcorrectness}
  For all implementations $P$ of the knowledge-based program $\kbp^1$ with 
  respect to $\gamma_{\fip,n,t}$,
  \begin{enumerate}[label=(\alph*)]
    \item $\I_{\gamma_{\fip,n,t},P} \models (\dec_i = 0)
        \Rightarrow B_i^\N (\exists 0 \wedge
      C^\boxdot_{\N \land \cO}\exists 0 \wedge \neg (\dec_i = 1))$
    \item $\I_{\gamma_{\fip,n,t},P} \models (\dec_i = 1)
        \Rightarrow B_i^\N (\exists 1 \wedge
      C^\boxdot_{\N \land \Z}\exists 1 \wedge \neg (\dec_i= 0))$
  \end{enumerate}
 \end{corollary}
}

\subsubsection{$\kbp^1$ is optimal}

To show that $\kbp^1$ is optimal, we first show that it suffices to
prove that weak safety implies optimality, and then show $\kbp^1$ is
weakly safe.

\rethm{thm:kbp1opt}
  If $\kbp^1$ is weakly safe with respect to $\gamma_{\fip,n,t}$
  then all implementations of $\kbp^1$
  are optimal with respect to $\gamma_{\fip,n,t}$.
\erethm
\begin{proof}
  \commentout{
  The proof uses the characterization of optimality from \cite{HMW}.
  Let $P$ be an implementation of $\kbp^1$ with respect to $\gamma_{\fip,n,t}$
  and $\I = \I_{\gamma_{\fip,n,t},P}$.

  We first define two operators $B_i^\N \phi$ and 
  $C_{\cS}^\boxdot \phi$ that
  are used in this characterization. $B_i^\N$ is an abbreviation of
  $K_i(i \in \N \rimp \phi)$. Thus, 
  $\I,(r,m) \models B_i^\N\phi$ if and only if  
  $\I,(r',m') \models \phi$ for all points $(r',m')$
  such that $r_i(m) = r'_i(m')$ and $i \in \N(r')$.
 Intuitively, $B_i^\N$ holds if $i$ knows that if it is nonfaulty, then
 $\phi$ holds.
  
  The common knowledge operator has a characterization in terms of 
  $\cS$-$\boxdot$-reachability. In \cite{HMW}, it is shown that 
  $\I,(r,m) \models C_\cS^\boxdot \phi$ if and only if 
  $\I,(r',m') \models \phi$ for all points $(r',m')$ 
  that are $\cS$-$\boxdot$-reachable from $(r,m)$.  We are interested
  in the cases that $\cS$ is either
  $\N \land \cO$ or $\N \land \Z$.  Theorem 5.4 in 
  \cite{HMW} shows that $\kbp^1$ is an optimal EBA protocol 
  if the following holds:
  \begin{align*}
    \I \models i \in \N \rimp (\Circ(\decided_i = 0)
        \Leftrightarrow B_i^\N (\exists 0 \wedge
      C^\boxdot_{\N \land \cO}\exists 0 \wedge \neg (\Circ(\decided_i = 1)))) \\
    \I \models i \in \N \rimp (\Circ(\decided_i = 1)
        \Leftrightarrow B_i^\N (\exists 1 \wedge
      C^\boxdot_{\N \land \Z}\exists 1 \wedge \neg (\Circ(\decided_i = 0))))
  \end{align*}
}
Suppose that $\kbp^1$ is weakly safe with respect to
$\gamma_{\fip,n,t}$.
  To prove that all implementations of $\kbp^1$ are optimal with
  respect to $\gamma_{\fip,n,t}$,
we use Theorem~\ref{thm:HMWchar}.  
  It suffices to prove the only if 
  direction, since Proposition 4.3 in \cite{HMW} shows that the if direction 
  holds for EBA protocols, and Proposition~\ref{p:kbp1correct} shows
  that $\kbp^1$ is an EBA protocol. 
    
  Suppose that $\kbp^1$ is weakly safe with respect to
  $\gamma_{\fip,n,t}$. Let $P$ be an implementation of $\kbp^1$ and
  let $\I = \I_{\gamma_{\fip,n,t},P}$,
  We give the argument for the $\N\land\cO$ case. 
  We first assume that $\I,(r,m) \models i \in \N$ for some point $(r,m)$. 
 In terms of $(\N\land\cO)$-$\boxdot$-reachability, we  want to show that
 if for all points $(r',m)$ such that $i \in \N(r')$, 
  $\I,(r',m) \models \exists 0 \land \neg (\Circ(\decided_i = 1))$, and 
  for all points $(r'',m')$ that are $(\N\land\cO)$-$\boxdot$-reachable from 
  $(r',m)$, we have $\I,(r'',m') \models \exists 0$, then 
  $\I,(r,m) \models \Circ(\decided_i = 0)$. Suppose by way of contradiction 
  that for all $(r',m)$ such that $i \in \N(r')$ the condition above holds but 
  $\I,(r,m) \models \neg \Circ(\decided_i = 0)$. Then, 
$\I,(r,m) \models i \in \N \land \Circ(\decided_i = \bot)$ and by weak safety, 
  there exist points $({r^1}',m)$ and $({r^1}'',m)$ such that
  $r_i(m) = {r^1}'_i(m)$,
  $i \in \N({r^1}')$,
  $({r^1}'',m')$ is $(\N\land\cO)$-$\boxdot$-reachable from $({r^1}',m)$, and
  $\I,({r^1}'',m') \models \neg \exists 0$. 
  This is a contradiction.
  Since this holds for all $(r,m)$, the only-if direction of the 
  first optimality condition holds.
  The argument for the $\N\land\Z$ case is completely analogous.

  Therefore all implementations $P$ of $\kbp^1$ with respect to 
  $\gamma_{\fip,n,t}$ are optimal.
\end{proof}

Let $P$ be an implementation  on $\kbp^1$ and let $\I =
\I_{\gamma_{\fip,n,t},P}$..  
We want to show that $\kbp^1$ is weakly safe.  So suppose that
$\I,(r,m) \models i \in \N \land \Circ(\decided_i = \bot)$.  We need
to show that there exist points       $({r^0}',m)$, $({r^0}'',m)$,
$({r^1}',m)$, and $({r^1}'',m)$ satisfying the conditions of weak
safety.  Before we do this, we introduce an invariant.

\begin{definition}[Invariant condition for $v \in \{0,1\}$]
    $\I,(r,m) \models \inv_{v}(i,j,k)$ if and only if there exist distinct 
agents  $i,j,k$ such that:
  \begin{itemize}
    \item $\I,(r,m) \models (\decided_{i} = \bot) \land (\deciding_{j} = v)
          \land (\decided_{k} = \bot)$,
    \item $\I,(r,m) \models i \in \N \land j \in \N \land k \not\in \N$, and
    \item 
            $k$ does not exhibit any faulty behaviour throughout $r$.
   \end{itemize}
\end{definition}

The motivation for the constraint on $k$ in the definition is the following: 
\begin{lemma} \label{lem:notCN}
  If $k \not \in \N(r)$ and $k$ does not exhibit faulty behavior in $r$,
  then $\I,(r,m) \models \neg
\ck{\N}(\faultyag)$ for all $m$. 
\end{lemma} 

\begin{proof} 
If 
If $k \notin \N(r)$ but does not exhibit fauilty behavior, let $r'$ be
a run such that $\N(r') = \N(r) - \{k\}$, the failure pattern in $r$
  and $r'$ is the same, and all agents have the same initial
  preferences in $r$ and $r'$.  Clearly, for all times $m$, $(r,m)
  \sim_i (r',m)$ and $\I,
  (r',m) \not\models \faultyag$, so $\I,(r,m) \models  \neg
\ck{\N}(\faultyag)$.
\end{proof}

It follows immediately that
$\I,(r,m) \models \inv_1(i,j,k)$, then no agent decides using the
common knowledge conditions in $r$.

We are now ready to construct the points $(r_i^{0''},m)$ and
$(r^{1''},m)$ required for weak safety.
The main part of the argument for this case is done by
Lemmas~\ref{lemma:dir1-induction} and~\ref{lemma:o-contra}.
Lemma~\ref{lemma:dir1-induction} shows that once we are at a point where 
the invariant condition for 1 holds, we can $(\N\land\cO)$-$\boxdot$-reach 
a point where every agent has initial preference 1. 
Lemma~\ref{lemma:o-contra} shows that if a nonfaulty agent $i$ is unable to 
decide, $i$ must consider possible a run where the invariant condition for 1
holds.
The desired result follows from these two lemmas.

\newcommand{\edges}{E}

The following technical lemma will play a key role in our proof
of weak safety.
For a point $(r,m)$, we write $\edges(r,m)$ for the set of edges
$(i,k-1) \lamport{r} (j,k)$  
with $1 \leq k \leq m$.

\begin{lemma} \label{lemma:zchaincut} 
  If $i\in \N(r)$, $\I,(r,m) \models \neg \ck{\N}(\faultyag)$,
  $\latest_i(r,m) =\ell < m$, and $\I,(r,m) \models \decided_i =
  \bot$, then there
  exists a run $r'$
such that  
\begin{itemize} 
\item $(r,m) \sim_i(r',m)$,
  \item $\N(r) = \N(r')$,
\item $\edges(r',m) \subseteq \edges(r,m)$, 
\item all \zchain{s} in $r'$ are known to agent $i$ at time $m$; that is, if 
  $i_0, \ldots i_k$ is a \zchain in $r'$, then $(i_k, k)
  \lamport{r'} (i,m)$,  
\item all agents that do not exhibit faulty behavior in $r$ also do
  no not exhibit faulty behavior in $r'$. 
\end{itemize}	
\end{lemma} 

\begin{proof} 
Note that since $\I,(r,m) \models \decided_i = \bot$ and $\latest_i(r,m)< m$, we cannot have 
$\I,(r,m) \models \deciding_i =0$, so in fact agent $i$ cannot hear
from any agent that decides 0 in round $m$.  
This means that we must have $\ell = \latest_i(r,m) <m-1$.
We construct $r'$ by modifying $r$ appropriately.  For all $(j,m') \not
\lamport{r} (i,m)$,
if  $m'=0$ and $\I,(r,0) \models \init_j = 0$, then we modify
$\init_j$ to 1 in $r'$, and if $m'>0$, we restrict the 
messages arriving at $(j,m')$ to be only those from the nonfaulty
agents in $\N(r)$, and the agents that exhibit no faulty behavior in
$r$. We have 
$\N(r') = \N(r)$.  

Clearly,  $\edges(r',m) \subseteq \edges(r,m)$. 
Since we do not modify the failure pattern for messages corresponding to
pairs on the path $(j,m') \lamport{r} (i,m)$, 
we have $(r,m) \sim_i (r',m)$. Moreover, the construction changes only the failure behaviour of agents who
exhibit faulty behavior in $r$, so all agents that do not exhibit
faulty behavior in $r$ also do not exhibit faulty behavior in $r'$. 

Note that because $\I,(r,m) \models i\in \N \land \neg \ck{\N}(\faultyag)$ and $(r,m) \sim_i(r',m)$, 
we have $\I,(r',m) \models \neg \ck{\N}(\faultyag)$, so in $r'$, no
agent decides using the common knowledge conditions before round
$m+2$. In particular, any 0-decisions made before this round must be
made using the fourth line of the knowledge-based program.  

The construction guarantees that if $i_0, \ldots, i_\ell$ is a \zchain
in $r'$, then $(i_\ell, \ell) \lamport{r'} (i,m)$.
To see this, first note that none of $i_0, \ldots, i_\ell$ can be
nonfaulty, otherwise $i$ would receive a \zchain in $r$ and $r'$,
and decide 0 before time $m$, contradicting the assumption that
$\I,(r,m) \models \deciding_i = \bot$.  If $\ell > 0$, then the fact
that $i_\ell$ received a message from $i_{\ell-1}$ in $r'$ and
$i_{\ell-1}$ is faulty means that $(i_{\ell},\ell) \lamport{r}
(i,m)$.  It easily follows that $(i_{\ell},\ell) \lamport{r'}
(i,m)$.  Thus, $i$ knows about this \zchain in $r'$.  If $\ell = 0$, then
either $(i_0,0) \lamport{r} (i,m)$, hence $(i_0,0) \lamport{r'}
(i,m)$, or the initial value of $i_0$ was changed to 1 in $r'$, so this is not
in fact a \zchain.

It remains to show that $r'$ has no \zchain{s}
of length greater than $\ell$.
Suppose to the contrary that that $i_0, \ldots, i_{\ell+1}$ is a
\zchain in $r'$ of length $\ell+1$.
If $\ell \ge 0$, then
since $\latest_i(r,m) = \ell$, we must have $(i_{\ell+1} ,\ell+1)
\not \lamport{r'} (i,m)$, for otherwise we would have  
$(i_{\ell+1} ,\ell + 1) \lamport{r} (i,m)$ and $i_0, \ldots,
i_{\ell+1}$  would be a \zchain in $r$ of length $\ell+1 >
\latest_i(r,m)$, a contradiction. 
\commentout{
We cannot have $\ell+1 = 1$, since then we have $\I,(r',0) \models
\init_{i_0} = 1$ by construction, contradicting 
that $i_1$ is a $\zchain$. On the other hand if $\ell+2 > 1$, then in round $\ell+1$, agent $i_{\ell+2}$
receives messages only from 
agents that do not exhibit faulty behavior in $r$. Since $\ell+1 <m$, none 
of these agents can have just decided 0, since otherwise
$i$ would also receive its message, and $i$ would decide 0 in round 
$\ell+2\leq m$, 
so that $\latest_i(r,m) > \ell$, a contradiction.
Consequently, we have $\I,(r,\ell+1) \models \neg \decides_{i_{\ell+2}} = 0$. But this contradicts the assumption 
that $i_1, \ldots, i_{\ell+2}$ is a \zchain in $r'$. We conclude that $r'$ contains no \zchain{s} longer than $\ell+1$. 
}
If $\ell = -1$ (i.e., $i$ does not know about any \zchains in $r$), suppose there
is \zchain in $r'$.  Then there must be some agent $i_0$ with an
initial preference of 0 in $r'$.  We cannot have $(i_0,0) \lamport{r}
(i,m)$, for otherwise $i$ would know about a \zchain in $r$.  But
then our construction guarantees that the initial preference of $i_0$
in $r'$ is 1, not 0.  We conclude that $r'$ contains no \zchain{s}
of length at least  $\ell+1$, as desired.
\end{proof}

\begin{lemma} \label{lemma:rollback}
  If 
  $m \geq 2$,
  $\I,(r,m) \models \inv_{1}(i,j,k)$, and 
  $\latest_j(r,m) \leq m-3$,
    then there exists a run $r^{*}$ such that 
  $\I,(r^{*},m-1) \models \inv_{1}(i',j',k')$ for some $i',j',k' \in \Agents$,
   $\edges(r^*,m-2) \subseteq \edges(r,m-2)$, and $\N(r) \cup \{k\} =
  \N(r^*) \cup \{k'\}$. 
\end{lemma}
\begin{proof}
 By Lemma~\ref{lemma:zchaincut}, there exists a point
 $(r^\dagger,m) \sim_j (r,m)$ such that $\N(r) = \N(r^\dagger)$, 
  the longest \zchain in $r^\dagger$ has length  at most $m-3$,
  and $\edges(r^\dagger,m) \subseteq \edges(r,m)$.
  Moreover, since $k$ exhibits no faulty behavior in $r$, the same
 holds in $r^\dagger$.  
In addition, since 
$\I,(r,m) \models \decided_i = \bot \land\decided_j = \bot \land
\decided_k = \bot$, we also
 have  
 that $\I,(r^\dagger,m) \models \decided_i = \bot \land\decided_j = \bot \land
\decided_k = \bot$.
 
 We construct a run $r'$ by modifying $r^\dagger$ so that agent $i$ is 
 faulty, 
 $k$ is nonfaulty, 
 and agent $k$ hears from every other agent except agent 
 $i$ in round $m$. 
Since the only modifications  are in round $m$, we have $\edges(r',m-1) =
 \edges(r^\dagger,m-1) \subseteq \edges(r,m-1)$. 
 Agent  
  $j$ is nonfaulty in both $r$ and $r'$ and 
  has the same local state in both $(r,m)$ and $(r',m)$. Hence, $j \in (\N\land\cO)(r,m) \cup 
  (\N\land\cO)(r',m)$; therefore, $(r',m)$ is 
  $(\N\land\cO)$-$\boxdot$-reachable from $(r,m)$ through agent $j$.
  Moreover, the length of the longest \zchain in $r'$ is $m-3$,
    because we do not modify 
$r^\dagger$ in round $m-2$ or earlier in constructing $r'$.  
  Since $m\geq 2$ and $\latest_k(r',m) \leq m-3$, 
    in $r'$, agent $k$ decides 1 in round $m+1$
  by Corollary~\ref{lemma:deciding1}.
  
    Next consider the run $r''$ that is identical to 
  $r'$ except that agent $i$ hears from every agent in round $m-1$. 
  Here we have $\edges(r'',m-2) = \edges(r',m-2) \subseteq \edges(r,m-2)$. 
  Since the messages received by agent $k$ are identical in round $m$,
  we have $k \in (\N\land\cO)(r',m) \cup (\N\land\cO)(r'',m)$. 
  Thus, $(r'',m)$ is $(\N\land\cO)$-$\boxdot$-reachable from $(r',m)$ 
  through agent $k$.
  Because the longest \zchain in $r'$ has length $m-3$,
  and we make no change to $r'$ in round $m-2$  or earlier in
  constructing $r''$,  
  the longest \zchain in
  $r''$ aiso has length $m-3$.  
By Corollary~\ref{lemma:deciding1},
  agent $i$ decides 1 in round $m$ of $r''$
  upon hearing from every agent in round $m-1$. 
  Since agent $j$ heard from agent $i$ and $j$ did not decide earlier 
  in run $r''$, 
  agent $j$ 
  decides 1 in round $m+1$ in $r''$. 
    
  Now
    consider the run $r^*$ that is identical to $r''$ except that 
  agent $i$ is in $\N$ (so that agent $k$ does get a message from agent $i$ in round $m$), 
    and $k$ is faulty in $\N$ (but does not exhibit any faulty behavior.) 
Clearly, $\edges(r^*,m-2) = \edges(r'',m-2) \subseteq \edges(r,m-2)$. 
    Since $j$ has the same local state in both 
 $(r'',m)$ and $(r^*,m)$, and is nonfaulty in both $r''$ and $r^*$, we have
    $j \in (\N\land\cO)(r'',m) \cup (\N\land\cO)(r^*,m)$. Therefore, 
  $(r^*,m)$ is $(\N\land\cO)$-$\boxdot$-reachable from $(r'',m)$ through 
  agent $j$.

Finally, observe that 
in run $r^*$, 
$i$ is nonfaulty and decides 1 in 
round $m$, $k$ is nonfaulty and does not decide before round $m$,
and $j$ is faulty but exhibits no faulty behavior and   
does not decide before round $m$. Hence, 
  $\I,(r^*,m-1) 
  \models \inv_1(k,i,j)$. By the transitivity 
  of the $(\N\land\cO)$-$\boxdot$-reachability relation, 
  $(r^*,m-1)$ is 
  $(\N\land\cO)$-$\boxdot$-reachable from $(r,m)$, and the claim holds with
  $(i',j',k') = (k,i,j)$. Each step of the
  construction involved a swap of faultiness between a faulty and
  nonfaulty agent among $i$, $j$, and $k$, leaving the 
  faultiness of other agents invariant, so we have $\N(r^*) \cup \{k'\} = \N(r) \cup \{k\}$. 
\end{proof}

\newcommand{\sends}{\mathit{Sends}}

In the following, we write $r[0..m]$ for the prefix of the run $r$ up
to and including to time $m$, but removing 
the information about which agents are nonfaulty from the environment
state, and
write $\sends_i(r,m)$ for the set of agents $j$ such that $(i,m-1) \lamport{r} (j,m)$. 

\begin{lemma} \label{lemma:delete-pastv2}
  If $\I,(r,m) \models \inv_i(i,j,k)$ and the longest \zchain
in $r$ has length at most $m-2$,
then for all $p \not \in \N(r)\cup\{k\}$, there exists a
$(\N\land\cO)$-$\boxdot$-reachable point $(r',m')$ with  
$m' \in \{m,m+1\}$ 
such that $r[0..m-1] = r'[0..m-1]$, $\I,(r',m') \models
\inv_i(i',j',k')$, $p \not \in \{i',j',k'\}$, and  
$(p,m-1)\not  \lamport{r'} (j',m')$.  
\end{lemma} 

\newcommand{\run}[1]{r^{#1}} 
\newcommand{\reach}{$(\N\land\cO)$-$\boxdot$ reachable\xspace}

\begin{proof} 
We construct an $(\N\land\cO)$-$\boxdot$-path that establishes the result. 
We remark that $\inv_1$ is generally not maintained along this path: we falsify $\inv_1$, 
but re-establish it in the final step.
The construction has several branches, depicted in Figure~\ref{fig:no-path1}, 
in which we show just the changes made to obtain each successive run,  
and track failure edges in the causality graph. Agent timelines are depicted horizontally. 
Failure edges are indicated by dashed lines, and nodes labelled 1 indicate that the agent 
is deciding 1 at that node. 

\begin{figure} 
\centerline{\includegraphics[height=8in]{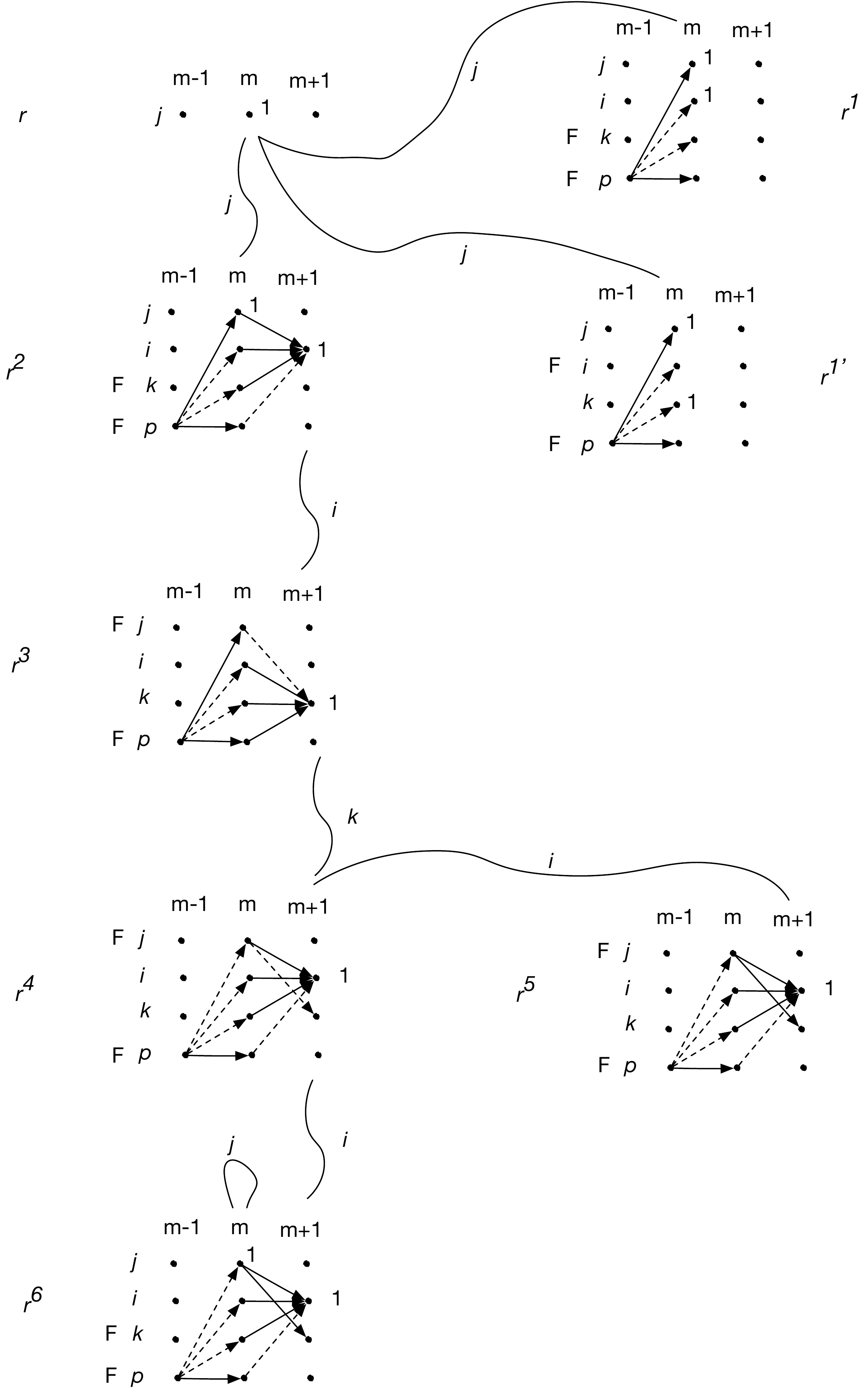}}
\caption{Construction for Lemma~\ref{lemma:delete-pastv2} \label{fig:no-path1}} 
\end{figure} 

If $(p,m-1)\not  \lamport{r} (j,m)$, we can take $(r',m') = (r,m)$ and
$(i',j',k') = (i,j,k)$, and we are done.  
Otherwise, consider the run $\run{1}$ obtained  by modifying $r$ so that in round $m$, 
$p$'s messages to all agents but $j$ and itself fail, that is, $\sends_p(\run{1},m) = \{p,j\}$.  
We have $\N(\run{1}) = \N(r)$. Since $(r,m) \sim_j (\run{1},m)$ and $j \in (\N \land \cO)(r,m) \cap (\N \land \cO)(\run{1},m)$, 
we have that $(\run{1},m)$ is \reach from $(r,m)$.

It remains the case that $i,j$ and $k$  
exhibit no faulty behavior in $\run{1}$, and that they are undecided
at time $m$.  
Since $k$ is faulty but exhibits no faulty behavior in $\run{1}$, by
Lemma~\ref{lem:notCN}, no agent decides using the common knowledge
conditions in $\run{1}$. It also remains the case that there is no \zchain
with length greater than $m-2$, so no  
agent decides 0 after round $m-1$ in $\run{1}$. 
If $\I,(\run{1},m) \models \deciding_i = 1$, then we are done, taking
$(r',m') = (\run{1},m)$ and $(i',j',k')=(j,i,k)$,  since $i\in
\N(\run{1})$  
but $(p,m-1) \not \lamport{\run{1}} (i,m)$  and $r[0..m-1] = \run{1}[0..m-1]$ by construction. 
On the other hand, if $\I,(\run{1},m) \models \deciding_k = 1$, then we
would similarly be done, taking $(r',m') = (\run{1'},m)$ and $(i',j',k')
= (j,k,i)$, where $\run{1'}$ is the run obtained from $\run{1}$ by  
setting $\N(\run{1'}) = \N(r) \cup \{k\}\setminus \{i\}$.

If $\I,(\run{1},m) \models \deciding_i = \bot \land \deciding_k = \bot$, 
we have  $\I,(\run{1},m+1) \models \deciding_i = 1$, since $\I,(\run{1},m) \models \deciding_j = 1$ and $j \in \N(\run{1})$. 
Let $\run{2} $ be the run obtained by modifying  $\run{1}$ 
so that in round $m+1$, agent $i$ receives messages from all agents but $p$. We have 
$\N(\run{2}) = \N(r)$. 
Since $(r,m) \sim_j (\run{2},m)$ and $j \in (\N \land \cO)(r,m) \cap (\N \land \cO)(\run{1},m)$, 
we have that $(\run{2},m)$ is \reach from $(r,m)$.

It remains the case that $k$  is faulty but
exhibits no faulty behavior in $\run{2}$, hence no agent decides using the 
common knowledge conditions in $\run{2}$. 
Note that $\run{2}[0..m-1] = \run{1}[0..m-1] = r[0,m-1]$ and
$\sends_p(\run{2},m) = \sends_p(\run{1},m) = \{p,j\}$. 
We also have $\run{2}_i(m) = \run{1}_i(m)$ and $\run{2}_k(m) = \run{1}_k(m)$, so 
$\I,(\run{1},m) \models \deciding_i= \bot \land \deciding_k = \bot$.  
Thus, using Corollary~\ref{lemma:deciding1}, we have that $\I,(\run{2},m+1) \models \deciding_i = 1$, because 
the longest \zchain in $\run{2}$ has length at most $m-2$.

Let $\run{3}$ be the run obtained by modifying $\run{2}$ so that 
$\N(\run{3}) = \N(r) \cup \{k\} \setminus \{j\}$,  and 
agent $k$ receives a message from all agents but $j$ in round $m+1$. 
Since $(\run{2},m)+1 \sim_i (\run{3},m+1)$ and $i \in (\N \land \cO)(\run{2},m+1) \cap (\N \land \cO)(\run{3},m+1)$, 
we have that $(\run{3},m)$ is \reach from $(r,m)$.
We remark that $\inv_1$ does not hold in $\run{3}$
with a permutation of $i,j,k$, since $j$ now exhibits faulty behavior.
However, we still have $\I,(\run{3},m+1) \models \neg
\ck{\N}(\faultyag)$, because $i \in \N(\run{3})$ and 
$i$ does not observe $j$'s faulty behavior.
It follows that no agent decides in $\run{3}$ using the common
knowledge conditions before round $m+3$.

Note that $\sends_p(\run{3},m) = \sends_p(\run{2},m) \setminus\{j\} =
\{p,j\}$ and $\sends_j(\run{3},m+1) = \Agents\setminus\{k\}$.  
We have $\run{3}[0..m-1] = \run{2}[0..m-1] = r[0,m-1]$. 
This implies that none of $i,j,k$ has decided at time $m$. 
In addition, $\run{3}_i(m) = \run{2}_i(m)$ and $\run{3}_k(m) =
\run{2}_k(m)$, so  
neither $i$ nor $k$ decides in round $m+1$. 
By Corollary~\ref{lemma:deciding1}, we have that $\I,(\run{3},m+1)
\models \deciding_k = 1$, because  
the longest \zchain in $\run{3}$ has length at most $m-2$.

Let $\run{4}$ be the run obtained by modifying $\run{3}$ so that agent
$p$ does not send any messages to agent $j$  
in round $m$, and in round $m+1$,  all agents except $p$ send a message to agent $i$. 
We have $\N(\run{4}) = \N(\run{3}) = \N(r) \cup \{k\} \setminus \{j\}$
and $\run{4}[0..m-1] = \run{3}[0..m-1] = r[0,m-1]$. 
In addition, $\sends_p(\run{4},m) = \sends_p(\run{3},m) \setminus \{j\} =
\{p,j\} \setminus \{j\} = \{p\}$,  
and  $\sends_j(\run{4},m+1) = \sends_j(\run{3},m+1) = \Agents\setminus\{k\}$. 
Since in $\run{3}$, agent $j$ does not send a message to agent $k$ in
round $m+1$,  
the change in round $m$ is not visible to agent $k$ at time $m+1$, nor
is the round $m+1$ change visible to $k$ at time $m+1$.  
Thus, we have $(\run{3},m+1) \sim_k (\run{4},m+1)$. 
In addition, $k \in  (\N \land \cO)(\run{3},m+1) \cap (\N \land \cO)(\run{4},m+1)$, 
so $(\run{4},m+1)$ is \reach from $(r,m)$. 
We have $\run{4}_i(m) = \run{3}_i(m)$, so agent $i$ does not decide in round $m+1$ in $\run{4}$. 
By Corollary~\ref{lemma:deciding1}, we get that  $\I,(\run{4},m+1) \models \deciding_i = 1$. 

It is still the case in $\run{4}$ that agent $j$ exhibits faulty behavior. To
reinstate $\inv_1$,  
let $\run{5}$ be the run obtained from $\run{4}$ by changing the failed message 
from $j$ to $k$ in round $m+1$ to be successfully transmitted. 
Since this was the only failure of $j$ introduced earlier, this ensures that 
$j$ does not exhibit faulty behavior in $\run{5}$. 
We have $\N(\run{5}) = \N(\run{4}) = \N(r) \cup \{k\} \setminus \{j\}$
and $\run{4}[0..m-1] = \run{3}[0..m-1] = r[0,m-1]$. 
The latter means that no agent decides 0 after round $m-1$. 
The change made in constructing $\run{5}$ is not visible to agent $i$ at time $m+1$, so we have 
$(\run{4},m+1) \sim_i (\run{5},m+1)$. We have 
 $i \in (\N \land \cO)(\run{4},m) \cap (\N \land \cO)(\run{5},m)$, 
so  $(\run{5},m)$ is \reach from $(r,m)$.

In addition, $\run{5}_i(m) = \run{4}_i(m)$, so agent $i$ does not decide in
round $m+1$ of run $\run{5}$,  
since it did not do so in round $m+1$ of $\run{4}$. 
Similarly, $\run{5}_k(m) = \run{4}_k(m)$, so agent $k$ does not decide in round $m+1$ of run $\run{5}$. 
With respect to agent $j$, we have two possibilities. 
\begin{itemize} 

  \item  If $\I,(\run{5},m) \models \deciding_j = 1$, 
then let $\run{6}$  be the run identical to $\run{5}$ except that $\N(\run{6}) = \N(r)$. 
(That is, we switch $k$ from being nonfaulty in $\run{5}$ to faulty in
$\run{6}$, and $j$ from being faulty to being
nonfaulty.) 
We still have that $(\run{4},m+1) \sim_i (\run{6},m+1)$ and 
 $i \in (\N \land \cO)(\run{4},m+1) \cap (\N \land \cO)(\run{6},m+1)$, 
so  $(\run{6},m+1)$ is \reach from $(r,m)$.
We also have that  $\I,(\run{6},m) \models \deciding_j = 1$, 
because $\run{6}_j(m) = \run{5}_j(m)$. 
Since $(\run{6},m) \sim_j (\run{6},m)$ and 
 $i \in (\N \land \cO)(\run{6},m) \cap (\N \land \cO)(\run{6},m)$, 
we have that $(\run{6},m)$ is \reach from $(r,m)$.
Note that  $(p,m-1) \not \lamport{\run{6}} (j,m)$. 
Hence we are done, taking $(r',m') = (\run{6},m)$
and $(i',j',k') = (i,j,k)$. 

\item  If  $\I,(\run{5},m) \models \neg \deciding_j = 1$, 
then we have $\I,(\run{5},m+1) \models \decided_i =  \decided_j = \decided_k = \bot$, 
and $\I,(\run{5},m+1) \models \deciding_i = 1$. 
Moreover, $\sends_p(\run{5},m) = \{p\}$ and $i \not \in 
\sends_p(\run{5},m+1)$.
Thus  $(p,m-1) \not \lamport{\run{5}} (i,m+1)$. 
Hence we are done, taking $(r',m') = (\run{5},m+1)$
and $(i',j',k') = (k,i,j)$. 

\end{itemize} 
\end{proof}

The following lemma allows us to reduce the size of the set $S_{m-1}(r)$, 
so that there are fewer paths by which a \zchain reaching time 
$m-2$ is visible at time $m$. 

\begin{lemma} \label{lem:shortenchain}
  If $\I,(r,m) \models \inv_1(i,j,k)$ and the longest \zchain in
  $r$ has length at least $m-2$, then there exists a run $r'$ and a
permutation $i',j',k'$ of $i,j,k$ such that  
$\I,(r',m) \models \inv_1(i',j',k')$,
$\N(r) \cup \{k\} = \N(r') \cup \{k'\}$, 
the point $(r',m)$ is $(\N \land \cO)$-$\boxdot$-reachable from $(r,m)$, and 
there are no \zchain{s} in $r'$ of length $m-2$. 
\end{lemma} 

\begin{proof} 
Consider an agent $q \in \Agents (\N(r) \cup \{k\})$. 
\newcommand{\faultyagin}{F}
We write $\faultyagin_q(r,m-1)$ for the set of edges
of the form $(p,m-2) \lamport{} (q,m-1)$ in $\edges(r,m-1)$,  
where $p \in  \Agents \setminus \N(r) \cup \{k,q\}$.  

We first show that there exists 
a point $(\run{q},m)$ that is $(\N \land \cO)$-$\boxdot$-reachable
from $(r,m)$
such that $\I, (\run{q},m) \models \inv_1(i^q,j^q,k^q)$, where  $\{i,j,k\} =
\{i^q,j^q,k^q\}$,
$\N(r) \cup\{k\} = \N(r^q) \cup \{k^q\}$, and
$\edges(\run{q},m-1) \subseteq \edges(r,m-1) \setminus \faultyagin_q(r,m-1)$.  
To do this, we first apply Lemma~\ref{lemma:delete-pastv2} to $(r,m)$
and $(N\land\cO)$-$\boxdot$-reach 
a point $(\run{1},m_{1})$ such that $\I,(\run{1},m_{1}) \models
\inv_{1}(i_1,j_1,k_1)$,
$N(r) \cup \{k\} = \N(\run{1}) \cup \{k_1\}$,
$\edges(\run{1},m-1) = \edges(r,m-1)$,
$(q,m-1) \not \lamport{\run{1}} (j_1,m_1)$, and $m_{1}$ is either $m$ or $m+1$.
    
Let $\run{2}$ be the run obtained from $\run{1}$ by converting
all edges $(p,m-2) \lamport{\run{1}} (q,m-1)$ 
in $\faultyagin_q(r,m-1)$ to failures,  so that $\edges(\run{2},m-1) = \edges(\run{1},m-1) \setminus \faultyagin_q(r,m-1)
= \edges(r,m-1) \setminus \faultyagin_q(r,m-1)$.
We have $(\run{2},m_1) \sim_{j_1} (\run{1},m_1)$, and
$\N(\run{1}) = \N(\run{2})$,  
so  $(\run{2},m_1)$ is  $(N\land\cO)$-$\boxdot$-reachable from $(\run{1},m_1)$. 
It remains the case that none of $i_1,j_1,k_1$ have decided at $(\run{2},m_1)$, since the state of 
these agents at time $m_1-1$ is visible to $j_1$ at time $m_1$ in $\run{1}$ and $(\run{2},m_1) \sim_{j_1} (\run{1},m_1)$.  
Hence $\I,(\run{2},m_1) \models \inv_1(i_1,j_1,k_1)$ and $\N(r) \cup \{k\} = \N(\run{2} \cup \{k_1\}$. 

If $m_1 = m$ we are done, taking $\run{q} = \run{2}$; otherwise, 
$m_1 = m+1$ and we apply Lemma~\ref{lemma:rollback} to obtain 
a point $(\run{q},m)$ \reach from $(r,m)$ with 
$\I,(\run{q},m_1) \models \inv_1(i_q,j_q,k_q)$ and $\N(\run{q}) \cup
\{k^q\} = \N(\run{2}) \cup \{k_1\} = \N(r) \cup \{k\}$. 
We have $\edges(\run{q},m-1) = \edges(\run{q},m_1-2) \subseteq
\edges(\run{2}, m_1-2)  
= \edges(\run{2}, m-1) = \edges(r,m-1) \setminus \faultyagin_q(r,m-1)$.

We successively repeat the steps above for all agents in $\Agents
\setminus (\N(r) \cup \{k\})$,  
thereby $(N\land\cO)$-$\boxdot$-reaching a point 
$(\run{3},m)$ with $\I,(\run{3},m)\models \inv_1(i_3,j_3,k_3)$,
$\N(\run{3}) \cup \{k_3\} = \N(r) \cup \{k\}$, and 
$\edges(\run{3},m-1) \subseteq \edges(r,m-1) \setminus (\bigcup_{q\in
\Agents \setminus (\N(r) \cup \{k\})} ~\faultyagin_q(r,m-1))$. 
That is, in round $m-1$ of $\run{3}$, no messages are transmitted between distinct (faulty) agents in $\Agents \setminus (\N(r) \cup \{k\})$. 

It may still be the case that we have faulty agents $q$ that decide in
round $m-1$  of $\run{3}$,  
for which $(q,m-1) \lamport{\run{3}} (j',m)$. 
In this case, we again apply Lemma~\ref{lemma:delete-pastv2} to $(\run{3},m)$ and $(\N\land\cO)$-$\boxdot$-reach
a point $(\run{4},m_{4})$  where $\I,(\run{4},m_{4}) \models
\inv_{1}(i_4,j_4,k_4)$,
$\{i_4,j_4,k_4\} = \{i,j,k\}$, and $\N(\run{4}) \cup \{k_4\} =
\N(\run{3}) \cup \{k_3\} = \N(r) \cup \{k\}$,
$\run{4}[0\ldots m-1] = \run{3}[0\ldots m-1]$, $(q,m-1) \not
\lamport{\run{4}} (j_1,m_4)$, and $m_{4}$ is either $m$ or $m+1$. 
At this point we apply the construction of Lemma~\ref{lemma:zchaincut}
to produce from $\run{4}$ a run  
$\run{5}$ such that $(\run{5},m_4)$ is 
$(\N \land \cO)$-$\boxdot$-reachable from $(\run{4},m_4)$ via agent $j_4$, 
$\I,(\run{5},m_{4}) \models \inv_{1}(i_4,j_4,k_4)$, 
$\N(\run{5} \cup \{k_4\} = \N(\run{4}) \cup \{k_4\} = \N(r) \cup \{k\}$,  
all \zchain{s} are visible to $(j_4,m_4)$, and
$\edges(\run{5},m-1) \subseteq \edges(\run{4},m-1)
 \subseteq \edges(r,m-1) \setminus (\bigcup_{q\in \Agents \setminus (\N(r) \cup \{k\})} ~\faultyagin_q(r,m-1))$.

We claim that $q$ does not decide in round $m-1$ of $\run{5}$. 
Obviously, we do not have an edge from $(q,m-2)$ into $(\N(r) \cup \{k\}) \times \{m-1\}$ in $\run{4}$, 
otherwise $r$ would have a  \zchain of length at least $m-1$.
We also do not have an edge in $\run{4}$ from $(q,m-2)$ into $(\N(r)
\cup \{k\})\times \{m-1\}$,  
since all such edges have been eliminated. 
Thus, the only edge from $(q,m-2)$ is to $(q,m-1)$. 
Since $(q,m-1) \not \lamport{\run{4}} (j_4,m_4)$, 
it follows that $(q,m-2) \not \lamport{\run{4}} (j_4,m_4)$, 
and the construction of $\run{5}$ ensures that   $\I,(\run{5},m-2) \models \neg \deciding_{q} = 0$.

If $m_4 = m+1$, we apply Lemma~\ref{lemma:rollback} to obtain a 
point $(\run{6},m)$ that is  \reach from $(r,m)$ such that 
$\edges(\run{6},m-1) \subseteq \edges(\run{5},m-1)$,
$\I,(\run{6},m) \models \inv_{1}(i_6,j_6,k_6)$,  
$\{i_6,j_6,k_6\} = \{i,j,k\}$, $\N(\run{6}) \cup \{k_6\} =
\N(r) \cup \{k\}$, and  
$q$ does not decide in round $m-1$. 
And if $m_4 = m$, we already have this conclusion with
$\run{6} = \run{5}$.  
 
We successively apply this construction for all $q \in \Agents \setminus (\N(r) \cup \{k\})$
until we have an $(\N\land\cO)$-$\boxdot$-reachable point $(\run{7},m)$
such that
$\I,(\run{7},m) \models \inv_{1}(i_7,j_7,k_7)$, $\{i_7,j_7,k_7\} =
\{i,j,k\}$,
$\N(\run{7}) \cup \{k_7\} = \N(\run{6}) \cup \{k_6\} = \N(r)\cup \{k\}$,
 $\edges(\run{7},m-1) \subseteq \edges(\run{4},m-1)$, and 
$\I,(\run{7},m-2) \models \neg \deciding_{q} = 0$ for all $q \in \Agents \setminus (\N(r) \cup \{k\})$. 
Since $\I,(\run{7},m) \models \inv_{1}(i_7,j_7,k_7)$ implies that 
no agent in $\N(\run{7}) \cup \{k_7\} = \N(r) \cup \{k\}$ decides in round $m-1$, we conclude that 
$\run{7}$ has  no \zchain{s} of length $m-2$. 
We therefore have the result with $r' = \run{7}$. 
\end{proof}

\begin{lemma} \label{lemma:dir1-induction}
  If $\I,(r,m) \models \inv_{1}(i,j,k)$, then there exists a run $\hat{r}$ such 
  that $(\hat{r}, 0)$ is $(\N \land \cO)$-$\boxdot$-reachable from $(r,m)$ and 
  $\I,(\hat{r},0) \models \neg \exists 0$.
\end{lemma}
\begin{proof}
  We proceed by induction on $m$.
Suppose that $\I,(r,m) \models \inv_{1}(i,j,k)$, so that $j$ decides 1 in
  round $m+1$ of $r$.  
  No agent can decide 1 in round 1, so we have $m\geq 1$.

  If $m = 1$, agent $j$ can't know $\exists 0$ at time 1 right before
  $j$ decides   1 in round 2, as otherwise $j$ would have gotten a
\zchain and, by Lemma~\ref{lemma:chain}, would decide 0. Hence,
  $j$ considers possible the run
  $\hat{r}$  where no agent  
  has an initial preference of 0 and the message pattern is identical to $r$.
Clearly, $(\hat{r},1)$ is $(\N\land\cO)$-$\boxdot$-reachable from $(r,1)$ 
  through agent $j$.

  If $m > 0$, suppose that $\len_)j(r,m) = m'$.
  Note that $j$ can't know about 
  a \zchain of length $m-1$ or greater, as in that case $j$ would decide 
  0 in round $m+1$.
  If $m' < m-2$, then we can immediately apply 
      Lemma~\ref{lemma:rollback} to $(\N\land\cO)$-$\boxdot$-reach a point 
      $(r^{*},m-1)$ where the invariant condition for 1 holds. We can then apply 
      the inductive hypothesis to conclude that there exists a run $\hat{r}$
      such that $\I,(\hat{r},0) \models \neg \exists 0$ and $(\hat{r},0)$ is 
      $(\N\land\cO)$-$\boxdot$-reachable from $(r^{*},m-1)$. By the 
      transitivity of the $(\N\land\cO)$-$\boxdot$-reachability relation,
      $(\hat{r},0)$ is also $(\N\land\cO)$-$\boxdot$-reachable from
      $(r,m)$ and the claim holds.
        And if $m' = m-2$, 
    then by Lemma~\ref{lem:shortenchain}, there exists a point $(r',m)$ that is \reach from $(r,m)$ 
    such that $\I,(r',m)\models \inv_1(i_1,j_1,k_1)$ and
        all \zchains in $r'$ have length less than $m-2$.  
    The previous case applied to $(r',m)$ then yields the result.
\end{proof}

\begin{lemma} \label{lemma:o-contra}
If  $\I, (r,m) \models (i \in \N \land \Circ (\decided_{i} = \bot))$,
then $i$ considers a point $(r',m)$ possible at $(r,m)$ where 
  $\I,(r',m) \models \inv_1(i,j,k)$.
\end{lemma}
\begin{proof}
  Suppose that at some point $(r,m)$, $\I,(r,m) \models i \in \N \land \Circ 
  (\decided_{i} = \bot)$. We first observe that we must have $\I,(r,m) \models 
  \neg C_{\N}(\faultyag)$, since otherwise, by Lemma~\ref{lemma:cfaulty}, 
  we get a contradiction.
To show that
  $i$ considers 
  a point $(r',m)$ possible at $(r,m)$ where 
  $\I,(r',m) \models \inv_1(i,j,k)$,
  we consider two cases:

  \begin{itemize}%
    \item 
    Suppose 
that $\I,(r,m) \models \neg K_{i}(\faultyag)$. By definition, there 
exists an agent $k$ such that $\I,(r,m) \models \neg K_{i} \neg 
    (k \in \N) \land \neg K_{i} \neg (k \not \in \N)$. Let $r'$ be a run 
    identical to $r$ except $k \not\in \N$, $k$ does not exhibit any
    nonfaulty behavior in $r'$, 
      some agent $j \in \N$ hears from all other agents in round $m$,
      and the only \zchains in run $r'$ are the ones that $i$ knows
      about.
      This means that the \zchains in $r'$ have length at most
      $m-2$, so eventually all agents should decide 1.
      We claim that $(r,m) \sim_{i} (r',m)$. Suppose by way of 
      contradiction that $i$ has different local states in $(r,m)$ and $(r',m)$.
      Then either $i$ hears in $r'$ from an agent that didn't receive
      a message from  
      $k$ in round $m$ of $r$ that heard from $k$ in round $m$ of $r'$
      or $i$ hears about a \zchain in round $m$ of $r'$ that it did not know
      about at time $m$ in $r$.
      In the first case, $i$ would have known that $k$ is faulty in 
            $(r,m)$, which is a contradiction; the second case 
      contradicts the assumption that $i$ doesn't know about such
      \zchains in $r'$.

      By Corollary~\ref{lemma:deciding1}, agent $j$ 
      decides 1 in round $m+1$ upon hearing from all other agents in round $m$.
Thus, $\I,(r',m) \models \inv_1(i,j,k)$ and $(r,m) \sim_i (r',m)$.

    \item 
      If
    $\I,(r,m) \models K_{i}(\faultyag)$, we make use of the observation 
      that $\I,(r,m) \models \neg C_{\N}(\faultyag)$. By
 Proposition~\ref{lemma:dfaulty}, it must be the case that $\I,(r,m-1) \models 
 \neg \dist_{\N}(\faultyag)$; equivalently, $\I,(r,m) \models \forall
      j \in \N (\neg K_{j} \neg (k \in \N))$ for some faulty agent $k$. 
      Since $i \in \N$, we must have $\I,(r,m-1) \models \neg
      K_{i}\neg (k \in \N)$. 
Since $\I,(r,m) \models K_{i}(\faultyag)$, it must be the case that agent $i$ learned that $k$ is faulty in 
      round $m$, either by not getting a message from $k$ for the first time or 
      by getting a message from some faulty agent $k'$ that knew about agent $k$
      being faulty. Agent $k'$ can't be a nonfaulty agent, since no 
      nonfaulty agent knows that $k$ is faulty at time $m-1$. 

      If $i$ learned that $k$ is faulty by not 
      hearing from $k$ in round $m$, then at $(r,m)$, agent $i$
      considers possible the point $(r',m)$, where in round $m$ of $r'$,   
      some nonfaulty agent $j$ received a message from all other agents.
      By Corollary~\ref{lemma:deciding1}, in $r'$, 
      $j$ decides 1 in round $m$. 
      Thus, $\I,(r',m) \models \inv_1(i,j,k)$, as desired.

      If $i$ learned that $k$ is faulty by hearing 
      from some faulty agent $k'$ that $k$ is faulty, then $i$
      considers possible the run $r'$ 
      which is identical to $r$ except that $j$ and $k$ receive
      a message from all other agents except $k'$ in round $m$.
      By Corollary~\ref{lemma:deciding1}, both agents $j$ and $k$ are 
      about to decide 1 at $(r',m)$,  since all \zchains in $r'$
      have length at most $m-2$ and 
      they both receive a message from all agents but $k'$. 
      We must have $(r,m) \sim_i (r',m)$, as the only change 
      between the runs is in round $m$, and $i$ receives the same
      messages in both runs.  Moreover, $i$ is nonfaulty in both runs.

      Now consider a run $r''$ where $i$ does not hear from 
      $k'$ in round $m$ or later, and otherwise receives all messages from 
      the other agents. Thus, $i$, $j$, and $k$ hear from all agents
      but $k'$ in rounds  
      $m$ and $m+1$ in $r''$. As $j \in (\N\land\cO)(r'',m') \cup 
      (\N\land\cO)(r',m')$, $(r'',m')$ is $(\N\land\cO)$-$\boxdot$-reachable 
      from $(r',m)$. 

      Since $j$ does not hear from $k'$ after round $m$ in $r''$, it does not 
      discover that $k$ is faulty. Therefore, there exists a run $r'''$ 
      that is identical to $r'$ except that agent $k$ does not exhibit
      nonfaulty behavior.  Thus,
            $\I,(r''',m') \models \inv_1(i,j,k)$. Moreover,
      $j \in (\N\land\cO)(r''',m') \cup (\N\land\cO)(r'',m')$, 
so            $(r''',m')$ is $(\N\land\cO)$-$\boxdot$-reachable 
      from $(r'',m')$.
  \end{itemize}
\end{proof}

\commentout{
\begin{lemma} \label{lemma:dir1}
  $\I \models i \in \N \Rightarrow (B_i^\N (\exists 0 
  \wedge C^\boxdot_{\N \land \cO}\exists 0 \wedge \neg (\dec_i = 1)) \Rightarrow
  (\dec_i = 0))$.
\end{lemma}
\begin{proof}
  To show that this formula is valid in this system, we show that it holds for 
  all points $(r,m)$. We first assume that $\I,(r,m) \models i \in \N$.
  Then, we show that $\I,(r,m) \models B_i^\N (\exists 0 \wedge 
  C^\boxdot_{\N \land \cO}\exists 0 \wedge \neg (\dec_i = 1)) \Rightarrow 
  (\dec_i = 0)$ by contradiction.

  Suppose $\I,(r,m) \models B_i^\N (\exists 0 \wedge C^\boxdot_{\N \land \cO}
  \exists 0 \wedge \neg (\dec_i = 1))$ but $\I,(r,m) \models \neg 
  (\dec_{i} = 0)$. Then, since agent $i$ has neither decided nor decides
  0 or 1 in round $m$, $\I,(r,m) \models \Circ (\decided_{i} = \bot)$. Then, 
  we can use Lemma~\ref{lemma:o-contra} to conclude that $\I,(r,m) \models \neg 
  B_{i}^{\N}(C^\boxdot_{\N \land \cO}\exists 0)$, and we get the desired 
  contradiction.
\end{proof}
}

Finally, we construct the points $(r_i^{0'},m)$ and
$(r^{1'},m)$ required for weak safety.
The argument for this case is simpler than that of the 
previous case.
Most of the work for this case is done by
Lemmas~\ref{lemma:dir2-induction} and 
\ref{lemma:z-contra}, which are analogues of
Lemmas~\ref{lemma:dir1-induction} and~\ref{lemma:o-contra}.

\begin{lemma} \label{lemma:dir2-induction}
  If $\I,(r,m) \models \inv_{0}(i,j,k)$,
  then there exists a run $\hat{r}$ such that $(\hat{r}, 0)$ is
  $(\N \land \Z)$-$\boxdot$-reachable from $(r,m)$ and $\I,(\hat{r},0) \models 
  \neg \exists 1$.
\end{lemma}
\begin{proof}
  We proceed by induction on $m$.

  If $m = 0$, then since agent $j$ decides 0 in round 1, it must be the
  case that  
  $\I,(r,0) \models \init_{j} = 0$. Hence, agent $j$ considers a run $\hat{r}$ 
  where every agent has initial preference 0 possible.
Clearly, $(\hat{r},0)$ is $(\N\land\Z)$-$\boxdot$-reachable from $(r,0)$.   
If $m > 1$, since agent $k$ is faulty but does not exhibit faulty
behavior in $r$, $\I,(r,m) \models \neg C_{\N}(\faultyag)$.
By Lemma~\ref{lemma:chain}, 
  agent $j$ 
  must have received at least one \zchain of length $m$ in order to
  decide 0.  
  Let $k'$ be an agent from whom $j$ received a \zchain.
  Note that $k'$ 
  can't be $i$  
    or $j$, since $i$ and $j$ do not decide 
    before round
     $m+1$. Let $r'$ be a run identical to $r$, except that $i$ is faulty
  instead of $k$, $k$ hears from $k'$, and $k$ does not hear from
  $i$ in  
    round $m$ of $r'$. Since $j \in (\N\land\Z)(r,m) \cap
  (\N\land\Z)(r',m)$ and $j$ has  
  the same local state at both points, $(r',m)$ is 
  $(\N\land\Z)$-$\boxdot$-reachable from $(r,m)$.

  In $r'$, agent $k$ is nonfaulty and receives a \zchain
  from $k'$ in round $m$. Let $k^*$ be an agent that sent a
  \zchain to $k'$ in round  
  $m-1$. Consider a run $r''$ that is identical to $r'$ 
  except that agent $i$ receives a \zchain in round $m-1$ from $k^*$ and
    agent $j$ receives a message from agent $i$ in round $m$.
  (If $m=1$, then we take $r''$ to be a run where $i$ has initial
  preference 0 and $j$  
  receives a message from $i$ in round 1.) Agent $i$ decides 0 in round 
  $m$ of $r''$. Since $k$ has the same local state in both $(r',m)$ and
  $(r'',m)$,  
  and $k \in (\N\land\Z)(r',m) \cap (\N\land\Z)(r'',m)$, $(r'',m)$ is 
  $(\N\land\Z)$-$\boxdot$-reachable from $(r',m)$.
  Because $j$ received a message from $i$ in round $m$ of $r''$ and $j$
is nonfaulty, $j$ decides 0 in round $m+1$ of $r''$.

Let $r'''$ be a run that is identical to $r''$ except that agent
$i$ is nonfaulty. Then $j$ has the same
local state 
  in both $(r'',m)$ and $(r''',m)$, and $j \in (\N\land\Z)(r'',m) \cap 
  (\N\land\Z)(r''',m)$.
  Thus, $(r''',m)$ is $(\N\land\Z)$-$\boxdot$-reachable from $(r,m)$ and 
  in $r'''$, agent $i$ decides 0 in round $m$, $i \in \N$, $j \in \N$,
  $k \not\in \N$, and $k$ acts nonfaulty throughout run $r'''$.
  Moreover, $\I,(r''',m-1) \models \inv_0(j,i,k)$.
  By the inductive 
  hypothesis, it follows
  that there exists a run $\hat{r}$ such that $(\hat{r},0)$ is 
  $(\N\land\Z)$-$\boxdot$-reachable from $(r''',m-1)$ and $\I,(\hat{r},0) 
    \models \neg \exists 1$. By the transitivity of the 
  $(\N\land\Z)$-$\boxdot$-reachability relation, $(\hat{r},0)$ is also
  $(\N\land\Z)$-$\boxdot$-reachable from $(r,m)$ and the claim holds.
\end{proof}

\begin{lemma} \label{lemma:z-contra}
If  $\I, (r,m) \models (i \in \N \land \Circ (\decided_{i} = \bot))$,
then $i$ considers a point $(r',m)$ possible at $(r,m)$ where 
  $\I,(r',m) \models \inv_0(i,j,k)$.
\end{lemma}
\begin{proof}
  Suppose that $\I,(r,m) \models i \in \N \land \Circ 
  (\decided_{i} = \bot)$. We must have $\I,(r,m) \models 
  \neg C_{\N}(\faultyag))$, since otherwise, by Lemma~\ref{lemma:cfaulty}, 
  we get a contradiction. We next show that, whether or not
  $i$ knows which agents are faulty, $i$ considers 
  a point $(r',m)$ possible at $(r,m)$ from which a point $(r''',m)$
  is   $(\N \land \cO)$-$\boxdot$-reachable from $(r,m)$ such that 
  $\I,(r''',m) \models \inv_0(i,j,k)$.

  \begin{itemize}%
  \item If $\I,(r,m) \models \neg K_{i}(\faultyag)$ then, by definition, there 
    exists a $k \in \Agents$ such that $\I,(r,m) \models \neg K_{i} \neg 
    (k \in \N) \land \neg K_{i} \neg (k \not \in \N)$.  Let $r'$ be a run 
        identical to $r$ except that $k \not\in \N$ and $k$ does not
        exhibit faulty behavior in $r'$.
    We claim that $(r,m) \sim_{i} (r',m)$. Suppose by way of 
      contradiction that $i$ has different local states in $(r,m)$ and $(r',m)$.
Since the only difference between these runs are the blocked 
      messages sent by $k$ that are no longer blocked, $i$ must have heard from 
      an agent (either directly or indirectly) that did not receive a message 
      from $k$ in $r$ but did in $r'$. But in that case, $i$ would have known
            that $k$ is faulty in $r$, which is a contradiction.

            Since $i$ has not yet decided by $(r',m)$ and does not decide 
in round $m$ (as $i$ has the same local state at $(r,m)$ and $(r',m)$),
     we have that $\I,(r',m) \models \neg K_{i}(\bigwedge_{j \in \Agents} \neg 
      (\deciding_{j} = 0))$. That is, $i$ considers it possible that
      there exists  an agent $j'$ that decides 0 in round $m+1$. 
      Let $r''$ be the run that is identical to $r'$ except that $j'$
    decides 0 in round $m+1$ of $r''$.  By
    Lemma~\ref{lemma:chain}, if $m \ge 1$, $j'$ receives a \zchain in $r''$
    from some agent $k'$ in round $m$, and if $m=0$, $\init_{j'} = 0$. Let 
      $r'''$ be a run that is identical to $r''$ except that agent $j$
    (who is also nonfaulty) receives a message from $k'$ in round $m$
      if $m \ge 1$ and has 
            initial preference 0 if $m=0$. Hence, agent $j$ decides 0 in
    round $m+1$ of $r'''$. By the  
      transitivity of the knowledge relation, we then have 
      $(r,m) \sim_{i} (r''',m)$.

    \item If $\I,(r,m) \models K_{i}(\faultyag)$, we use the observation 
            that $\I,(r,m) \models \neg C_{\N}(\faultyag)$.
      As in the proof of Lemma~\ref{lemma:o-contra}, we must have that
       $\I,(r,m-1) \models \neg K_{i}\neg (k \in \N))$ for some
      nonfaulty agent $k$.

           It must be the case that agent $i$ learned that $k$ is faulty in 
 round $m$, due to either not getting a message from $k$ for the first time or 
 getting a message from some faulty agent $k'$ that knew that $k$
 was faulty. (Agent $k'$ can't be nonfaulty, since no 
 nonfaulty agent knows that $k$ is faulty at time $m-1$.) 
      Since agent $i$ has not decided yet and does not decide in round 
      $m+1$ of $r$, agent $i$ knows that the other nonfaulty agents
      did not decide 
      in an earlier round. Moreover, in $(r,m)$, $i$ considers a point 
      $(r',m)$ possible where $\I,(r',m) \models (\deciding_{j'} = 0)$ for some 
            agent $j'$, since $i$ does not decide 1 in round $m+1$ of $r'$.

            If agent $i$ learned that $k$ is faulty due to not 
            hearing from $k$ in round $m$ of $r'$, then at the point $(r',m)$,
            agent $i$ considers 
      $(r'',m)$ possible,  
            where $r''$ is such that some nonfaulty agent $j$ received a
       message from $k$ in round $m$
       and $k$ decides 0 in round $m-1$ (due to hearing from the agent
       that sent a 
      message to the agent that sent a message to $j'$ in the \zchain). 
      In $r''$, $j$ receives a \zchain in round $m$ and thus decides
      0 in round  $m+1$.
      It follows that $\I,(r'',m) \models \inv_0(i,j,k)$.
      
      If agent $i$ learned that $k$ is faulty due to hearing 
      it from some faulty agent $k'$, then agent $i$ considers a point
      $(r'',m)$ possible at $(r',m)$, where $r''$ is identical to $r'$
      except that in $r''$, agent $j$ does not  
      hear from agent $k'$ at or after round $m$ and receives a \zchain 
from agent $j'$ in round $m$.
If $k \ne k'$, then again we have $I,(r'',m) \models \inv_0(i,j,k)$.

If $k = k'$, then $I,(r'',m) \not\models \inv_0(i,j,k)$,   because 
$k$ does not act nonfaulty throughout run $r''$.  In that case,
consider a run $r^*$ that is identical to $r''$ except that
agent $k$ sends a message to all agents but $j$ in round $m$ of $r^*$. 
      Since $j \in (\N\land\Z)(r'',m) \cap (\N\land\Z)(r^*,m)$ and 
      $j$ has the same local state at both points, $(r^*,m)$ is 
      $(\N\land\Z)$-$\boxdot$-reachable from $(r'',m)$. Let 
      $r'''$ be the run that is identical to $r^*$ except that agent $k$
does not exhibit faulty behavior in $r'''$ and $i$ receives a \zchain
from $j'$.  
      By construction, $\I,(r''',m) \models \inv_0(i,j,k)$
      Since agent $i$ did not know $k$ was faulty in $r^*$,
      $i \in (\N\land\Z)(r^*,m) \cap (\N\land\Z)(r''',m)$
      and $(r''',m)$ is $(\N\land\Z)$-$\boxdot$-reachable from
      $(r^*,m)$.
Again, this suffices for the desired result.
        \end{itemize}
\end{proof}
This completes the proof that $\kbp^1$ is weakly safe
and we get the following result:
\rethm{thm:kbp1-safe} 
  $\kbp^1$ is weakly safe with respect to $\gamma_{\fip,n,t}$.
\erethm

\subsubsection{Implementation}
We now explicitly define the full-information context using
communication graphs 
similar to those used by Moses and Tuttle \cite{MT}.
Intuitively, a communication graph for agent $i$ is a labeled graph
that provides a compact
description of all messages sent and received by agent $i$.  Formally,
the \emph{communication graph} $G_{i,m}$ is defined as
follows.  The set  $V(G_{i,m})$ of vertices consists of all pairs of
the form $(j,m')$ for all 
    $m' \leq m$ and agents $j$; the set $E(G_{i,m})$ of edges consists of
    all edges 
    from $(j,m'-1)$ to $(j',m')$ for $j,j' \in \Agents$ and $m' \leq
    m$; there is a message label $l_{G_{i,m}} \in \{0,1,?\}$ for each
    edge (see below), and an initial  
    preference label $p_{G_{i,m}} \in \{0,1,?\}$ for each $j \in
    \Agents$ (which can     be viewed as a label on vertices of the
    form $(j,0)$).     
    An edge from $(j,m'-1)$ to $(j',m')$ is labeled with a 1 if $i$ knows
    that $j$ sent a message to $j'$ in round $m'$; it is labeled with
    a 0 if $i$ knows that $j$ did not send a message to $j'$ in round
    $m'$; and it is labeled with a ? if $i$ does not know whether $j$ sent a
    message to $j'$ in round $m'$.  Note that with a full-information
    protocol, if $i$ knows the initial preferences of agents, and
    which agents sent round $m''$ messages for $m'' < m'$, then it is
    easy for $i$ to figure out what the content of a message that was
    sent would be.
    A preference label of $v \in \{0,1\}$ on $(j,0)$
    indicates that $i$ knows that $j$'s initial preference was $v$,
    while a label of $?$ indicates that $i$ does not know $j$'s
    initial preference. 
    We write $G_{i,m}(r)$ for the communication graph of agent $i$ at time $m$ 
        in a run $r$ in a full-information exchange and
        $\mathcal{G}_{i,m}$ for the set 
    of all time-$m$ communication graphs for agent $i$.

    \renewcommand{\G}{\mathcal{G}}

    Let $\exchange_{\fip}(n)$ be the 
full information-exchange protocol for $n$ agents, where for each agent $i$,
the following hold:
\begin{itemize}
  \item The local states have the form $\langle \Time_i, \decided_i, \init_i,
    G_{i,\Time_i} \rangle$, where $G_{i,\Time_i}$ is a communication graph.
    
  \item The initial local states of each agent $i$ have the form 
    $\langle 0, \bot, \init_i, G_{i,0} \rangle$, where $G_{i,0} \in
    \G_{i,0}$.  (Note that that in $G_{i,0}$, we must have 
    $p_{G_{i,0}}(j) = \;?$ for all agents $j \ne i$ and 
    $p_{G_{i,0}}(i) = \init_i$.)
    
  \item $M_i =  \G_{i,\Time_i}$. 
  \item For all agents $j$ and actions $a$, $\mu_{ij}(\langle \Time_i, 
    \decided_i, \init_i, G_{i,\Time_i} \rangle,a) = G_{i,\Time_i}$.
  \item $\delta_i(\langle \Time_i, \decided_i, 
    \init_i, G_{i,m} \rangle,a,(m_1,\dots,m_n)) = 
 \langle \Time_i+1, \decided'_i, \init_i, G_{i,\Time_i+1} \rangle$,
    where $G_{i,\Time_i+1} \in \G_{i,\Time_i+1}$ is obtained by adding
    vertices and edges for round
    $\Time_i+1$ and combining the labels from all graphs that were 
    received by $i$ and $G_{i,\Time_i}$. 
More precisely, if $RG_i$ consists of all the graphs that $i$ received
up to and including round $\Time_i+1$, then
    \begin{align*}
          l_{G_{i,\Time_i+1}}((j,m),(j',m+1)) &= \begin{cases}
            v & \mbox{if } \exists  G \in RG_i ( l_G((j,m),(j',m+1)) = v) 
            \land v \in \{0,1\}, \\
            1 & \mbox{if } m = \Time_i \land j' = i \land m_j \ne \bot, \\  
            0 & \mbox{if } m = \Time_i \land j' = i \land m_j = \bot, \\
            ? & \text{otherwise.}
          \end{cases} \\
          p_{G_{i,\Time_i+1}}(j) &= \begin{cases} 
            v & \mbox{if } \exists G \in S, v \in \{0,1\}( (p_{G}(j) = v)),\\
            ? & \text{otherwise}.  
          \end{cases}
    \end{align*}
    
        Finally, $\decided_i'$ is determined by the action $a$, just
    as in the standard EBA context. 
\end{itemize}

Let $\gamma_{\fip,n,t} = (\exchange_{\fip}(n), SO(t), \pi_{\fip,n})$,
where $\pi_{\fip,n}$ interprets the standard propositions in the
standard way.
To  check the knowledge conditions in $\kbp^1$, it is useful to define
the following sets, which can be computed in polynomial time from the
communication graph: 
\begin{itemize}
\item For $m' < m$, $f(j,m',G_{i,m})$ is the set of faulty agents that
  $i$ knows that $j$ knows about at time $m'$, given $G_{i,m}$.  The 
set 
$f(j,m',G_{i,m})$ is the union of (a) $f(j',m'-1,G_{i,m})$ for all
  $j'$   that sent a message to $j$ in round $m'$ in $G_{i,m}$ if $m' >0$, (b)
   $\{j'\}$ for each agent $j'$ that did
    not send a message to $j$ in round $m'$ in $G_{i,m}$, and (c) 
    $f(j,m'-1,G_{i,m})$ if $m' > 0$.  (Note that $f(j,0,G_{i,m}) = \emptyset$.)
  \item For $m' < m$,   $D(S,m',G_{i,m})$ is the set of faulty agents
    that $i$ knows that the agents in $S$ know about at time $m'$,
    given $G_{i,m}$.  
 $D(S,m',G_{i,m}) = \cup_{k \in S} f(k,m',G_{i,m})$.
\end{itemize}

In addition to $f$ and $D$, agent $i$ can compute 
the actions of each agent $j$
at time $m' < m$ if 
$(j,m') \lamport{r} (i,m)$ where $r$ is the current run, 
since we are using
a full-information protocol. Let $d(j,m',G_{i,m}) \in \{0,1,\bot,?\}$ represent
what $i$ knows about the action of agent $j$ 
in round $m'+1$. If $d(j,m',G_{i,m}) = v \in \{0,1\}$, then $i$ knows
that  that $j$ decides $v$ 
in round $m'+1$; if $d(j,m',G_{i,m}) = \bot$, then $i$ knows that $j$
      does not decide in  
round $m'+1$; finally, if $d(j,m',G_{i,m}) = \;?$, then 
$(j,m') \not\lamport{r} (i,m)$.

The set of values known by each agent that sent a message 
either directly or indirectly to $i$ can be also computed from the
communication  graph in polynomial time.
Let $V(j,m',G_{i,m})$ be the set of values that $i$ knows that $j$
knows about at time $m'$
if $(j,m') \lamport{r} (i,m)$ 
according to $G_{i,m}$ and $\emptyset$ otherwise.

We next define families $\common_v(i,m,G_{i,m})$ and
      $\cond_1(i,m,G_{i,m})$ of Booleans  that can also be
      computed in polynomial time. 
\begin{definition}[$\common_v$ and $\cond_v$]
  Intuitively, $\common_v$ holds if $K_i(\dist_N(\faultyag)$ 
  $\land$\newline
  $(\bigwedge_{j \in \N} (K_j \Circ \nodecided_j(v)))$ $\land$
  $(\bigvee_{j \in \N} K_j (\exists v)))$ holds at time $m$ (which
  means that $C_\N(\faultyag \land \nodecided_\N(1) \land \exists v)$
  holds at time $m+1$), 
  given $G_{i,m}$.  We compute $\common_v$ as follows.
   If either of the following three conditions hold, then
  $\common_0(i,m,G_{i,m}) = 
 \mathit{false}$:
  \begin{itemize}
 \item $|D(\bar{f}(i,m,G_{i,m}),m-1,G_{i,m})| \ne t$,
where $\bar{f}(j,m',G_{i,m}) = \Agents - f(j,m',G_{i,m})$ (these are
the agents that $i$ thinks might be nonfaulty at time $m'$, given $G_{i,m}$);
\item there exists an agent $j \not \in f(i,m,G_{i,m})$
  such that $d(j,m',G_{i,m}) = 1-v$ for some $m' < m$.
  \item  for all agents $j \notin
    D(\bar{f}(i,m,G_{i,m}),m-1,G_{i,m})$,
      $v \not\in V(j,m-1,G_{i,m})$
\end{itemize}
  Otherwise, $\common_v(i,m,G_{i,m}) = \mathit{true}$.  

The first condition for taking $\common_v(i,m,G) = \mathit{false}$ corresponds
to agent  $i$ thinking that the agents who
might be nonfaulty at time $m$ do not have distributed knowledge
of $t$ faulty agents at time $m - 1$.  If so, certainly the agents
who are actually nonfaulty will not have distributed knowledge
of who the faulty agents are at time $m - 1$, so there will not
be common knowledge among the nonfaulty agents
of who the faulty agents are at time $m$—see Lemma 5.
The second condition holds if
some agent $j$ that that $i$ considers
  possibly nonfaulty at time $m$ has decided 
  $1-v$.  We can assume that $|D(\bar{f}(i,m,G_{i,m}),m-1,G_{i,m})| =
      t$ (otherwise $\common_v(i,m,G) = \mathit{false}$ by the first
      condition).    Thus, $i$ 
  knows who the nonfaulty agents are at time $m$, so $i$ knows that a
      nonfaulty agent 
  has decided $1-v$, so it cannot be common knowledge among the
  nonfaulty agents that no nonfaulty
      agent decided $1-v$.

  The Boolean $\cond_0 = \mathit{true}$ holds if the formula $\init_i
      = 0 
      \vee
      \bigvee_{j 
    \in \Agents} K_i (\jdecided_j = 0)$ holds at time $m$, given
      $G_{i,m}$.  Formally,   
  \begin{itemize}
    \item $\cond_0(i,0,G_{i,0}) = (\init_i = 0)$.
    \item For $m>0$,  $\cond_0(i,m,G_{i,m}) = \mathit{true}$  if 
      there exists an agent $j$ such that $d(j,m-1,G_{i,m}) = 0$ and 
      $l_{G_{i,m}}((j,m-1),(i,m)) = 1$;
      otherwise, $\cond_0(i,m,G_{i,m}) = \mathit{false}$.
  \end{itemize}
The Boolean $\cond_1(i,m,G_{i,m}) = \mathit{true}$ holds if $K_i(\bigwedge_{j \in
\Agents} \neg (\deciding_j = 0))$ holds at time $m$, given $G_{i,m}$.
      Formally,
  \begin{itemize}
    \item $\cond_1(i,0,G_{i,0}) = \mathit{false}$.
    \item For $m> 0$, let $m'$ be the latest time such that
      $d(j,m',G_{i,m}) = 0$ for some  agent $j$ (as usual,
$m' = -1$ if $d(j,m',G_{i,m}) \ne 0$ for all  agents $j$),
and 
let $m_{j}$ be the latest time that 
$(j,m_j) \lamport{r} (i,m)$, where $r$ is a run for which $G_{i,m}$
describes $i$'s view at time $m$ in $r$.  (There are many such runs;
it does not matter which one we choose, since they all agree on the
whether $(j,m_j) \lamport{r} (i,m)$.)  Intuitively, $m' = \len_j(r,m)$.
If,
for all $m''$ with $m' < m'' \leq m$, there exist at least $m'' - m'$ 
agents $j$ such that $d(j,m'',G_{i,m}) = \;?$, then 
$\cond_1(i,m,G_{i,0}) = \mathit{true}$; otherwise
$\cond_1(i,m,G_{i,0}) = \mathit{false}$.

  \end{itemize}
\end{definition}

Using these definitions, we can define an implementation of $\kbp^1$ in the 
full-information context. Let $P^{\opt}$ be the EBA decision protocol 
implemented by the following program:

\begin{program}
  \DontPrintSemicolon
  \lIf{$\decided_i\neq \bot$}{$\noop$}
    \lElseIf{$\common_0(i,\Time_i-1,G_{i,\Time_i})$}{$\decide_i(0)$}
  \lElseIf{$\common_1(i,\Time_i-1,G_{i,\Time_i})$}{$\decide_i(1)$}
  \lElseIf{$\cond_0(i,\Time_i,G_{i,\Time_i})$}{$\decide_i(0)$}
  \lElseIf{$\cond_1(i,\Time_i,G_{i,\Time_i})$}{$\decide_i(1)$}
  \lElse{$\noop$}
  \caption{$P^{\opt}_i$}
\end{program}

The following lemma shows that the initial conditions that are checked in 
the definition of $\common_v$ correspond to checking for $C_\N(\faultyag)$. 
\begin{lemma} \label{lemma:f=d=t}
  $|f(i,m,G_{i,m}(r))| = 
  |D(\bar{f}(i,m,G_{i,m}(r)),m-1,G_{i,m}(r))| = t$ for some agent $i$
  if and only if 
 $\I_{\gamma_{\fip,n,t},\kbp^1},(r,m) \models C_\N(\faultyag)$.
\end{lemma}
\begin{proof}
  Let $\I = \I_{\gamma_{\fip,n,t},\kbp^1}$ and $G = G_{i,m}(r)$.
    Suppose that $|f(i,m,G)| = |D(\bar{f}(i,m,G),m-1,G)| 
    = t$ for some agent $i$. We first observe that 
    $\Agents - f(i,m,G) = \N$, as there are $t$ faulty agents.
  This implies that the set $D(\bar{f}(i,m,G),m-1,G)$ is the 
  set of all faulty agents that are known by the nonfaulty agents at time $m-1$. 
  Since $|D(\bar{f}(i,m,G),m-1,G)| = t$, it must be the case that 
  for all faulty agents $j$, there exists a nonfaulty agent that knows that 
  agent $j$ is faulty. Hence, $\I,(r,m) \models \ominus \dist_\N(\faultyag)$. By 
Proposition~\ref{lemma:dfaulty}, we have $\I,(r,m) \models C_\N(\faultyag)$. 

  Conversely, suppose that $\I,(r,m) \models C_\N(\faultyag)$. Again, by 
  Proposition~\ref{lemma:dfaulty}, we get that $\I,(r,m) \models \ominus 
  \dist_\N(\faultyag)$. By definition, the union of all faulty agents known
by nonfaulty agents at time $m-1$ is the set of all faulty agents. Hence, 
$|D(\N,m-1,G)| = t$.
Therefore, $\N = \bar{f}(i,m,G)$.  
\end{proof}

\begin{theorem} \label{theorem:implementation}
  If $n - t \geq 2$, then $P^{\opt}$ implements $\kbp^1$ in the full-information
  EBA context $\gamma_{\fip,n,t}$. 
\end{theorem}
\begin{proof}
  Let $\I = \I_{\gamma_{\fip,n,t},\kbp^1}$ and $G = G_{i,m}(r)$. 
  We show that for all points $(r,m)$, 
  $P^{\opt}_i(r_i(m)) = (\kbp^1_i)^\I(r_i(m))$.
  \begin{itemize}%
    \item If $P^{\opt}_i(r_i(m)) = \noop$ because
     $\decided_i \ne \bot$ in $r_i(m)$, we clearly also have 
     $(\kbp^1_i)^\I(r_i(m)) = \noop$, because 
     $\I,(r,m) \models K_i(\decided_i \ne \bot)$.
    \item If $P^{\opt}_i(r_i(m)) = \decide_i(0)$ by the second line,
            we must have $\decided_i = \bot$ in $r_i(m)$ and $m > 0$. 
    By the definition of $\common_0$, we must also have 
    (a) $|D(\bar{f}(i,m-1,G),m,G)| = t$, 
      (b) $d(j,m',G) \ne 1$ for all $j \not\in f(i,m,G)$ 
    and $m' < m$, and 
    (c) $0 \in V(j,m-1,G)$ for some $j \not\in f(i,m,G)$.
    From (a), it follows that $\I,(r,m) \models C_\N(\faultyag)$ 
    using Lemma~\ref{lemma:f=d=t}.
        From (b), it follows
 that no nonfaulty agent decides 1 at any round $m' 
      < m+1$. Hence, $\I,(r,m) \models \nodecided_\N(1)$.
      Finally, (c) implies that $i$ knows that a nonfaulty agent $j$ 
      had an initial preference 0 at time $m-1$; that is,
            $\I,(r,m) \models \ominus K_j(\exists 0)$.
            Combining these observations, using
            Proposition~\ref{lemma:dfaulty},  
      we can conclude that 
            $\I,(r,m) \models C_{\N}(\faultyag \land \nodecided_{\N}(1)
            \land \exists 0)$.
By Lemma~\ref{lemma:cfaultyimpliesall},
            $\I,(r,m) \models K_i(C_{\N}(\faultyag \land \nodecided_{\N}(1) 
 \land \exists 0))$, so $(\kbp^1_i)^\I(r_i(m)) = \decide_i(0)$.

     \item If $P^{\opt}_i(r_i(m)) = \decide_i(1)$ by the third line, we must have 
      $\decided_i = \bot$ in $r_i(m)$, $\common_0(i,m,G) = \mathit{false}$,
      $\common_1(i,m,G) =\mathit{true}$,       
      and $m > 0$.
As before, $\common_1(i,m,G) = \mathit{true}$ 
      implies that
      $\I,(r,m) \models C_\N(\faultyag) \land \nodecided_\N(0) \land
  \ominus K_j(\exists 1)$, so $\I,(r,m) \models
      K_i(C_{\N}(\faultyag \land \nodecided_{\N}(0) \land \exists 1))$.

            It thus suffices to show that $\I,(r,m) \models \neg 
      C_{\N}(\faultyag \land \nodecided_{\N}(1) \land \exists 0)$. 
            Since $\common_0(i,m,G) = \mathit{false}$, we must either have 
                  (a) $|D(\bar{f}(i,m-1,G),m-1,G)| \ne t$,
      (b) $|D(\bar{f}(i,m,G),m-1,G)| = t$ and 
for some $j \not\in f(i,m,G)$ and $m' < m$, $d(j,m',G) = 1$, or 
(c) $|D(\bar{f}(i,m-1,G),m-1,G)| = t$ and 
        for all $j \not\in f(i,m,G)$, $0 \not\in V(j,m-1,G)$.
        If (a) holds, then Lemma~\ref{lemma:f=d=t} implies that 
        $\I,(r,m) \models \neg C_\N(\faultyag)$.
        If (b) holds, then, as we have observed,  $\I,(r,m) \models \neg
        C_\N(\faultyag)$ and, in addition, 
        $d(j,m',G) = 1$ for some nonfaulty $j$ and 
      $m' < m$. Hence, $\I,(r,m) \models \neg \nodecided_{\N}(1)$.
        If (c) holds, then we again have $\I,(r,m) \models \neg \nodecided_{\N}(1)$; moreover,
        $0 \not\in V(j,m-1,G)$ for all nonfaulty $j$. Hence, $i$ 
      considers it possible that none of the nonfaulty agents knows
      about a 0 
      at time $m-01$ given $G$, so by Proposition~\ref{lemma:dfaulty},
      $\I,(r,m) \models \neg C_N(\exists
      0)$. 
      Thus, in all cases $\I,(r,m) \models \neg 
      C_{\N}(\faultyag \land \nodecided_{\N}(1) \land \exists 0)$,
      so $(\kbp^1_i)^\I(r_i(m)) = \decide_i(1)$.

    \item If $P^{\opt}_i(r_i(m)) = \decide_i(0)$ by the fourth line, we must 
      have $\decided_i = \bot$ in $r_i(m)$, $\common_0(i,m,G) = \common_1(i,m,G) =
      \mathit{false}$, and $\cond_0(i,m,G) = \mathit{true}$. 
Again, $\common_0(i,$ $m,G)  = \mathit{false}$
      implies that $\I,(r,m) \models \neg C_{\N}(\faultyag \land 
      \nodecided_{\N}(1) \land \exists 0)$. Similarly,
      $\common_1(i,m,G) = \mathit{false}$ implies that
      $\I,(r,m) \models \neg C_{\N}(\faultyag \land 
      \nodecided_{\N}(0) \land \exists 1)$. 
      We thus need to show only that $\I,(r,m) \models \init_i = 0 \lor 
      \bigvee_{j \in \Agents} K_i (\jdecided_j = 0)$. We proceed by induction on
   $m$. For the base case, if $\cond_0(i,0,G) = \mathit{true}$, by definition,
      it must be the case that $\init_i = 0$.
      For the inductive step, suppose that $\cond_0(i,m,G) = \mathit{true}$ for some $m > 0$. 
      By definition, this implies that $d(j,m-1,G) = 0$ and 
      $l_{G}((j,m-1),(i,m)) = 1$ for some agent $j$. 
      Thus, $(\kbp^1_i)^\I(r_i(m)) = \decide_i(0)$.
    Then agent $j$ decides 0 and agent $i$ hears from agent $j$ 
    in round $m$. It follows that
      $\I,(r,m) \models K_i(\jdecided_j = 0)$, so
    $(\kbp^1_i)^\I(r_i(m)) = \decide_i(0)$.

    \item If $P^{\opt}_i(r_i(m)) = \decide_i(1)$ by the fifth line,
      we must have $\decided_i = \bot$ in $r_i(m)$, $\common_0(i,$ $m,G) =
      \common_1(i,m,G) = \cond_0(i,m,G) = \mathit{false}$, and
      $\cond_1(i,m,G) = \mathit{true}$.
      We also have $m > 0$ since $\cond_1(i,m,G) = \mathit{true}$.
      As before, $ \common_0(i,m,G) = \common_1(i,m,G) =
      \mathit{false}$ implies that the common knowledge conditions  
      don't hold. 
      We thus need to show that $\I,(r,m) \models \neg (\init_i = 0 \lor 
      \bigvee_{j \in \Agents} K_i (\jdecided_j = 0))$ and $\I,(r,m) \models 
      K_i(\bigwedge_{j \in \Agents} \neg (\deciding_j = 0))$. 

      By definition, $\cond_0(i,m,G) = \mathit{false}$
      implies that for all $j \in \Agents$,
      $d(j,m-1,G) \ne 0$ or $l_G((j,m-1),(i,m)) \ne 1$, so
either $j$ did not decide 0 in round $m$ or 
$i$ did not receive a message from agent $j$ in round $m$. In either case,
$\I,(r,m) \models \neg K_i (\jdecided_j = 0)$, so
      $\I,(r,m) \models \neg (\init_i = 0 \lor 
      \bigvee_{j \in \Agents} K_i (\jdecided_j = 0))$.
      Finally, if $\cond_1(i,m,G) = \mathit{true}$,
      by Proposition~\ref{lemma:notdeciding1}.
      we can conclude that $\I,(r,m) \models 
      K_i(\bigwedge_{j \in \Agents} \neg (\deciding_j = 0))$, so
$(\kbp^1_i)^\I(r_i(m)) = \decide_i(1)$.

    \item If $P^{\opt}_i(r_i(m)) = \noop$ by the last line, we must have 
    $\decided_i = \bot$ in $r_i(m)$, and $\common_0(i,$ $m,G) = 
\common_1(i,m,G) = \cond_0(i,m,G)  =  \cond_1(i,m,G) = \mathit{false}$.
If $m = 0$, none of the conditions in $\kbp^1$ can hold
except $\I,(r,m) \models \init_i = 0 \lor \bigvee_{j \in \Agents} 
K_i (\jdecided_j = 0)$. However, since $\cond_0(i,0,G) = \mathit{false}$, 
we must have $\init_i \ne 0$ and $(\kbp^1_i)^\I(r_i(0)) = \noop$. 
If $m > 0$, then arguments above
show that the common knowledge conditions don't hold and 
      $\I,(r,m) \models \neg (\init_i = 0 \lor \bigvee_{j \in \Agents} 
K_i (\jdecided_j = 0))$. We thus need to show only that 
$\I,(r,m) \models \neg K_i(\bigwedge_{j \in \Agents} \neg 
      (\deciding_j = 0))$. Since the common knowledge conditions don't 
      hold, we can apply Proposition~\ref{lemma:notdeciding1} to conclude that this is 
      the case. 
      Thus, none of the conditions in $\kbp^1$ hold and we have 
      $(\kbp^1_i)^\I(r_i(m)) = \noop$.
  \end{itemize}
\end{proof}

We can then conclude that $P^{\opt}$ 
is also optimal with respect to full-information exchange.
Since each condition in $P^{\opt}$ can be checked in polynomial time in the 
size of the communication graph, and the communication graph itself uses 
$O(n^2t)$ bits, we then get the following 
result:
\repro{prop:kbp1-polynomial}
  There exists a polynomial-time implementation $P^{\opt}$ of $\kbp^1$ with 
  respect to a full-information exchange. 
\erepro

\commentout{
\repro{prop:cfaulty-necessary}
  If $P$ is an optimal protocol in an 
EBA context $\gamma$, and $\I_{P,\gamma} ,(r,m) \models
  \decided_i = \bot \land K_i(\ck{\N}(\faultyag \land \nodecided_\N(1) \land
   \exists 0))$, then all undecided agents in $\N(r)$ make a
    decision in round $m+1$, and 
    similarly if $\I_{P,\gamma} ,(r,m) \models \decided_i = \bot \land
    K_i(\ck{\N}(\faultyag \land \nodecided_\N(0) \land \exists 1))$.
\erepro

\repro{p:kbp1correct}
All implementations of $\kbp^1$ with
respect to $\gamma_{\fip,n,t}$ are EBA decision protocols for $\gamma_{\fip,n,t}$.
\erepro

\rethm{thm:kbp1opt} 
  If $\kbp^1$ is weakly safe with respect to $\gamma_{\fip,n,t}$
      then all implementations of $\kbp^1$
  are optimal with respect to $\gamma_{\fip,n,t}$.
\erethm

\rethm{thm:kbp1-safe} 
  $\kbp^1$ is weakly safe with respect to $\gamma_{\fip,n,t}$.
\erethm

\repro{prop:kbp1-polynomial}
  There exists a polynomial-time implementation 
  $P^{\opt}$ 
  of $\kbp^1$ with respect to a 
  full-information exchange. 
\erepro
} %

\subsection{Proof for Section~\ref{sec:cost}}
\repro{prop:decision-times} 
  If $r$ is a failure-free run, then
  \begin{enumerate}[(a)]
      \item If there is at least one agent with an initial preference of 0 in
          $r$, then all agents decide by round 2.
   \item If all agents have an initial preference of 1, then all agents decide 
     by round $t+2$ with $P^{min}$ and by round $2$ with $P^{\basic}$ and 
     $P^{\fip}$.      
  \end{enumerate}
\erepro
\begin{proof}
    For the first part, suppose that some nonfaulty agent has an initial preference
    of 0.  Clearly that agent decides 0 in the first round and tells all
  the other agents, who decide in the second round (for all three protocols).
  
  For the second part, suppose that all the agents are nonfaulty and
  have an initial preference 
    of 1.  Then with $P^{min}$, since no agent will decide 0 or hear about
  a decision of 0, the agents will wait for $t+1$ rounds of information
  exchange decide 1 in round $t+2$.  With $P^{\basic}$ and $P^{\fip}$, no agent will
  decide right away and  since all agents $i$ will get a message from every other
  agent $j$ in the first round from which they can conclude that $j$'s
  initial preference was 1 ($(\init,1)$ in the case of $P^{\basic}$ and
  an explicit message saying that $j$'s initial preference was 1 in the
  case of $P^{\fip}$), agents can all decide on 1 in round 2.
\end{proof}

\end{full}

\begin{full}
\begin{acks}
  Alpturer and Halpern were supported in part by AFOSR grant
FA23862114029.  Halpern was additionally supported in part by ARO grants
W911NF-19-1-0217 and W911NF-22-1-0061. 
The Commonwealth of Australia (represented by the Defence Science and Technology
Group) supported this research through a Defence Science Partnerships agreement.
We thank Yoram Moses for useful comments on the paper.
\end{acks}
\end{full} 

\begin{full}
\bibliographystyle{ACM-Reference-Format}
  
\bibliography{z,joe}
\end{full}

\end{document}